\documentclass{CSML}
\pdfoutput=1

\usepackage{lastpage}
\newcommand{\dOi}{14(2:1)2018}
\lmcsheading{}{1--\pageref{LastPage}}{}{}%
{Jan.~03, 2017}{Apr.~10, 2018}{}


\pdfoutput=1

\usepackage[utf8]{inputenc}
\usepackage[T1]{fontenc}
\usepackage{amsmath}
\usepackage{amsfonts}
\usepackage{amssymb}
\usepackage{mathtools}
\usepackage{enumitem}   
\usepackage{accents}    
\usepackage{stmaryrd}   
\usepackage{bm}
\usepackage{upgreek}
\usepackage{xspace}
\usepackage{xparse}
\usepackage{xypic}

\usepackage[english]{babel}

\usepackage{hyperref}
\hypersetup{hidelinks}
\usepackage[capitalise]{cleveref}
\crefname{thm}{Thm.}{Thm.}
\Crefname{thm}{Theorem}{Theorems}
\crefname{lemma}{Lem.}{Lem.}
\crefname{proposition}{Prop.}{Prop.}
\crefname{cfact}{Fact}{Fact}
\Crefname{cfact}{Fact}{Facts}
\crefname{equation}{Eqn.}{Eqn.}
\creflabelformat{equation}{#2(#1)#3}
\crefname{definition}{Def.}{Def.}

\let\originalleft\left
\let\originalright\right
\renewcommand{\left}{\mathopen{}\mathclose\bgroup\originalleft}
\renewcommand{\right}{\aftergroup\egroup\originalright}

\makeatletter
\def\namedlabel#1#2{\begingroup
	#2%
	\def\@currentlabel{#2}%
	\phantomsection\label{#1}\endgroup
}
\makeatother

\newcommand{\hfun}{\mathcal{H}}
\newcommand{\lfun}{\mathcal{L}}

\newcommand{\setTypes}[2]{ \mathcal{#1}\ifnotempty{#2}{^{(#2)}} }
\newcommand{\representation}[2]{ #1\ifnotempty{#2}{^{(#2)}} }

\newcommand{\sizedescriptor}[2]
{
	\ifthenelse{\equal{#1}{0}}{}{
	\ifthenelse{\equal{#1}{1}}{\big}{
	\ifthenelse{\equal{#1}{2}}{\Big}{
	\ifthenelse{\equal{#1}{3}}{\bigg}{
	\ifthenelse{\equal{#1}{4}}{\Bigg}{
	#2}}}}}
}
\newcommand{\st}[3][auto]{\sizedescriptor{#1}{\left}\{#2\;\sizedescriptor{#1}{\middle}|\;#3\sizedescriptor{#1}{\right}\}}

\newcommand{\some}[4][auto]{\exists\, #2 \,{\in}\, #3\,.\sizedescriptor{#1}{\left}({#4}\sizedescriptor{#1}{\right})}

\newcommand{\xall}[3]{\forall\, #1 \,{\in}\, #2\,.\,#3}
\newcommand{\xsome}[3]{\exists\, #1 \,{\in}\, #2\,.\,#3}

\newcommand{\xusome}[2]{\exists\, #1\,.\,#2}

\newcommand{\enc}[2][auto]{\sizedescriptor{#1}{\left}< #2 \sizedescriptor{#1}{\right}>}
\newcommand{\encdot}{\enc{\cdot,\cdot}}
\newcommand{\len}[1]{\ell(#1)}
\newcommand{\dH}{d_\mathsf{H}}

\newcommand{\ID}{\mathbb{D}}
\newcommand{\IN}{\mathbb{N}}

\newcommand{\IR}{\mathbb{R}}
\newcommand{\IZ}{\mathbb{Z}}

\newcommand{\id}{\mathrm{id}}

\newcommand{\parcol}{\colon\mathnormal\subseteq}
\newcommand{\dom}{\mathrm{dom}} 
\newcommand{\cod}{\mathrm{cod}} 
\newcommand{\img}{\mathrm{img}} 
\newcommand{\range}{\img}

\newcommand{\boundary}{\partial}
\newcommand{\inner}[1]{{#1}^\circ}

\newcommand{\Sast}{\Sigma^\ast}
\newcommand{\Cantor}{\Sigma^\omega}
\newcommand{\Baire}{\Sigma^{\ast\ast}}

\newcommand{\Reg}{\mathsf{LM}}

\newcommand{\bigO}{\mathcal{O}}

\newcommand{\const}{\mathrm{const}}
\newcommand{\co}{\mathsf{co}}
\newcommand{\PTime}{\mathsf{P}}
\newcommand{\FPTime}{\mathsf{FP}}
\newcommand{\UPTime}{\mathsf{UP}}
\newcommand{\NPTime}{\mathsf{NP}}

\newcommand{\modcont}{\overline{\mu}}
\newcommand{\modsu}{\underline{\mu}}

\newcommand{\unary}{\mathsf{un}}
\newcommand{\binary}{\mathsf{bin}}

\newcommand{\dyrep}[1][\empty]{ \representation{\binary}{#1}_{\ID} }
\newcommand{\dyadicrep}[1][\empty]{ \representation{\binary}{#1}_{\ID} }
\newcommand{\unatrep}[1][\empty]{ \representation{\unary}{#1}_\mathbb{N} }
\newcommand{\bnatrep}[1][\empty]{ \representation{\binary}{#1}_\mathbb{N} }
\newcommand{\binatrep}[1][\empty]{ \representation{\binary}{#1}_\mathbb{N} }
\newcommand{\bintrep}[1][\empty]{ \representation{\binary}{#1}_\mathbb{Z} }
\newcommand{\uintrep}[1][\empty]{ \representation{\unary}{#1}_\mathbb{Z} }
\newcommand{\binaryrep}[1][\empty]{ \representation{\bm{\mathsf{bin}}_\mathbb{N}}{#1} }
\newcommand{\unaryrep}[1][\empty]{ \representation{\bm{\mathsf{un}}_\mathbb{N}}{#1} }

\newcommand{\uintXrep}[1][\empty]{ \representation{\bm{\unary}}{#1}_\mathbb{Z} }
\newcommand{\sastrep}{\bm{\nu}_{\Sast}}
\newcommand{\realrep}[1][\empty]{ \bm{\rho}\ifnotempty{#1}{^{#1}} }

\newcommand{\distrep}[1][\empty]{ \representation{\bm{\delta}}{#1} }
\newcommand{\reldistrep}[1][\empty]{ \representation{%
	\bm{\delta}\kern.05ex{\scriptstyle\mathsf{rel}}\kern.1ex%
	}{#1} }
\newcommand{\setrep}[1][\empty]{ \representation{\bm{\psi}}{#1} }
\newcommand{\sisetrep}[1][\empty]{ \representation{ \widehat{\bm{\psi}} }{#1} }
\newcommand{\gridrep}[1][\empty]{ \representation{\bm{\kappa}}{#1} }
\newcommand{\wmemrep}[1][\empty]{ \representation{\bm{\omega}}{#1} }
\newcommand{\wopt}[1][\empty]{ \representation{\bm{\varpi}}{#1} }
\newcommand{\reptpl}[1][\empty]{ \representation{\bm{\xi}}{#1} }

\NewDocumentCommand{\funrep}{O{\empty} O{\empty}}{\bm{\lambda}^{#1}_{#2}}
\NewDocumentCommand{\parfunrep}{O{\empty} O{\empty}}{\bm{\lambda}^{#1}_{\subseteq}}
\NewDocumentCommand{\invrep}{O{\empty} O{\empty}}{\bm{\iota}^{#1}_{#2}}
\NewDocumentCommand{\parinvrep}{O{\empty} O{\empty}}{\bm{\iota}^{#1}_{\subseteq}}
\NewDocumentCommand{\setinvrep}{O{\empty} O{\empty}}{\bm{\theta}^{#1}}
\NewDocumentCommand{\imgfunrep}{O{\empty} O{\empty}}{\bm{\Lambda}^{#1}_{\subseteq}}

\newcommand{\norm}[2][\empty]{
   \ifthenelse{\equal{#1}{\empty}}{%
      \left\|#2\right\|
   }{%
      \left\|#2\right\|_{#1}
   }
}
\newcommand{\normdot}[1][\empty]{\norm[#1]{\cdot}}
\newcommand{\ndot}[1][\empty]{\normdot[#1]}
\newcommand{\wrtn}[2]{#1_{#2}}

\NewDocumentCommand{\vnorm}{O{auto} m O{\empty}}{%
	\sizedescriptor{#1}{\left}\|#2\sizedescriptor{#1}{\right}\|_{#3}%
}
\NewDocumentCommand{\abs}{O{auto} m}{%
	\sizedescriptor{#1}{\left}|#2\sizedescriptor{#1}{\right}|%
}

\newcommand{\repnorm}[2]{\wrtn{#1}{\normdot[#2]}}

\newcommand{\closedset}[1][\empty]{ \representation{\mathcal{A}}{#1} }
\newcommand{\clset}[1][\empty]{ \representation{\mathcal{A}}{#1} }
\newcommand{\compset}[1][\empty]{ \representation{\mathcal{K}}{#1} }
\newcommand{\regset}[1][\empty]{ \representation{\mathcal{R}}{#1} }
\newcommand{\convset}[1][\empty]{ \representation{\mathcal{C}}{#1} }

\newcommand{\ifnotempty}[2]{ \ifthenelse{ \equal{#1}{\empty} }{}{#2} }

\newcommand{\CR}[1][\empty]{\setTypes{CR}{#1}}
\newcommand{\KC}[1][\empty]{\setTypes{KC}{#1}}
\newcommand{\KR}[1][\empty]{\setTypes{KR}{#1}}
\newcommand{\KCR}[1][\empty]{\setTypes{KCR}{#1}}

\newcommand{\cb}{\CR}   
\newcommand{\bcb}{\KCR} 

\newcommand{\pleq}{\preceq_\mathrm{p}}

\newcommand{\pless}{\prec_\mathrm{p}}
\newcommand{\pequiv}{\equiv_\mathrm{p}}
\newcommand{\parampleq}{\preceq_{\mathrm{pp}}}

\newcommand{\ball}{\mathrm{B}}
\newcommand{\cls}[1]{\overline{#1}}
\newcommand{\cball}{\cls{\ball}}
\NewDocumentCommand{\clb}{D(){auto} O{\empty} m m}{%
	\cls{\ball}_{#2}
	\sizedescriptor{#1}{\left}(%
		#3,#4%
	\sizedescriptor{#1}{\right})%
}
\NewDocumentCommand{\opb}{D(){auto} O{\empty} m m}{%
	\ball_{#2}%
	\sizedescriptor{#1}{\left}(%
		#3,#4%
	\sizedescriptor{#1}{\right})%
}

\NewDocumentCommand{\dist}{O{\empty} m D(){\empty}}{%
	d_{\ifthenelse{\equal{#1}{\empty}}{}{#1,} #2}%
	\ifthenelse{\equal{#3}{\empty}}{}{(#3)}
}

\newcommand{\mto}{\rightrightarrows}
\newcommand{\mmapsto}{\Mapsto}
\newcommand{\shortto}{\!\!\to\!\!}

\newcommand{\dffn}{\colon}
\newcommand{\dfeq}{\coloneqq}
\newcommand{\dfeqrev}{=\mathrel{\mathop:}}

\newcommand{\tuple}[1]{\langle #1 \rangle}

\newcommand{\trsp}[1]{{#1}^\mathsf{T}}

\newcommand{\COMMENT}[1]{}
\newcommand{\secref}[1]{\S{#1}}

\newcommand{\ie}{\mbox{i.\,e.}\xspace}
\newcommand{\eg}{\mbox{e.\,g.}\xspace}
\newcommand{\wrt}{with respect to\xspace}

\newcommand{\enp}[1]{\sqcap \mathsf{#1}}
\newcommand{\ens}[1]{\mathsf{#1}}

\newcommand{\opname}[1]{\mathsf{#1}}
\newcommand{\dsoch}{\mathsf{Choice}}
\newcommand{\dsocap}{\mathsf{Intersection}}
\newcommand{\dsocup}{\mathsf{Union}}
\newcommand{\dsoproj}{\mathsf{Projection}}
\newcommand{\dsoinv}{\mathsf{Inversion}}
\newcommand{\dsoimg}{\mathsf{Image}}

\newcommand{\cfn}{\mathrm{C}}

\DeclareMathOperator{\lb}{lb}

\newcommand{\ul}[1]{\underline{#1}}
\newcommand{\ol}[1]{\overline{#1}}

\newcommand{\polar}[1]{ {#1}^{\bullet} }
\newcommand{\dpolar}[1]{ {#1}^{\bullet\bullet} }
\newcommand{\mc}[1]{\mathcal{#1}}

\newcommand{\eqnsp}{\;}
\newcommand{\eword}{\varepsilon}

\newcommand{\emdash}{\leavevmode\unskip\kern.2ex---\kern.2ex\ignorespaces}
\newcommand{\usubseteq}{\mathnormal{\subseteq}}

\begin{document}

\title[Closed Sets and Operators thereon]{Closed Sets and Operators thereon:\\%
	Representations, Computability and Complexity}

\author[C.~R\"osnick-Neugebauer]{Carsten R\"osnick-Neugebauer}	
\address{Technische Universit\"at Darmstadt, Germany\thanks{until March 2015}}	
\email{research@carstenrn.com}  
\thanks{%
	Supported by the \emph{German Research Foundation} (DFG) with project
	\texttt{Zi\,1009/4} and by the \emph{Marie Curie International Research
	Staff Exchange Scheme Fellowship} \texttt{294962} within the 7th European
	Community Framework Program.
	A preliminary version (extended abstract) have been appeared in Proc.~CCA'2013.
  Some parts have also been published in the author's PhD thesis \cite{crnphdthesis}.
}	

\begin{abstract}
	The TTE approach to Computable Analysis is the study of so-called
	representations (encodings for continuous objects such as reals, functions,
	and sets) with respect to the notions of computability they induce.
	A rich variety of such representations had been devised over the past
	decades, particularly regarding closed subsets of Euclidean space plus
	subclasses thereof (like compact subsets). In addition, they had been
	compared and classified with respect to both non-uniform computability of
	single sets and uniform computability of operators on sets.
	In this paper we refine these investigations from the point of view of
	computational complexity.
	Benefiting from the concept of second-order representations and complexity
	recently devised by Kawamura \& Cook (2012), we determine parameterized
	complexity bounds for operators such as union, intersection, projection,
	and more generally function image and inversion. By indicating natural
	parameters in addition to the output precision, we get a uniform view on
	results by Ko (1991-2013), Braverman (2004/05) and Zhao \& M\"uller (2008),
	relating these problems to the $\PTime$/$\UPTime$/$\NPTime$ question in
	discrete complexity theory.
\end{abstract}

\maketitle

\section{Introduction}
	\label{sec:intro}

Closed subsets of Euclidean space, and in particular subclasses thereof like
compact subsets, are important throughout many parts of pure theoretical
mathematics, but also of no less relevance in disciplines like numerical
analysis, convex optimization, or computational geometry. It is
necessary to first define encodings for sets in order to describe computations
on them which can be performed in a reasonably realistic computational model
(which can even be implemented and used in practice \cite{Mueller:iRRAM}).
We choose the function oracle Turing machine model as in
\cite{KF82,Ko91,KawamuraCook} with encodings (functions of form
$\phi \colon \Sast \to \Sast$) of continuous objects (reals, functions, sets)
are given as oracles.
Computational efficiency is gauged by second-order polynomial runtime bounds \cite{KawamuraCook}
-- with the explicit addition of parameters which leads to a second-order
equivalent of discrete parameterized complexity \cite{KMRZarXiv,Ret13}.

The introduction of such encodings for sets, called \emph{representations} in
the TTE-branch of Computable Analysis, 
constitutes the first of two parts of this paper.
One possible representation, $\distrep$, of a closed non-empty set
$S \subseteq \IR^d$ is by a function $\phi$ approximating its distance
function $d_S$ at any point up to arbitrary precision.
Another representation, $\setrep$, asserts,
given a point $q$ and a precision parameter $n$, that either
$q$ is of distance less than $2^{-n}$ to $S$, or that it is of distance
greater than $2 \cdot 2^{-n}$. Both representations allow for printing an
arbitrarily precise picture of the respective set. So are these two
representations computably equivalent, and if, how are they related
complexity-wise? While computably equivalent in any dimension $d$, they are
only polynomial-time equivalent in dimension $d = 1$.
From dimension $d = 2$ onward, the question of whether a set $S$ is polynomial-time
computable with respect to $\distrep$ iff it is \wrt $\setrep$, has been linked
to the $\PTime$ vs.~$\NPTime$ question \cite{Braverman04}.
Several more representations had been suggested
\cite{Weih87,Hertling02,RettingerHabil}
and compared with respect to their computable equivalence
\cite{KuSch95,BW99,Weih00,Ziegler02,Hertling02,BrattkaPresser}.
The complexities of these relations, and in particular the \emph{uniform}
formulations (\ie, the complexity of an operator translating between two
representations) of them, appear to be mostly unmentioned or unknown (except
for a few examples \cite{GLS88,Braverman04}). We strengthen these
previously known equivalence results from mere computable equivalence to
\emph{parameterized polynomial-time equivalence}, and prove \emph{uniform
exponential lower bounds} for the other relations.
For dimension $d = 2$ these uniform (non-)polynomial-time equivalence results
relate to complexity results for subsets of $\IR^2$ with respect to various
representations \cite{CK95:Rtwo,CK05:distance}; and they allow us to restate
complexity results like for Julia sets \cite{RW02:JuliaSets,Brav05:JuliaSets}
with respect to polynomial-time equivalent representations.

The second part of this work investigates the parameterized computational complexity
of natural operators on sets, such as binary intersection and union,
or projection to lower dimensional subspaces, but also forming
the image or local inverse of a function with respect to a given set.
The situation concerning their parameterized complexity is similar to
that for representations of sets: Operators on closed, compact or regular
subsets have been considered with respect to computability (\eg,
\cite{zhou1996computable,Ziegler04:Operators,ziegler2004linalg})
and non-uniform complexity bounds (\eg,
\cite[Chap.~4]{Ko91}, \cite{KoYu08}),
but it appears that less is known about the uniform complexity of operators
(exceptions include
\cite{ZM08}).
In addition, complexity bounds of \eg~projection
and function inversion had been linked to classical problems from discrete
complexity theory, namely $\PTime$ vs.~$\NPTime$ and $\PTime$ vs.~$\UPTime$.
Results like these are in the spirit of well-known ones for maximization and
integration of functions \cite{KF82,Fried84} (we refer the reader \eg to
\cite{Ko98:Survey,BHW08:Tutorial} for an overview and more examples of
this kind).
In this paper we present uniform worst-case parameterized complexity bounds
for all of the aforementioned operators. Providing operators through
parameters with more information about their arguments turns out to be
valuable and fruitful approach to achieve uniformity and also allows for a
fine-grained perspective on their inherent complexity.\footnote{%
	Take the projection of a subset of $\IR^d$ to its first component as an
	example:
	The question whether a given point $q$ is contained in the projection is
	uncomputable as long as no information about a bound of the set is given.}
In addition to upper bounds we also present exponential-time lower bounds,
thus extending upon the former non-uniform bounds that depended on
believed-to-be-hard problems from discrete complexity theory.

\subsection*{Results obtained in this paper}

Primarily based on \cite[Chap.~4]{GLS88}, \cite[\secref{5}]{Weih00},
\cite{Ziegler02} and \cite{RettingerHabil}, we introduce five representations
in \cref{sec:representations}: $\distrep$ and $\reldistrep$, $\setrep$,
$\gridrep$, and $\wmemrep$. Each representation will depend on a fixed yet
arbitrary norm \emdash a dependence we will show in
\cref{subsec:poly-norm-invariance} to be of ``polynomial-time irrelevance''
for all but one representation ($\distrep$:
\cref{thm:distrep-not-polytime-norm-invariant}). We furthermore compare
relations between representations by means of (parameterized) polynomial-time
translations; and observe that, although they are all (uniformly)
computably equivalent over \emph{convex regular} sets
\cite[Cor.~4.13]{Ziegler02}, mutual (parameterized) polynomial-time reductions exist
only for intervals (\cref{lem:polytime-relations-dim-one}), but in general not from
dimension $d = 2$ onward.
In fact, we identify a kind of hierarchy:
$\wmemrep$ forms the \emph{poorest}, $\distrep$ the \emph{richest}
representation, and $\reldistrep$, $\setrep$, $\gridrep$ are parameterized
polynomial-time equivalent over compact sets
(\cref{s:Braverman-distrep,thm:setrep-gridrep-equiv,thm:enrichments-wmem}).
Parameterized polynomial-time equivalence of $\wmemrep$ with $\distrep$
(hence of all of the former representations) is finally achieved if
restricted to compact convex regular sets (\cref{s:convex-opt}).

\cref{sec:operators} then uses the formerly unveiled connections between
representations by discussing the complexity of operators. Operators include
$\dsoch$ (finding \emph{some} point in a set; the presumably most basic set
operation) and $\dsocup$, which are fully polynomial-time computable
for all of the above representations but $\wmemrep$); and $\dsocap$, which
(in contrast to the former result) polynomial-time computable \emph{only} for $\wmemrep$.
More involved operators, $\dsoinv$ and $\dsoimg$, are discussed in
\cref{sec:dsoinv,sec:dsoimg}: We prove that $\dsoinv$ is parameterized
polynomial-time computable for Lipschitz-continuous functions whose inverse
is also Lipschitz continuous (\cref{s:inv-lip}) \emdash which
fits right in the gap between naive exponential-time algorithms and results
of Ko \cite[Thm.~4.23+4.26]{Ko91}, the latter relating \emph{non-uniform}
$\dsoinv$ for a more general class of functions to the (considered to be hard)
question whether $\PTime \neq \UPTime$ holds true.

\subsection*{Preliminaries, nomenclature}

We introduce some notations and concepts we frequently use throughout this
paper.
Let $\Sigma$ be the binary alphabet $\{0,1\}$, $\Sast$ denotes the set of
finite $0/1$-strings, $\Cantor$ the set of $0/1$ \emph{sequences}
(isomorphic to $\Sigma^{(\Sast)}$, the set of all total functions from
$\Sast$ to $\Sigma$), and $\Baire \dfeq (\Sast)^{(\Sast)}$ the set of all
total functions from $\Sast$ to $\Sast$.
A finite string $s = s_1 \cdots s_k \in \Sast$ with $s_i \in \Sigma$ is also
called \emph{word}.
The \emph{length} of $s$ (as above) is defined by $\len{s} \dfeq k$, and
$\eword$ denotes the unique word of length $0$ (the \emph{empty word}).
We further consider the sets $\Cantor$ and $\Baire$ as topological spaces
equipped with the product topology (equip $\Sigma$, $\Sast$ with the
discrete topology).

Let $\IN \dfeq \{0\} \cup \IN_+$ with $\IN_+ \dfeq \{1,2,3,\dots\}$,
and abbreviate the \emph{binary logarithm of $x$} by $\lb x \dfeq \log_2 x$.
By $\ID_m \dfeq \st{a/2^m}{a \in \IZ}$ we denote the set of \emph{dyadic
rationals of precision $m \in \IZ$}, and set
$\ID \dfeq \bigcup_{m \in \IZ} \ID_m$.
Let further $\bnatrep \dffn \IN \to \Sast$ denote the usual \emph{binary
coding} of naturals as words, $\unatrep$ denotes the \emph{unary coding}
which we usually abbreviate by $0^n \dfeq \unatrep(n)$.
Even though $\unatrep$ is not surjective and thus does not admit an inverse,
we like to understand by $\unatrep^{-1}$ the mapping
$s \in \Sast \mapsto \len{s} \in \IN$.
Note that the notions of binary and unary encodings naturally extend to the
set of integers by embedding $\IZ$ into $\IN$.
By abuse of notations, we casually write $0^{k}$ for the unary coding of
an \emph{integer} $k$.
Pairing functions (usually total, bijective, computable and invertible in
polynomial time, although we do not need them to be surjective) are denoted
by $\encdot_X \colon X \times X \to X$ whereby the ``${}_X$'' will be usually
clear from context (typically $\IN$, $\Sast$ or $\Baire$) and henceforth
omitted.
Explicitly, $\enc{s,t}_{\Sast} \dfeq
	\bnatrep \big( \enc[1]{\bnatrep^{-1}(s), \bnatrep^{-1}(t)}_{\IN} \big)$
with $\encdot_{\IN}$ being the \emph{Cantor pairing function}.
Further define the pairing function $\enc{\phi,\psi}$ on Baire space
$\Baire$ through
$\enc{\phi,\psi}(\eword) \dfeq \eword$,
$\enc{\phi,\psi}(0\,s) \dfeq 0^{\len{\psi(s)}}\,1\,\phi(s)$, and
$\enc{\phi,\psi}(1\,s) \dfeq 0^{\len{\phi(s)}}\,1\,\psi(s)$
for all $s \in \Sast$.
The binary encoding $\dyrep[d] \dffn \ID^d \to \Sast$ of dyadic rationals
is recursively defined:
Let $\dyrep[1] \dffn a/2^{n} \mapsto \enc[1]{\bintrep(a),0^n}_{\Sast}$
and for $d \geq 2$ let $\dyrep[d] \dffn (q_1,\dots,q_d)
\mapsto \enc[1]{ \dyrep[1](q_1), \dyrep[d-1](q_2,\dots,q_d) }_{\Sast}$.

A \emph{normed (vector) space} is a pair $(X,\normdot)$ of a vector space $X$
together with a norm $\normdot$ on $X$.
A set $S \subseteq X$ is \emph{open in $X$} if it is the set of its
inner points, \ie, $\inner{S} = S$, and closed (in $X$) if it is the closure
of itself, \ie, $\cls{S} = S$. The boundary is defined as
$\boundary{S} \dfeq S \setminus \inner{S}$.

On $(\IR^d, \normdot)$ we denote \emph{closed} balls with center $x \in \IR^d$
and radius $\delta > 0$ by $\clb[\normdot]{x}{\delta} \dfeq
\st[0]{y \in \IR^d}{\norm{x-y} \leq \delta}$ \emdash or simply by
$\clb{x}{\delta}$ if the norm used is understood.
Similarly denote \emph{open} balls as $\opb[\normdot]{x}{\delta}$.
Whenever useful, we use the abbreviation
$\ID^d_n(R) \dfeq \ID^d_n \cap \cball(0,R)$.
A ``ball'' (actually a neighborhood) around a set $S \subseteq \IR^d$ of
radius $\delta > 0$ is defined through the union of balls around the points
of $S$, \ie, $\cball(S,\delta) \dfeq \bigcup_{x \in S} \cball(x,\delta)$,
and similarly for open balls.
The same works in the reverse direction, \ie, for negative radii:
Denote by $\cball(S,-\delta) \dfeq
	\st{x \in \IR^{d}}{\cball(x,\delta) \subseteq S}$
the (possibly empty) set of all points $x$ lying \emph{$\delta$-deep in $S$}.
Further define \emph{hollow} closed balls centered at $x$ with inner radius
$\delta' > 0$ and outer radius $\delta \geq \delta'$ through
$\cball_{\normdot}(x,\delta,\delta') \dfeq
\clb[\normdot]{x}{\delta} \setminus \opb[\normdot]{x}{\delta'}$.

The domain and co-domain (also: image) of a function $f$ mapping from a set
$X$ into $Y$ are denoted by $\dom(f)$ and $\cod(f)$, respectively. Besides
total functions, $f \colon X \to Y$ with $\dom(f) = X$, we also consider
partial functions $f \parcol X \to Y$ ($f$ is defined only on a subset of $X$,
thus the ``$\usubseteq X$''), and partial multi-valued functions as
$f \dffn X \mto Y$, $f(x) \subseteq Y$.
Phrased differently, a multi-valued partial function $f \dffn X \mto Y$ is
a partial function from $X$ into the powerset of $Y$.
The multi-valued assignment of an element $x \in X$ to a subset
$Y' \subseteq Y$ is abbreviated by $x \mmapsto Y'$.



\section{Model, representations and complexity} \label{sec:model}


All the concepts we discuss in this section are guided by the question how
computations on subsets of real vector spaces could be carried out on a
reasonably realistic machine model. Using a \emph{Turing-like machine model}
(\cref{sec:models}), we define encodings of sets through \emph{representations}
(\cref{sec:representations}) and proceed by giving definitions for
\emph{computability} and \emph{parameterized complexity}
(\cref{subsec:enrichments,sec:complexity}).

\subsection{Computational models}
	\label{sec:models}

Two mainstream models in Computable Analysis are the \emph{Type-2 Theory
of Effectivity} \cite{Weih00} and the \emph{oracle Turing machine model}
\cite{KF82,Ko91}. The former even yields a topological interpretation of
computability. 
Therefore, we start
by introducing computability using the former model.
Carrying these notation over to the latter model will then naturally give rise
to a suitable notion of complexity.
From a computational point, however, both models are equivalent.
The notions discussed in this section are based on
\cite[\secref{2.1+2.2}]{KawamuraPhD} and \cite[\secref{2}]{KMRZarXiv}.

\subsubsection{Type-2 machines: computations on finite and infinite strings.}


Type-2 machines extend upon classical Turing machines by operating on
\emph{infinite} strings instead of finite ones, \ie, on $\Cantor$
instead of $\Sast$. Such a machine consists of finitely many left-to-right
readable input and bidirectional working tapes, and one left-to-right writable
output tape.
A \emph{computation} on finite strings is carried out as on classical Turing
machines: Given a type-2 machine $M$ with $k \in \IN$ input tapes
plus an input $(s_1,\dots,s_k) \in \Sast \times \dots \times \Sast$ to it,
$M$ either reads one symbol from one of its $k$ input tapes, reads or
writes one symbol on one of its working tapes, or writes a symbol on its output
tape. The same applies if the input is not a tuple of finite strings, but of
infinite ones from $\Cantor \times \dots \times \Cantor$.
Given such a machine $M$, we say it \emph{computes a partial function}
$f \parcol \Sast \times \dots \times \Sast$ if it terminates on all inputs
$(s_1,\dots,s_k)$ from $f$'s domain and writes $f(s_1,\dots,s_k)$
symbol-by-symbol on the output tape.
As for the infinite case, $M$ computes a partial function
$f \parcol \Cantor \times \dots \times \Cantor \to \Cantor$ if $M$ continues
forever on input $(s_1,\dots,s_k)$ whilst producing $s \dfeq f(s_1,\dots,s_k)$
on its output tape.


Machines which have to run infinitely long to produce their answer would
certainly not deserve to be called ``practicable''. The key here, however, is
that for every \emph{finite prefix} of the input read a type-2 machine
produces a \emph{non-revisable finite prefix} of the still infinitely long
correct output. We postpone discussing of the strong topological implications
to computability until \cref{s:tte-cont} and \cref{sec:proof-techniques}.


As for classical Turing machines, type-2 machines are also capable to compute
functions $f \parcol X_1 \times \dots \times X_k \to X'$ for sets $X_i, X'$
different from $\Cantor$ by encoding their elements through items from
$\Cantor$ (appropriate cardinalities assumed).
Following \cite[Def.~2.3.1(2)]{Weih00}, we call such encodings
\emph{representations}. Building upon them, we formulate the computability of
functions $f \parcol X \to X'$ through \emph{realizers}.

\begin{defi}[representations, realizers] \label{def:t2m-realizer}
	Let $X$ and $X'$ be sets.
	\begin{enumerate}
	\item A \emph{representation of $X$} is a surjective partial function
		$\reptpl \parcol \Cantor \to X$.
	\item An element $\phi \in \reptpl^{-1}[\{x\}]$ is said to be a
		\emph{$\reptpl$-name} of $x$.
	\end{enumerate}
	Further assume $\reptpl$ and $\reptpl'$ to be representations of $X$ and
	$X'$, respectively, and let $f \parcol X \to X'$ be some function.
	\begin{enumerate}[resume]
		\item A function $g \parcol \Cantor \to \Cantor$ is called a
			\emph{$(\reptpl,\reptpl')$-realizer} of function $f$ if for all
			$\phi \in \reptpl^{-1}\big[ \dom(f) \big]$,
			$(\reptpl' \circ g)(\phi) = (f \circ \reptpl)(\phi)$ holds true.
			A more visual way to think of realizers is by a commuting diagram:
			\begin{equation*}
				\xymatrix@R+.2em@C+.2em{
					\Cantor \ar[r]^g \ar[d]^{ \reptpl } &
					\Cantor \ar[d]^{ \reptpl' } \\
					X \ar[r]^f &
					X'
				}
			\end{equation*}
		\item Function $f$ is called \emph{$(\reptpl,\reptpl')$-computable}
			(-\emph{continuous}) if it has a computable (continuous)
			$(\reptpl,\reptpl')$-realizer.
	\end{enumerate}
\end{defi}

The above notions 
are reasonable in the
sense that topological continuity of a function corresponds to
$(\reptpl,\reptpl')$-continuity if $\reptpl$ and $\reptpl'$ are both
\emph{admissible} (cf.~\cite[\secref{3.2}]{Weih00}, \cite{Schroeder02}).
Note that all representations mentioned in this paper are admissible.

We present a few examples of representations in form of a definition.

\newpage
\begin{defi}[representations]                           \label{def:repr}
	\begin{enumerate}
		\item Representations $\unaryrep$ and $\binaryrep$ extend upon the
			unary and binary encodings, $\unatrep^{-1}$ and $\binatrep^{-1}$
			from $\Sast$ to $\Cantor$.\footnote{%
				To justify notational between $\bnatrep^{-1}$ and $\binaryrep$:
				Computations are performed on the level of names (\ie, objects
				from $\Cantor$).
				Objects like natural numbers or dyadic rationals, on the contrary,
				are usually used ``as the are'', \ie, not encoded as words or
				sequences.
				They are encoded back into words (via $\bnatrep$ or
				$\unatrep$) not before the end of the respective argument.}
			We say, a natural number $k \in \IN$ is represented by a
			$\binaryrep$-name
			$\phi = 0\,b_0\,0\,b_1 \dots 0\,b_\ell\,1^\omega \in \Cantor$,
			$\binaryrep \parcol \Cantor \to \IN$, if $\phi$ essentially is $k$'s
			binary encoding: $\binaryrep(\phi) = \sum_{i=0}^\ell b_i 2^{i} = k$.
			Its unary counterpart, denoted $\unaryrep$, can be obtained through
			representing a natural number $k \in \IN$ by
			$\phi \dfeq (0\,1)^k\,1^\omega \in \Cantor$.
		\item \label{def:realrep}
			We define a $\realrep$-name $\phi$ of a real number $x \in \IR$ to
			be a suitably encoded sequence $(q_n)_n$ of dyadic rationals
			$q_n \in \ID_n = \st{a/2^n}{a \in \IZ}$ (\ie,
			$\phi = \left< (q_n)_{n\in\IN} \right> \in \Cantor$;
			cf.~\cite[Def.~4.1.5+4.1.17]{Weih00}) converging to $x$ in the sense
			that $|q_n - x| \leq 2^{-n}$ holds true for all $n \in \IN$.
		\item Based on $\realrep$, a representation
			$[\realrep \!\to\! \realrep] \parcol \Cantor \to \cfn(\IR)$ of
			continuous functions $f \colon \IR \to \IR$ may intuitively be
			understood as follows: A $[\realrep\shortto\realrep]$-name encodes
			how ($\realrep$-names of) $x \in \IR$ are translated into
			$\realrep$-names of $f(x)$ (cf.~\cite[Def.~3.3.13]{Weih00}
			and \cite{Grzegorczyk57}).
	\end{enumerate}
\end{defi}

An important property of the TTE model and its
representations is due to its concise topological roots, resulting in the
\emph{Main Theorem} in the TTE-branch of Computable Analysis.

\begin{fact} \label{s:tte-cont}
	Computability implies (topological) continuity.
\end{fact}

Recall that $(\reptpl,\reptpl')$-computability by a Type-2 machine $M$ means
that $M$ maps finite prefixes of a $\reptpl$-name $\phi$ to finite prefixes of
a $\reptpl'$-name $\phi'$.
The reader is referred to \cite[Thm.~2.2.3+3.2.11]{Weih00} for detailed
explanations and proofs.

\subsubsection{Oracle machines.}

The type-2 model, and in particular the way we have introduced representations
so far, does not yield a viable notion of complexity:
Say $\phi$ is a $\realrep$-name of a real number $x \in \IR$ as
defined in \cref{def:repr}(\labelcref{def:realrep}), \ie, an encoded sequence
$(q_n)_{n \in \IN}$ of dyadic rationals.
In order to access a specific element encoded through $\phi$,
say $q_N$, a type-2 machine has first to skip over a possibly large (compared
to the coding length of $q_N$) prefix of $\phi$. Such an initial motion has to
reflect in some way in any complexity notion, although the search for $q_N$
does not contribute anything to the \emph{actual} computation on it.
Granting a machine access to individual information encoded through $\phi$
(black-box approach) without charging too much for such access can be realized
by \emph{oracle Turing machines} (oracle machines, or OTMs, for short).

Recall that an oracle machine is a classical (possibly multi-tape) Turing
machine with the addition of a special query tape and two new states: One to
initiate the query to the oracle with the content of the question written on
the query tape,
and a second to mark that the oracle has written its respective answer on
the query tape. The oracle attached to a machine can either be a subset of
$\Sast$ (a possibly undecidable decision problem), or a string function.
We choose the latter type, a \emph{function-oracle} machine model
(cf.~\cite[Def.~2.11]{Ko91}).

\begin{defi}[second-order representations]\ %
	\label{def:second-order-rep}
	\begin{enumerate}
	\item A \emph{second-order representation} $\reptpl$ of a set $X$
		is a partial surjective function $\reptpl \parcol \Baire \to X$.
	\item Any ordinary representation $\reptpl \parcol \Cantor \to X$
		(\ie, in the sense of \cref{def:t2m-realizer})
		induces a second-order representation $\tilde{\reptpl}$:
		Any $\reptpl$-name $\phi = (b_i)_i$, $b_i \in \Sigma$,
		yields a $\tilde{\reptpl}$-name $\tilde{\phi}$ through
		$\tilde{\phi}(s) \dfeq b_{\len{s}}$ for any $t \in \Sast$.%
		\footnote{The reader is referred to \cite[Def.~1.16]{KMRZarXiv} for
			more (formal) definitions and extensions of second-order
			representations.}
		\item Second-order representations $\tilde{\reptpl}_1$ and
			$\tilde{\reptpl}_2$ of $X_1$ and $X_2$, respectively, induce a
			second-order representation
			$\tilde{\reptpl}_1 \times \tilde{\reptpl}_2$
			of $X_1 \times X_2$:
			If $\phi_i$ is a $\tilde{\reptpl}_i$-name of $x_i \in X_i$, then
			$\enc{\phi_1,\phi_2}_{\Baire}$ is a $\tilde{\reptpl}_1 \times \tilde{\reptpl}_2$-name of $X_1 \times X_2$.
	\end{enumerate}
\end{defi}

Functions computable by oracle machines can be defined over realizers similar
to \cref{def:t2m-realizer}.

\begin{defi}[computable functions, realizers]
	\label{def:otm-realizer}
	Assume $\reptpl$ and $\reptpl'$ to be second-order representations of $X$
	and $X'$, respectively, and let $f \parcol X \to X'$ be some function.
	\begin{enumerate}
		\item A function $g \parcol \Baire \times \Sast \to \Sast$ is computable
			by an oracle machine $M^?$ if for all
			$(\phi,s) \in \Baire \times \Sast$,
			$M^\phi$ started with $s$ halts and writes $g(\phi,s)$ on its
			output tape.
		\item A function $g \parcol \Baire \times \Sast \to \Sast$ is called a
			$(\reptpl,\reptpl')$-realizer of function $f$ if for all
			$\phi \in \reptpl^{-1}\big[ \dom(f) \big]$,
			$(f \circ \reptpl)(\phi) = \reptpl'\big( g(\phi,\cdot) \big)$
			holds true. (Note that $g(\phi,\cdot) \in \Baire$.)
		\item Function $f$ is called $(\reptpl,\reptpl')$-computable if it has
			a $(\reptpl,\reptpl')$-realizer computable by some oracle machine.
	\end{enumerate}
\end{defi}

\subsubsection{Relation between both models.}
	\label{sec:ttm-otm-relation}

Although the type-2 machines on one hand side and oracle machines on the
other are seemingly different approaches to \emph{Real Computability}, they
are actually computably identical.

\begin{fact}
	Let $\reptpl$ and $\reptpl'$ be ordinary representations of $X$ and $X'$,
	respectively. Every $(\reptpl,\reptpl')$-computable (-continuous) function
	(\ie, realized by some type-2 machine computable function)
	$f \parcol X \to Y$ is also $(\tilde{\reptpl},\tilde{\reptpl}')$-computable
	(-continuous) (\ie, realized by some oracle machine computable function);
	and vice versa.
\end{fact}

This follows by type conversion \cite[Lem.~2.1.6]{Weih00}: Since $f$ is
$(\reptpl,\reptpl')$-computable (-continuous), it has some computable
(continuous) realizer $g \parcol \Cantor \to \Cantor$. By the aforementioned
Lemma, a function $G \parcol \Cantor \times \Sast \to \Sast$,
$(\reptpl,\reptpl')$-computable by some type-2 machine, exists such that the
following hold true:
\begin{enumerate}
	\item Function $G$ has a suitable domain: for all $\phi \in \Cantor$,
		$\phi \in \dom(g)$ if $\xall{s}{\Sast} (\phi,s) \in \dom(G)$;
	\item and $G$ does behave like $g$ (extensionally):
		$\xall{\phi}{\dom(g)} \xall{s}{\Sast} G(\phi,s) = g(\phi)(s)$.
\end{enumerate}
Since $G$ is computable by some type-2 machine, it is computable by some
oracle machine as well, thus $(\tilde{\reptpl},\tilde{\reptpl}')$-realizing
$f$. The reverse direction follows similarly when used that representations
in TTE can equivalently written stated over $\Baire$ instead of $\Cantor$
\cite[Ex.~3.2(17)]{Weih00}.

\emph{Convention.}
From now on we omit the ``tilde'' and simply write ``$\reptpl$'' whenever we
reason about second-order representations $\tilde{\reptpl}$. Arguments
$s \in \Sast$ of names $\phi$ are usually of form $s = \enc{q,0^n}$ for a
dyadic point $q \in \ID^d$ and a precision parameter $n \in \IN$. We use
$\phi(q,0^n)$ as a shorthand for the correct but more verbose
$\phi \big( \tuple{\dyadicrep[d](q),\unatrep(n)} \big)$.


\subsection{Second-order representations of sets}
	\label{sec:representations}

Throughout this paper, we solely concentrate on closed \emph{non-empty}
subsets of $\IR^d$ (for various $d$) and subclasses thereof.
More precisely:
In any dimension $d \in \IN$ we denote the class of \emph{closed non-empty
subsets of $\IR^d$} by $\closedset[d]$, the class of \emph{compact subsets} by
$\compset[d] \dfeq \st[0]{S \in \clset[d]}{
	\xusome{\delta > 0} S \subseteq \clb{0}{\delta}}$,
\emph{convex subsets} are in
$\convset[d] \dfeq \st[0]{S \in \clset[d] \cap (\IR^d,\ndot[2])}{
	\xall{x,y}{S} \xall{\lambda}{[0,1]} \lambda x + (1-\lambda)y \in S}$, and
\emph{regular subsets} in
$\regset[d] \dfeq \st[0]{S \subseteq \clset[d]}{\ol{\inner{S}} = S}$.\footnote{%
	Normed vector spaces over equivalent norms are homeomorphic, thus imply
	the same topology.
	It therefore is not necessary to tie any of these subclasses (except for
	$\convset$) to a concrete norm.
}
All notions are depicted on \cref{fig:kinds-of-sets}.

\begin{figure}[htb]
	\centering
	\includegraphics[width=\textwidth]{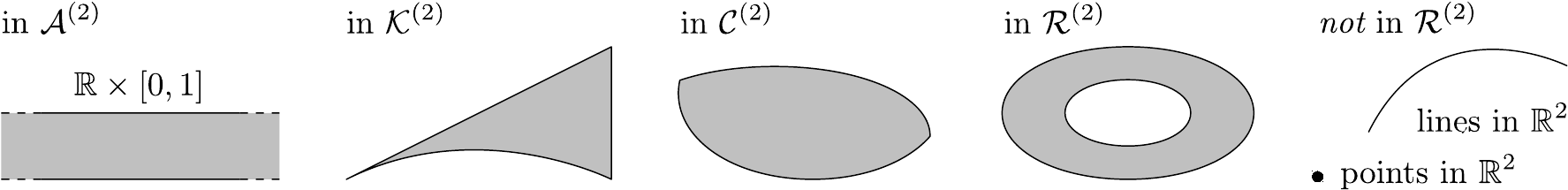}
	\caption{Classes of subsets of Euclidean space:
		Closed $\closedset$,
		compact $\compset$,
		convex $\convset$, and regular $\regset$.}
	\label{fig:kinds-of-sets}
\end{figure}

Intersections of the above subclasses will also be of interest, \eg,
the class $\KR[d] \dfeq \compset[d] \cap \regset[d]$ of \emph{bounded bodies};
understand $\CR[d]$ (\emph{convex bodies}), $\KC[d]$, and $\KCR[d]$
(\emph{bounded convex bodies}) similarly.
Omitting the dimension on any class of subsets denotes the union over all $d$,
\ie, $\clset \dfeq \bigcup_{d \in \IN} \clset[d]$ and so forth.

For a set $S \in \closedset[d]$ and a norm $\normdot$ we define
the distance function $\dist[\normdot]{S} \colon \IR^d \to \IR_{\geq 0}$,
mapping any point $q \in \IR^d$ to its minimal distance to set $S$,
by $\dist[\normdot]{S}(q) \dfeq \min_{x \in S} \norm{q - x}$.

Every representation for a class of sets provides approximate
information to the specific set $S$ it encodes in terms of answers to a
specific type of questions:
Given a point $x$, is $x \in S$?
If not, is $x$ far from $S$? How far?
From \cref{s:tte-cont} we can infer that only trivial sets $S$ (\ie,
$S = \emptyset$ or the whole space, $S = \IR^d$) are representable by their
characteristic functions since they are discontinuous in all other cases.
We thus have to allow any name $\phi$ of a reasonable
representation to be in some sense vague or "fuzzy"
when queries are close to the boundary of the set $S$ it encodes.
To be more precise, we have to allow any name $\phi$ to make errors somewhere if
$\phi$ represents $S \in \closedset[d]$ with $\emptyset \subset S \subset \IR^d$.
This error, however, has to be controllable through a precision parameter $n$,
just like for reals and functions.


Subsequently, we cite five different definitions (visualized in
\cref{fig:set-representations}) for representations of sets.

\begin{figure}[htb]
	\centering
	\includegraphics[width=\textwidth]{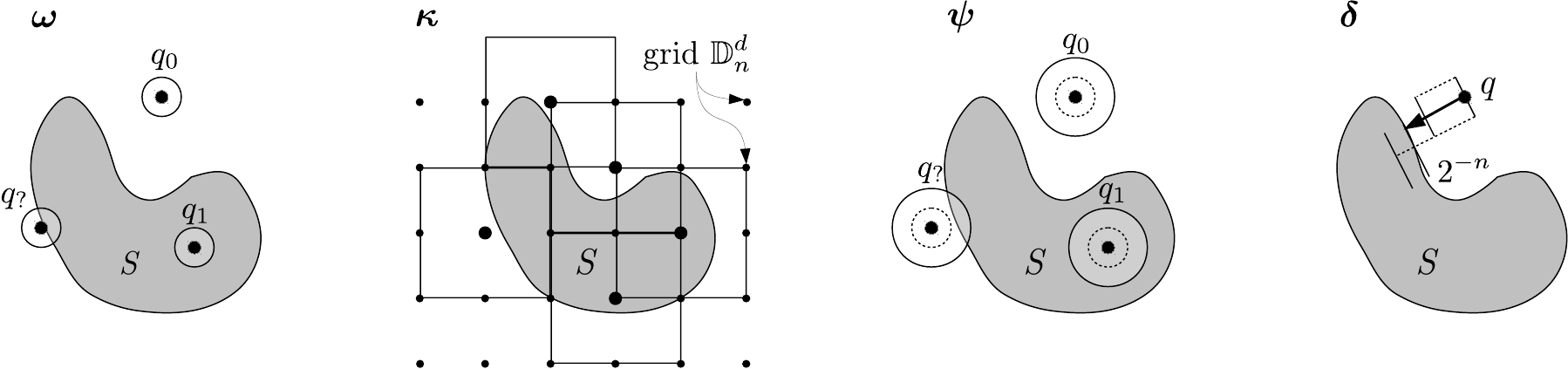}
	\caption{Representations of sets: Illustration of \cref{def:set-representations}}
	\label{fig:set-representations}
\end{figure}

\begin{defi} \label{def:set-representations}
	Fix a dimension $d \in \IN$ and a norm $\normdot$ on $\IR^d$. Points $q$
	are chosen from $\ID^d$, and \emph{precision parameters} are denoted by
	$n \in \IN$.
	\begin{enumerate}
		\item	\emph{Weak-membership representation:}
			A $\wrtn{\wmemrep[d]}{\ndot}$-name $\phi$ of $S \in \regset[d]$
			satisfies
			\begin{enumerate}
				\item $\phi(q,0^n) = 1$ if $\cball_{\normdot}(q,2^{-n}) \subseteq S$
					(\ie, $q$ lies $2^{-n}$-deep within $S$), or
				\item $\phi(q,0^n) = 0$ if $\cball_{\normdot}(q,2^{-n}) \cap S = \emptyset$
					(\ie, $q$ is off by more than $2^{-n}$).
			\end{enumerate}
	\item A $\wrtn{\gridrep[d]}{\normdot}$-name $\phi$ of $S \in \compset[d]$
		satisfies
		\begin{enumerate}
			\item $b \dfeq \uintrep^{-1}(\phi(\eword))$ is an upper bound on the
				size of $S$, \ie, $S \subseteq \cball_{\ndot}(0,2^{b})$, and
			\item $\phi$ encodes a sequence of sets $B_n \subset \ID^d_n$ through
				$\phi(q,0^n) = \chi_{B_n}(q)$ that is $2^{-n}$-close to $S$ in the
				Hausdorff-distance $\dH$, \ie,
				$\dH\left( B_n, S \right) \leq 2^{-n}$:
				\begin{enumerate}
					\item[\namedlabel{eqn:kappa-x-d}{($\kappa1$)}]
						$\xall{n}{\IN} \xall{x}{S} \xsome{q}{B_n}
						\vnorm{q-x} \leq 2^{-n}$, and
					\item[\namedlabel{eqn:kappa-d-x}{($\kappa2$)}]
						$\xall{n}{\IN} \xall{q}{B_n} \xsome{x}{S}
						\vnorm{q-x} \leq 2^{-n}$.
				\end{enumerate}
		 \end{enumerate}
	\item A $\wrtn{\setrep[d]}{\normdot}$-name $\phi$ of $S$ satisfies
		\begin{enumerate}
			\item $\phi(q,0^n) = 1$ if
				$\ball_{\normdot}(q, 2^{-n}) \cap S \neq \emptyset$
				($q$ is $2^{-n}$-close to $S$), or
			\item $\phi(q,0^n) = 0$ if $\cball_{\normdot}(q, 2^{-n+1}) \cap S = \emptyset$
				($q$ is at least $2^{-n+1}$-far off of $S$).
		\end{enumerate}
	\item A $\phi \in \Baire$ is a $\wrtn{\distrep[d]}{\normdot}$-name of
		$S \in \clset[d]$
		whenever $\phi$ encodes $S$' distance function, \ie,
		$\abs[0]{ \phi(q,0^n) - \dist[\normdot]{S}(q) } \leq 2^{-n}$.
	\item Relative version of $\distrep[d]$ (specialization of
		\cite[Def.~1.27]{RettingerHabil}):
		A $\phi \in \Baire$ is a $\wrtn{\reldistrep[d]}{\normdot}$-name of
		$S \in \clset[d]$ if for all $q \in \ID^d$ and $n \in \IN$,
		$\phi$ satisfies
		\begin{align} \label{eqn:def-reldistrep}
			3/4 \cdot \dist[\normdot]{S}(q) - 2^{-n}
			\leq \phi(q,0^n)
			\leq 5/4 \cdot \dist[\normdot]{S}(q) + 2^{-n}.
		\end{align}
	\end{enumerate}
\end{defi}

\subsubsection{A few historical remarks.}

The concept underlying $\distrep$ (representation of the distance function;
cf.~\cite[$\bm{\psi}^\mathrm{dist}$: Def.~5.1.6]{Weih00}) is the same as for
\emph{Turing located sets} (\cite{ge1994extreme}; dating even back to Brouwer
\cite[\emph{katalogisierte Mengen}]{Brouwer19}),
and the concept of \emph{recognizable sets} \cite[Def.~3.5]{CK95:Rtwo} is
underlying the \emph{weak membership problem/representation} $\wmemrep$
(\cite[Def.~2.1.14]{GLS88}; also in \cite[Def.~4.2]{KuSch95}).
A strengthening of $\wmemrep$, where the positive information is exact
(in the sense that condition ``$x$ is $2^{-n}$-close to $S$'' is replaced
with ``$x \in S$''), was also considered under the notion of \emph{strong
recognizability} \cite[Def.~4.1]{CK95:Rtwo} and revisited later as
\emph{weak computability} by \cite[Def.~3, Thm.~4]{Braverman2005complexity}.
Although Chou/Ko seemed to be the first to \emph{formally} present this
strengthening, this concept of one-sided error also has already been present
implicitly as part of \cite[Lem.~4.3.3]{GLS88}.
Representation $\gridrep$ was defined and used \eg~in
\cite[Def.~5.2.1]{Weih00}, \cite[$\bm{\kappa}_\mathrm{G}$: Def.~2.2]{ZM08};
and $\setrep$ in \cite[$\setrep$: Def.~5.1.1]{Weih00},
\cite[$\bm{\psi}_\varocircle$: \secref{2.2.3}]{KawamuraCook}.
Questions regarding both the computability of sets \wrt different
representations (which, however, are not part of this work) and the
\emph{computability relation} of representations (which are and will be
discussed to some extent in \cref{sec:comparisons}) has been covered in many
articles (see, \eg,
\cite{BW99,Weih00,Ziegler02,Hertling02,BrattkaPresser}).

\subsection{Enrichments}
	\label{subsec:enrichments}


The representations we have seen in the previous section are rather generic.
In practice, however, additional parameters are usually known, \eg, bounds on
diameters of sets or rate of growth of functions.
Such additional discrete information (\emph{discrete advice parameters}, or
just \emph{advice parameters} for short \cite[p.18]{KMRZarXiv}) may be
\emph{uncomputable} from a given representation, but will turn out to be of
great use (complexity-wise) in
\cref{sec:comparisons,sec:operators,sec:fun-set-ops}.

\begin{defi}[enrichments; cf.~{\cite[Def.~2.4(c+d)]{KMRZarXiv}}]
	\label{def:enrichments}
Let $\sastrep \parcol \Baire \to \Sast$ denote a representation of $\Sast$.
Further, let $\reptpl \parcol \Baire \to X$ be a representation of a set $X$,
and $\ens{E} \dffn X \mto \Sast$ a multi-valued function (encodes information
$\reptpl$-names are \emph{enriched} with).
Then $\phi$ is a $\reptpl \enp{E}$-name of $x \in X$ if it is of form
$\phi = \enc{\phi_1,\phi_2}$ with $\reptpl(\phi_1) = x$ and
$\sastrep(\phi_2) \in \ens{E}(x)$.
\end{defi}

More specific, we use the following four concrete enrichments in the remainder
of this paper.

\begin{defi}[concrete enrichments for sets]\ %
\begin{enumerate}
\item Outer radii:
	Consider the enrichment
	\[
		\ens{b} \dffn \compset[d] \mto \Sast
		\eqnsp , \quad
		\ens{b} \dffn S \mmapsto \st{ \uintrep(b) }{
			b \in \IZ \text{ and } S \subseteq \cball_{\ndot}(0,2^b) }
	\]
	By definition, $\setrep[d]|^{\compset} \enp{b}$ then is a representation
	of $\compset[d]$ whereby each name contains an \emph{outer radius parameter
	$b$} (encoded in unary according to $\ens{b}$) on the encoded compact set.
	We refer to $2^b$ as an \emph{outer radius} \wrt the outer radius parameter
	$b$.
\item Inner radii and inner points:
	In a similar fashion to $\ens{b}$ define enrichments
	\begin{alignat*}{-1}
	&	\ens{r} \dffn \regset[d] \mto \Sast
		\eqnsp , \quad
	&&	\ens{r} \dffn \kern.3ex S \mmapsto \st{ \uintrep(r) }{
		r \in \IZ ,\, \xsome{x}{\inner{S}} \cball_{\ndot}(x,2^{-r}) \subseteq S }
		\eqnsp , \\
	&	\ens{a} \dffn \regset[d] \mto \Sast
		\eqnsp , \quad
		&&	\ens{a} \dffn S \mmapsto \st[1]{ \dyrep[d](a) }{
	a \in \ID^d ,\, \xusome{\delta > 0} \cball_{\ndot}(a,\delta) \subseteq S }
		\eqnsp .
	\end{alignat*}
	We refer to decoded images under $\ens{r}$ as \emph{inner radii parameter}
	(giving an \emph{inner radius of $2^{-r}$}),
	and to decoded images under $\ens{a}$ as \emph{inner points}.
\item Information encoded by $\ens{a}$ and $\ens{r}$ is independent of the
	other, \ie, the bound on an inner radius parameter according to $\ens{r}$
	need not necessarily be centered at $\ens{a}$.
	If we need both information, \ie, an inner \emph{ball}, then we combine
	it to
	\[
		\ens{ar} \dffn S \in \regset[d] \mmapsto
		\st[1]{ \enc[1]{\dyrep[d](a), \uintrep(r)} }{
				\ball_{\ndot}(a,2^{-r}) \subseteq S}
		\eqnsp .
	\]
\end{enumerate}
\end{defi}

These choices of encodings also meet both theory and practice: cf., for
example, \cite[Def.~2.1.20]{GLS88} and \cite[Def.~2.2+2.3]{Hoover90}.
Note that both the dimension and the norm will be always understood from the
context and therefore is not considered enrichment.

\emph{Convention.}
The correct way to work with enriched representations would be like this:
Let $\ens{E}$ be an enrichment, and $\enc{\phi_1,\phi_2}$ be a
$\reptpl \enp{E}$-name.
Then $E \dfeq \sastrep(\phi_2)$ is a concrete instance of
enrichment $\ens{E}$ of object $x \dfeq \reptpl(\phi_1)$.
As this intermediate step of ``extracting'' $E$ from $\enc{\phi_1,\phi_2}$
is just a technical though necessary detail which does not add to any proof
argument, we use the typographical convention to denote a concrete decoded
instance of $\ens{E}(x)$ ``variable style'', that is, as $E$.
In the above spirit, further abbreviate a
$\reptpl[d]|^{\KR} \enp{a} \enp{r} \enp{b}$-name
$\enc{\phi_1,\phi_2,\phi_3,\phi_4}$ by $\enc{\phi_1,a,0^r,0^b}$.
This notation has been purposefully chosen as a reminder that
\begin{enumerate*}[label=(\emph{\alph*})]
\item inner points as advice parameters are encoded in \emph{binary}, while
\item both inner and outer radii are encoded in \emph{unary} according to
	enrichments $\ens{r}$ and $\ens{b}$, respectively.
\end{enumerate*}


\subsection{Complexity of functions and operators: upper and lower bounds}
	\label{sec:complexity}

We briefly recap some facts from discrete complexity theory.
Assume $M$ to be a Turing machine that either accepts its input $s \in \Sast$
or rejects it; \ie, $M$ always terminates. The computation time of such a
machine $M$ is bounded by some non-decreasing function $t \colon \IN \to \IN$
(or: $t$-time bounded) if for all $s \in \Sast$, $M$ started on $s$ holds
within $t(\len{s})$ steps.

Unless stated otherwise, we use ``Turing machine'' as a synonym for
``\emph{deterministic} Turing machine''. Allowing a machine also to
\emph{guess} strings from $\Sast$ makes it \emph{non-deterministic}.
Through the course of this paper we need three complexity classes:
$\PTime$ marks the class of all problems $A \subseteq \Sast$ decidable by a
deterministic polynomial-time bounded Turing machine, and $\NPTime$ the class
of problems decidable by a \emph{non}-deterministic polynomial-time Turing
machine.
Decision problems $A \in \NPTime$ can equivalently be stated as being
\emph{polynomial-time verifiable} by a deterministic Turing machine; \ie,
there exists a decision problem $A' \subseteq \Sast$ which is polynomial-time
equivalent to $A$ and satisfies $\xsome{B}{\PTime} A' = \st[0]{s \in \Sast}{
	\xsome{w}{\Sigma^{\len{s}}} \enc{w,s} \in B}$.
Given an $s \in \Sast$, a $w$ which \emph{verifies} $\enc{w,s} \in B$ is
usually called a \emph{witness for $s \in A'$}.

The class $\UPTime$ contains problems $A \subseteq \Sast$ decidable by an
\emph{unambiguous} non-deterministic poly\-nomial-time Turing machine; that is,
a machine that for each $s \in A$ has \emph{exactly one} accepting path.
It is easy to see that $\PTime \subseteq \UPTime \subseteq \NPTime$, but
whether any of these inclusions is proper is a wide-open problem. We defer
the discussion of the hypothetical case $\PTime \neq \UPTime$ and its
implications until \cref{sec:fun-set-ops}.


As for example pointed out in \cite{FlumGrohe}, problems usually come
with a variety of structural information (like the number of nodes in a graph,
number of variables in a formula, number of faces of a polyhedron), which
however is not reflected in the above one-ary notion of complexity.
\emph{Parameterized complexity} extends upon that:
A parameterized decision problem $(A,k)$ ($A \subseteq \Sast$ with
parameterization $k \colon \Sast \to \IN$, which typically is required to be
at least computable) is $(\tau,t)$-time computable, iff a deterministic Turing
machine $M$ exists whose computation time is bounded by
$\tau(k(s)) \cdot t(\len{s})$ for all $s \in \Sast$. If $t$ moreover is a
polynomial, then $(A,k)$ is said to be \emph{parameterized polynomial-time}
decidable (also: fixed-parameter tractable).

\subsubsection{Time complexity.}

The complexity notion for oracle machines is similar to that for classical
Turing machines, except for the extension that it takes the oracle in total
one step to read the content written on the oracle tape and to produce its
answer. The computation time of an oracle machine $M^\phi$ (and thus the
complexity of the element of $\Baire$ it computes) with \emph{set-oracle} (or
equivalently a function-oracle $\phi \in \Cantor$) can solely be measured in
the length of the discrete input given to $M^?$; which is the case for
representations $\wmemrep$, $\gridrep$ and $\setrep$.

\begin{defi}   \label{def:cantor-polytime}
Let $t \colon \IN \to \IN$ be some non-decreasing function.
\begin{enumerate}
\item A function $g \parcol \Cantor \times \Sast \to \Sast$ is $t$-time
	computable if an oracle machine $M^?$ exists which for all
	$\phi \in \Cantor$ and $s \in \Sast$ computes $g(\phi,s)$ in time
	bounded by $t(\len{s})$.
\item Let $\reptpl$ and $\reptpl'$ be a second-order representations of
	sets $X$ and $X'$, respectively, and let
	$\dom(\reptpl) \subseteq \Cantor$.\footnote{%
		That is, $\reptpl$-names are predicates.}
	A $(\reptpl,\reptpl')$-computable function $f \parcol X \to X'$ is
	\emph{$t$-time computable} if it is realized by a $t$-time
	$(\reptpl,\reptpl')$-computable function
	$g \parcol \Cantor \times \Sast \to \Sast$.
\end{enumerate}
\end{defi}

As hinted prior to the definition, it is not that obvious how to define
complexity in case of names $\phi$ from $\Baire$ instead of $\Cantor$.
The problem with names from $\Baire$ is that oracle answers in general are
\emph{not} bounded in the length of its argument as it has been the case for
$\phi \in \Cantor$.
However, combining enrichments with \cref{def:cantor-polytime} allows us to
define time bounds whenever the oracle answers can be bounded in terms of the
parameters a representation has been enriched with.
For that purpose, force both $\distrep[d]$- and $\reldistrep[d]$-names $\phi$
to additionally satisfy $\phi(q,0^n) \in \ID^d_{n+1}$.\footnote{%
	Note that the latter condition is only added to prevent unnecessarily
	long answers as they are not more accurate (with respect to the conditions
	on $\distrep$- and $\reldistrep$-names) when provided with a precision
	much higher than $n$.
	This restriction is not necessary in the general theory of
	\emph{second-order polynomials} and \emph{second-order polynomial-time},
	but we defer this discussion until \cref{sec:fun-set-ops}.}
The subsequent definition is based on \cite[Def.~2.1+2.2]{KMRZarXiv}.

\begin{defi}[parameterized complexity] \label{def:baire-polytime}
Let $\reptpl$ and $\reptpl'$ be second-order representations of sets $X$
and $X'$, respectively, and $\ens{E} \dffn X \mto \Sast$ an enrichment of $X$.
Moreover, let $\tau, t \dffn \IN \to \IN$ be non-decreasing functions.
\begin{enumerate}
\item A function $f \parcol X \to X'$ is $(\tau,t)$-time
	$(\reptpl \enp{E},\reptpl')$-computable if it has a
	$(\reptpl \enp{E},\reptpl')$-realizer
	$g \parcol \Baire \times \Sast \to \Sast$ such that the computation time
	on every input $(\phi,s) \in \dom(\reptpl) \times \Sast \mapsto g(\phi,s)$ is bounded
	by $\tau\big( \len{\ens{E}(\reptpl(\phi))} \big) \cdot t(\len{s})$.%
	\footnote{Notice the term ``$\ens{E}(\reptpl(\phi))$'' in the time bound:
		Enrichments are by definition multi-valued but $\tau$ is not set-valued.
		Although technically incorrect, the meaning ``this bound has to hold
		true for \emph{every advice parameter in $\ens{E}(...)$}'' clearly is
		supported by this notation while the correct statement would be
		``the computation time has to be bounded by
		$\tau(\len{E}) \cdot t(\len{s})$
		for \emph{all} $E \in \ens{E}(\reptpl(\phi))$''.}
\item If $t$ is a polynomial, then $f$ is said to be \emph{parameterized
	polynomial-time $(\reptpl \enp{E},\reptpl')$-computable}.
\item If both $t$ and $\tau$ are polynomials, then $f$ is said to be
	\emph{fully polynomial-time $(\reptpl \enp{E},\reptpl')$-computable}.
\end{enumerate}
As advice parameters are part of a name anyway, we simply speak about
``polynomial time'' whenever ``fully polynomial-time'' is meant.
This identification is justified as fully polynomial-time and unparameterized
polynomial-time coincide for $\ens{E} \dffn x \mmapsto \{\eword\}$.
\end{defi}

On compact sets $K \in \compset[d]$, this definition allows to bound the
answer lengths in terms of an outer radius parameter $b$ as in
\cref{def:enrichments}(1).
Take representation $\distrep$ as an example:
Assume $\phi \dfeq \enc[1]{\phi',0^b}$ to be a $\distrep[d] \enp{b}$-name of $K$.
Then $\len{\phi'(q,0^n)}$ can be bounded linearly in
$\abs{b} + \len{\enc{q,0^n}}$ for all $q \in \ID^d$ and $n \in \IN$.

\subsection{Common proof arguments}
	\label{sec:proof-techniques}

We review two common arguments that allows us to prove lower bounds or
even the uncomputablitiy of operations.

\subsubsection{Adversary argument.} \label{sec:adversary-method}

The \emph{adversary method} is used to prove lower bounds on the \emph{uniform}
computational complexity of functions.
Let $\reptpl$, $\reptpl'$, $X$, $X'$ and $f$ as in the previous subsections.
For any discrete argument $s \in \Sast$ pick an element $x \in X$ and construct
a subset $Y \subset X$ of cardinality at least exponential in $\len{s}$
such that every $y \in Y$ has a $\reptpl$-name $\phi$ close to one of $x$,
yet $f(x)$ differs by at least $2^{-n}$ from $f(y)$.
Then any machine $M^?$ that $(\reptpl,\reptpl')$-realizes $f$ necessarily has to
ask exponentially many queries to $\phi$.

This approach is similar to the adversary method from \emph{Information-Based
Complexity} \cite{IBC} where computations are exact, but only finite
information is known about the input. As an example, take Riemann
integration, done on $2^n$ many sampling points in order to achieve an
approximation which is always guaranteed to be within error $2^{-n}$.

\subsubsection{Topological discontinuity.}

Given second-order representations $\reptpl$ and $\reptpl'$ of sets $X$ and
$X'$, respectively, and a function $f \parcol X \to X'$. By \cref{s:tte-cont}
we already know that $f$ is not $(\reptpl,\reptpl')$-computable whenever it is
not $(\reptpl,\reptpl')$-continuous. Recall that $\Baire$ comes equipped with
the product topology, providing a way to prove the latter:
Construct an $x \in X$ and an
appropriate $\reptpl$-name $\phi$. Any machine for a hypothetical
$(\reptpl,\reptpl')$-realizer for $f$ does only inspect finitely many values
of $\phi$. Now pick a slightly different $\reptpl'$-name, say $\phi'$, for a
different value, say $x'$, which coincides with $\phi$ on values observed by
$M^?$, but leads it to produce an answer exceeding the prescribed error bound.


\section{Comparing representations of sets} \label{sec:comparisons}

In this section, we compare the representations introduced in
\cref{def:set-representations} with respect to their mutual polynomial-time
\emph{reducibility}.
Two aspects will play a key role in these comparisons:
Whether a representation $\reptpl$ is \emph{norm-invariant},
\ie, if $\wrtn{\reptpl}{\ndot}$ and $\wrtn{\reptpl}{\ndot'}$ are
polynomial-time equivalent,
and the influence of the dimension parameter.
Both together will prove $\distrep[d]$ to be \emph{richer} (intuition: to
carry \emph{more} information) than all of the other representations from
dimension $d=2$ onward by combining that%
\begin{enumerate*}[label=(\emph{\alph*})]
\item $\distrep[d]$ is not norm-invariant for $d \geq 2$,
\item all of the other representations we discuss \emph{are} norm-invariant,
	and
\item $\distrep$ reduces to all of the other representations in polynomial
	time.
\end{enumerate*}
Representation $\wmemrep$, on the other hand, will prove to be \emph{poorer}
than all of the other representations.
However, this gap between $\distrep$ and $\wmemrep$ can be closed by
restricting to $\KCR$, adding parameters to $\wmemrep$ and applying techniques
from discrete optimization (\cref{s:convex-opt}), which proves \emph{all}
representations to be polynomial-time equivalent in this particular setting.

We now turn to the formalization of what has been described above.

\begin{defi}[translations/reductions; cf.~{\cite[Def.~2.3.2]{Weih00}}]
	\label{def:translations}
Let $\reptpl$ and $\reptpl'$ be representations of the same set $X$.
Then $\reptpl$ uniformly \emph{translates} (or: \emph{reduces}) to
$\reptpl'$ \emdash $\reptpl \preceq \reptpl'$ for short \emdash if
$\id_X$ is $(\reptpl,\reptpl')$-computable.
If $\id_X$ is parameterized polynomial-time $(\reptpl,\reptpl')$-computable,
then we write $\reptpl \parampleq \reptpl'$.
If $\id_X$ is even fully polynomial-time $(\reptpl,\reptpl')$-computable,
then we write $\reptpl \pleq \reptpl'$.
\end{defi}

Note that while the notation $\reptpl \preceq \reptpl'$ makes sense when read
as ``$\reptpl$ translates to $\reptpl'$'', it is counter-intuitive when read
as a reduction:
$\reptpl$ reduces to $\reptpl'$ if $\reptpl$-names carry
\emph{more} information than $\reptpl'$-names; hence, $\reptpl$-names are
\emph{harder} to compute than $\reptpl'$-names.
Reductions in classical complexity theory are usually thought the other way
around, \ie, the \emph{harder} problem being ``greater or equal'' to
\emph{easier} problems.

The intuition about representations encoding more or less information
also explains the following fact which we will use in many places throughout
this paper.

\begin{fact}
Let $\reptpl_1,\reptpl_2$ be representations of $X$, and
$\reptpl'_1,\reptpl'_2$ be representations of $X'$.
Then every $(\reptpl_1,\reptpl'_1)$-computable function is also
$(\reptpl_2,\reptpl'_2)$-computable whenever $\reptpl_2 \preceq \reptpl_1$
(providing potentially \emph{more} information about the input) and
$\reptpl'_1 \preceq \reptpl'_2$ (requiring potentially \emph{less} information
about the output) hold. The same applies if $\preceq$ is replaced with $\pleq$.
\end{fact}

\emph{Convention.}
For $Y \subseteq X$ let $\reptpl \preceq^Y \reptpl'$ be an abbreviation for
$\reptpl|^Y \preceq \reptpl'|^Y$. This new representation
$\reptpl|^Y \parcol \Baire \to Y$ is the result of the restriction of
$\reptpl$'s image to $Y$.
Apply the same to $\pleq$, $\pless$, $\equiv$ and $\pequiv$.


\subsection{Technicalities}

Two subtle details need to be addressed before we can start comparing
representations with respect to polynomial-time reducibility:
\begin{enumerate*}[label=(\emph{\alph*})]
\item their dependence on the choice of norm underlying $\IR^d$, and
\item considering also 'negative' values of the precision parameter $n$,
   that is, absolute error boundy larger than one.
\end{enumerate*}
The former leads to the notion of \emph{well-behaved norms}, while the latter
introduces \emph{scale-invariant representations}.

\subsubsection{Restriction to well-behaved norms.} \label{subsec:norms}

Representations $\setrep$, $\gridrep$ and $\wmemrep$ depend on the notion of
``being close''.
Practically speaking, a point $q$ gets printed on the screen whenever
$\cball_{\normdot}(q,2^{-n})$ meets the represented set, where points $q$
were chosen from $\ID^d_n$.
The implicit assumption underlying all representations from
\cref{def:set-representations} is \emph{compatibility with the grid
$\ID^d_n$}:
The whole space, say $X$, can be covered with $\normdot$-balls with radii
$2^{-(n+c)}$ and centers from $X \cap \ID^d_n$, where $c \in \IZ$ is a
constant depending on the pair $\ndot, \ndot[\infty]$.
As an example, take $\normdot \dfeq 4 \cdot \normdot[\infty]$ with
$c \dfeq -\lb{4}$.

The implied necessity to incorporate a norm-dependent constant $c$ into
precision parameters is cumbersome and we avoid it by \emph{imposing}
the above mentioned ``compatibility'' on the respective norm.
For the same reason we disallow \emph{slanted} (or otherwise distorted) norms
like $\norm{(x_1,x_2)} \dfeq (\abs{x_1/2}^2 + \abs{x_2}^2)^{1/2}$ although
this restriction can be avoided (as we will see in \cref{rem:cov-pat}) and
is only present to simplify things.

We denote such norms satisfying both of the above motivated properties as
being \emph{well-behaved}.

\begin{defi}[well-behaved norms] \label{def:valid-norms}
A norm $\normdot$ on $\IR^d$ is said to be \emph{well-behaved} if it has
the following two properties:
\begin{enumerate}
\item $\ndot$ is invariant under 90-degree rotations.
	More precisely:
	Let $\{e_1,\dots,e_d\}$ be the canonical basis of $\IR^d$.
	Then $\norm{e_i} = \norm{e_j}$ for all $1 \leq i,j \leq d$.
\item $\ndot$-balls are not too small, \ie,
	$\cball_{\ndot[\infty]}(q,2^{-(n+1)}) \subseteq \cball_{\ndot}(q,2^{-n})$
	for all $n \in \IN$ and $q \in \ID^d_n$.
\end{enumerate}
\end{defi}

It then follows by the second condition that $\IR^d$ can be covered by
$\ndot$-balls with centers from $\ID^d_n$ and radii $2^{-n}$.
Examples of well-behaved norms include the $p$-norms
$\norm[p]{(x_1,\dots,x_d)} \dfeq (\abs{x_1} + \cdots + \abs{x_d})^{1/p}$
for $p \geq 1$, while
$\norm{(x_1,x_2)} \dfeq (\abs{x_1/2}^2 + \abs{x_2}^2)^{1/2}$
(violates the first condition) and
$3/2 \normdot[1]$ (violates the second condition) are not.

\emph{Convention.}
For the rest of this paper, we only consider well-behaved norms unless
stated otherwise.

\subsubsection{Scale-invariance.}
	\label{sec:scale-inv}


Starting with \cite{KF82}, complexity results were stated for functions
whose domains were a subset of the \emph{unit} hypercube; the same was true
for sets. This restriction rendered (at least for sets) the question about
precision parameters \emph{smaller than $0$} (\ie, absolute error bounds $>
2^{-0}$) pointless, which allowed for a complexity notion solely in the
precision parameter.
However, as we will see many times in \cref{sec:operators,sec:fun-set-ops},
algorithms for operators on sets often involve an unavoidable preprocessing
step (given the representations we have seen so far):
Given $b \in \IN$, chop $\clb[\normdot[\infty]]{0}{2^b}$ into $2^{d(b+1)}$
unit hypercubes,
pick a subset of them (usually one cube), then proceed by applying the
given $\reptpl$-name to this subset.
It is this preprocessing step which seems to be artificial and superfluous
as the \emph{real} algorithm often starts only \emph{after} this step.
For this reason the author believes that a closed subset $S$ of $[0,1]$ (or
any fixed compact set) should admit, up to a polynomial rather than
exponential in $k$, the same complexity as $S$ inflated by a factor of $2^k$.
Both sets are still \emph{structurally} the same!

To this end, let $\sisetrep$ be the extension (or: relaxation) of
representation $\setrep$ to \emph{integer precisions}, \ie,
a $\sisetrep[d]$-name $\phi$ satisfies
\[
	\phi(s) = 1
	\text{ if } S \cap \opb{q}{2^{-n}} \neq \emptyset
	\eqnsp ; \qquad
	\phi(s) = 0
	\text{ if } S \cap \clb{q}{2^{-n+1}} = \emptyset
	\eqnsp .
\]
with $s \dfeq \enc[1]{\dyrep[d](q),\uintrep(n)}$.
Recall that we agreed to equivalently write $s$ as $\enc[1]{q,0^n}$ where
$0^n$ abbreviates the ``unary encoding'' of integer $n$.
With the following Lemma we attempt to provide a way around the above
described dilemma.\\

\begin{lem}[properties of $\sisetrep$] \ \label{s:sisetrep-scale}
\begin{enumerate}
\item Scaling a closed set by factor $2^{k}$ for $k \in \IZ$ is a
	\emph{parameterized} polynomial-time operation in the \emph{absolute value}
	of $k$ (cf.~\cite[Lem.~2.7(4)]{ZM08}), that is,
	the binary length of $2^k$.
	More precisely:
	Operator $\opname{Scale} \dffn (S,k) \mapsto \st[0]{2^k x}{x \in S}$ is
	$(\setrep[d] \times \uintXrep, \setrep[d])$-computable in parameterized
	polynomial time.
\item In contrast,
	$\opname{Scale}$ is \emph{fully} polnomial-time
	$(\sisetrep[d] \times \uintXrep, \sisetrep[d])$-computable.
\item $\setrep[d] \parampleq \sisetrep[d] \pleq \setrep[d]$.
\end{enumerate}
\end{lem}


\begin{proof}[sketch]
The first statement follows by the argument hinted prior to this Lemma:
Let $q \in \ID^d$ and $n \in \IN$.
If $n-k \geq 0$, then simply query the $\setrep[d]$-name with
$\enc{2^k q,0^{n-k}}$.
If $n-k < 0$, then first split $\cball(2^k q, 2^{k-n})$ into unit-balls and
combine the queries on the center and precision $0$ on each of these balls.
For the second statement, use the argument from the above first case, namely,
query the $\sisetrep[d]$-name with precision $n-k$.
The first reduction in statement three follows immediatiely from 1.~and 2.
For the second reduction, use the split of $\cball(2^k q,2^{k-n})$ into
unit-balls and argue as in the first case of statement one.
\qedhere
\end{proof}

It follows by the previous statement that \emph{all} scaled versions of a set
are polynomially equivalent with respect to $\sisetrep$.

\begin{rem}                                             \label{rem:only-si}
As the concept of a scale-invariant representation avoids the above described
deficiencies, we like to impose it on \emph{every} representation $\reptpl$
from \cref{def:set-representations}.
Therefore, we will denote $\widehat{\reptpl}$ to be understood as the
scale-invariant version of $\reptpl$, and then associate $\reptpl$ with
$\widehat{\reptpl}$ (\ie, drop the explicit hat).
As a consequence, precision parameters shall now usually be \emph{integers}.
\end{rem}


\subsection{Topological versus computable equivalence of norms}
	\label{subsec:poly-norm-invariance}


In this section we examine the question which representations $\reptpl$ are
\emph{norm-invariant}, \ie, whether
$\wrtn{\reptpl}{\ndot} \pequiv \wrtn{\reptpl}{\ndot'}$ holds true for
\emph{all} topological equivalent well-behaved norms $\ndot, \ndot'$.
Notice that ``norm-invariance'' inherently asks about polynomial-time
equivalence:
norm-exchange is a computable operation for all representations from
\cref{def:set-representations}.

Our first result generalizes Braverman's remark
\cite[following Def.~2]{Braverman2005complexity} on the interchangeability of
$\wrtn{\setrep}{\normdot[2]}$ and $\wrtn{\setrep}{\normdot[\infty]}$.

\begin{prop}
	\label{lem:setrep-norm-equiv}
$\wrtn{\setrep[d]}{\normdot} \pequiv \wrtn{\setrep[d]}{\normdot'}$
holds in any dimension $d \in \IN$ and for any two norms
$\ndot,\ndot'$ on $\IR^d$.
\end{prop}

The key to prove this proposition is its \emph{non-uniformity} \wrt any two
well-behaved norms $\normdot,\normdot'$:
The necessary information (here: the ``coverage pattern'' of the unit
$\normdot'$-ball) for a machine to translate from
$\wrtn{\setrep[d]}{\normdot}$ to $\wrtn{\setrep[d]}{\normdot'}$ can be
directly encoded into it.

\begin{rem} \label{rem:cov-pat}
For every two norms $\ndot,\ndot'$ on $\IR^d$ exists a constant $k \in \IN$
and a finite set $D \subset \ID^d_k$ (``coverage pattern'') such that
\[
	\cball_{\ndot'}(0,1)
	\subseteq
	\bigcup\nolimits_{p \in D} \cball_{\ndot}(p,2^{-k})
	\subseteq
	\cball_{\ndot'}(0,3/2)
	\eqnsp .
\]
\end{rem}

Note that $\ndot'$ can only be approximated by $\ndot$-balls up to a
\emph{constant factor} by the above coverage pattern $D$.
Approximating the shape of a $\ndot'$-ball up to \emph{arbitrary precision},
however, might still be uncomputable.

\begin{proof}[of \cref{lem:setrep-norm-equiv}]
Let $k \in \IN$ and $\ID^d_k$ as in \cref{rem:cov-pat}.
Let further $\phi$ be a $\wrtn{\setrep[d]}{\ndot}$-name of $S \in \clset[d]$,
$q \in \ID^d$, and $n \in \IZ$.
Claim: $\phi'$, defined as
\[
	\phi'(q,0^n) \dfeq \max_{p \in D} \phi(p', 0^{n+k})
	\eqnsp , \quad
	p' \dfeq q + 2^{-(n+k)} p
	\eqnsp ,
\]
is a $\wrtn{\setrep[d]}{\ndot'}$-name of $S$.
Note that since $D$ is finite, the maximum ranges only over finitely many
values and is therefore computable in time linear in $n + k + \len{\enc{q}}$.

If $\ball_{\ndot'}(q,2^{-n}) \cap S \neq \emptyset$, then by
\cref{rem:cov-pat} there exists a point $p_0 \in D$ such that
$\cball_{\ndot}(p',2^{-(n+k)})$ meets $S$, justifying
$\phi'(q,2^{-n}) = 1$.
If, on the other hand, $\cball_{\ndot'}(q,2^{-n+1}) \cap S = \emptyset$, then
in particular $\cball_{\ndot}(p',2^{-(n+k)}) \cap S = \emptyset$ for all
$p \in D$.
Their union covers $\cball_{\ndot'}(q,2^{-n})$ which renders
$\phi'(q,2^{-n}) = 0$ to be correct.
\qedhere
\end{proof}

\emph{Noteworthy:}
Neither one of the norms has actually to be \emph{computable} \emdash a direct
consequence of the note following \cref{rem:cov-pat}.

The argument from \cref{lem:setrep-norm-equiv} generalizes to $\wmemrep[d]$
(over $\regset[d]$) and $\gridrep[d]$ (over $\compset[d]$), rendering both
representation to be norm-invariant, too. Representation $\distrep$,
however, turns out to be \emph{not} norm-invariant \emdash not even
non-uniformly (provided $\PTime \neq \NPTime$):

\begin{thm}[$\distrep$ is not polynomial-time invariant under a change of norms \emph{unless} $\PTime \neq \NPTime$]
   \label{thm:distrep-not-polytime-norm-invariant}
In any dimension $d \geq 2$ there is a set $S \in \compset[d]$ that is
polynomial-time $\wrtn{\distrep[d]}{\normdot[1]}$-computable but \emph{not}
polynomial-time $\wrtn{\distrep[d]}{\normdot[\infty]}$-computable
if and only if $\PTime \neq \NPTime$.
\end{thm}

\begin{proof}
\emph{Only if ($\PTime = \NPTime$ implies that $\distrep$ is
(non-uniformly) norm-invariant)}.
Suppose $\PTime = \NPTime$ and let $\phi_1$ be a polynomial-time
$\wrtn{\distrep[d]}{\ndot[1]}$-computable name of $S$.
Now consider the sets $N$ and $P$,
\begin{align*}
	N & \dfeq \st[1]{ \enc{p,\delta,0^n,0^m} }{
		\xusome{p' \in \ID^d_{n+2},\, \phi_1(p',0^{n+2}) \leq 2^{-(n+2)}}
		\enc{p,p',\delta,0^n,0^m} \in P }
	\eqnsp , \\
	P & \dfeq \st[1]{ \enc{p,p',\delta,0^n,0^m} }{
		\abs{\delta - \norm[\infty]{p-p'}} \leq 2^{-m} }
	\eqnsp ,
\end{align*}
which in turn are polynomial-time decidable by the above assumption.

\begin{figure}[htb]
	\centering
	\includegraphics{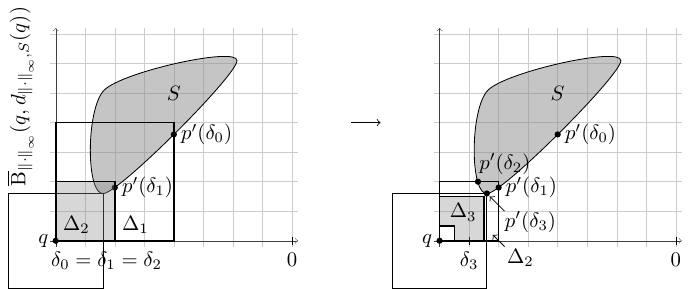}
	\caption{%
		Search for a $2^{-n}$-approximation $\delta_{n+1}$ to
		$\dist[\ndot[\infty]]{S}(q)$ by iteratively determining distances
		$\delta_i$ and associated narrowed sets $\Delta_{i+1} \dfeq
		\st{x \in [0,1]^2}{\abs{\delta_i-\norm[\infty]{x-q}} \leq 2^{-(i+1)}}$
		such that $\delta_i \leq \dist[\ndot[\infty]]{S}(q)$ and
		$\xsome{p'(\delta_i)}{\Delta_{i+1} \cap \ID^d_{n+2}}
			\enc{q,p'(\delta_i),\delta_i,0^n,0^{i+1}} \in P$.
	}
	\label{fig:distrep-binsearch}
\end{figure}

A $\wrtn{\distrep[d]}{\ndot[\infty]}$-name for $S$ can be recovered from
queries ``$\enc{q,\delta_i,0^n,0^i} \in N$?'' by the following iterative
procedure (cf.~\cref{fig:distrep-binsearch}):
Let $\delta_0 \dfeq 0$.
Then for each $1 \leq i \leq n+1$ set $\delta_i \dfeq \delta_{i-1}$ if
$\enc{q,\delta_{i-1},0^n,0^i} \in N$, and
$\delta_i \dfeq \delta_{i-1} + 2^{-i}$ otherwise.
This way,
\begin{align}
	\label{eqn:distrep-if}
	\delta_i
	\leq \dist[\ndot[\infty]]{S}(q)
	\leq \delta_i + 2^{-i} + 2 \cdot 2^{-(n+2)}
	\eqnsp ,
\end{align}
and therefore
$\abs{\dist[\ndot[\infty]]{S}(q) - \delta_{n+1}} \leq 2^{-n}$.
We prove the correctness of \cref{eqn:distrep-if} by induction.
For $i=0$ it surely is true, so consider the case $i > 0$.
If $\enc{q,\delta_{i-1},0^n,0^i} \in N$, then \labelcref{eqn:distrep-if}
holds true for $\delta_{i} \dfeq \delta_{i-1}$ by the construction of $N$.
If, on the other hand, $\enc{q,\delta_{i-1},0^n,0^i} \not\in N$,
then for all $p' \in \ID^d_{n+2}$ with $\phi_1(p',0^{n+2}) \leq 2^{-(n+2)}$
we have $\abs{ \delta_{i-1} - \norm[\infty]{q-p'} } > 2^{-i}$.
Since \labelcref{eqn:distrep-if} holds for $\delta_{i-1}$, it firstly
implies $\dist[\ndot[\infty]]{S}(q) > \delta_{i-1} + 2^{-i}$.
But then \labelcref{eqn:distrep-if} rewrites as
\[
	\delta_{i-1} + 2^{-i}
	\leq \dist[\ndot[\infty]]{S}(q)
	\leq \delta_{i-1} + 2^{-i+1} + 2^{-n}
\]
which is exactly \labelcref{eqn:distrep-if} for
$\delta_i \dfeq \delta_{i-1} + 2^{-i}$.

Consequently, $\phi(q,0^n) \dfeq \delta_{n+1}$ gives a
$\wrtn{\distrep[d]}{\ndot[\infty]}$-name of $S$.


\emph{If}.
We prove this direction only for $d = 2$, but the generalization to higher
dimensions follows by similar constructions.
Assuming $\PTime \neq \NPTime$, we construct an adversary set $A$ through a
proper encoding of an $\NPTime$-complete problem $N \subset \Sigma^\ast$ of
form $N = \st{s \in \Sigma^\ast}{\xsome{w}{\Sigma^{\len{s}}} \enc{w,s} \in P}$,
$P \in \PTime$, into $A$.
To this end, for $n \in \IN$ and $0 \leq i < 2^n$ associate the $i$-th
string $s \in \Sigma^{n}$ with the set
$A_{n,i} \subset [s_{n,i}, s_{n,i+1}] \times [0,2^{-(2n+1)}]$ where
$s_{n,0} \dfeq 1-2^{-n}$, $s_{n,i} \dfeq s_{n,0} + i \cdot 2^{-(2n+1)}$ and
(just to simplify the notation) $s_{n,2^n} \dfeq s_{n+1,0}$.
For each word $s \in \Sigma^n$ we then split its associated set $A_{n,i}$
into $2^n$ slices $A_{n,i,j}$, $0 \leq j < 2^n$, where $A_{n,i,j}$ is
associated with the $j$-th string $w \in \Sigma^{n}$.
To this end, let $s_{n,i,j} \dfeq s_{n,i} + j \cdot 2^{-(3n+1)}$ and
$s_{n,i,j+1}$. 
Whenever $\enc{w,x}$ is in $P$ we code a ``bump'' in $A_{n,i,j}$, and a
simple line otherwise; \ie, for $w,s \in \Sigma^{n}$,
$A_{n,i,j} \dfeq s_{n,i,j} + 2^{-(3n+1)} \cdot A_{\wedge}$ if
$\enc{w,s} \in P$, and $A_{n,i,j} \dfeq s_{n,i,j} + 2^{-(3n+1)} \cdot A_{-}$
otherwise; $A_{\wedge} \dfeq \st[0]{(x,y) \in [0,1]^2}{
	x - y = 0 \text{ for } x \leq 1/2,
	\text{ and } x + y = 1 \text{ for } x > 1/2}$,
$A_{-} \dfeq \st[0]{(x,y) \in [0,1]^2}{y = 0}$.
Thus $A \dfeq \bigcup_{n,i,j \in \IN,\, 0 \leq i,j < 2^n} A_{n,i,j}$
encodes $N$.

\begin{figure}[htb]
	\centering
	\textrm{(a)}\ %
	\includegraphics[scale=.285]{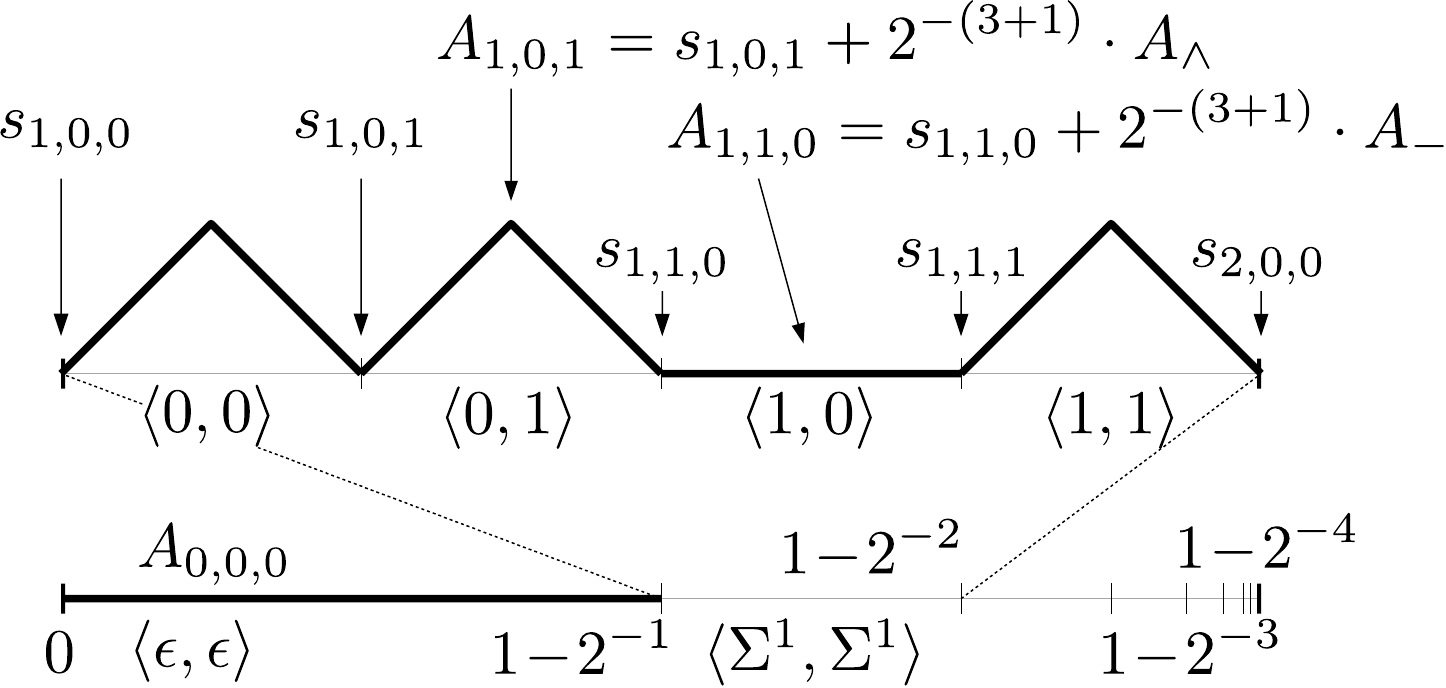} \quad
	\textrm{(b)}\ %
	\includegraphics[scale=.41]{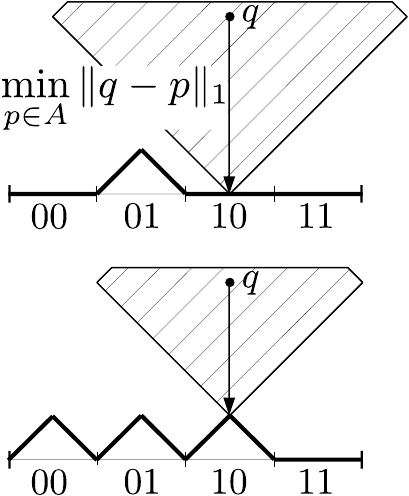}
	\textrm{(c)}\ %
	\includegraphics[scale=.41]{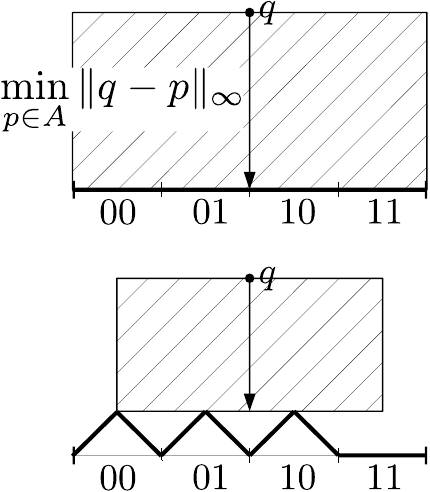}
	\caption{Encoding a certain $\NPTime$-set $N$ into a
		polynomial-time $\wrtn{\distrep[2]}{\normdot[1]}$-computable set $A$
		such that $A$ being also polynomial-time
		$\wrtn{\distrep[2]}{\normdot[\infty]}$-computable would imply
		$\PTime = \NPTime$.}
	\label{fig:distrep-no-uniform-norm-inv-detailed}
\end{figure}

Without further notational overhead associate each point
$q \in \ID^1 \cap [0,1]$ with the (lexicographically) largest triple
of indices $(n,i,j)$ such that $q$ belongs to $[s_{n,i,j},s_{n,i,j+1}]$.
As before, $\enc{w,s}$ is also uniquely identified by this triple.
Now it is easy to construct a $\wrtn{\distrep[2]}{\normdot[1]}$-name
for $A$, while it is hard (\ie, not computable in polynomial time) to
construct one \wrt $\wrtn{\distrep[2]}{\normdot[\infty]}$ (both cases are
also sketched in \cref{fig:distrep-no-uniform-norm-inv-detailed}(b) and
\cref{fig:distrep-no-uniform-norm-inv-detailed}(c), respectively).

\begin{itemize}
\item A $\wrtn{\distrep[2]}{\normdot[1]}$-name $\phi$ of $A$ can be
	constructed in polynomial time:
	\[
		\phi\big( (q_1,q_2), 0^n \big) \dfeq
		\abs[1]{
			q_2 - \chi_P\enc{w,s} \cdot \big( 2^{-(3n+2)}
			- \abs{ q_1 - (s_{n,i,j} + s_{n,i,j+1})/2 } \big)
		}
	\]
\item Now consider $\wrtn{\distrep[2]}{\normdot[\infty]}$.
	Assume there was a polynomial-time OTM $M^?$ which could compute a
	$\repnorm{\distrep[d]}{\infty}$-name $\phi'$ for $A$.
	Evaluating $\phi'$ at
	$(q'_1,q'_2) \dfeq \big(
		(s_{n,i,0} + s_{n,i,0})/2,\,
		2^{-(2n+1)}
	\big)$
	with precision $n' \dfeq 3n+4$ then decides $N$ because
	$\phi'\big( (q'_1,q'_2), 2^{-n'} \big) \geq 2^{-(2n+2)} - 2^{-(3n+3)}$
	if and only if a witness $w \in \Sigma^n$ exists with $\enc{w,s} \in P$.
\qedhere
\end{itemize}
\end{proof}


\subsection{Polynomial-time relations}      \label{subsec:poly-time-relations}

Representations $\setrep[d]$, $\gridrep[d]$ and $\wmemrep[d]$ are uniformly
poly\-nomial-time invariant under a change of norms; and so is $\reldistrep[d]$
according \cref{prop:setrep-reldrep} below---however representation $\distrep[d]$
in general is not, even non-uniformly subject to $\PTime \neq \NPTime$,
although it is computably equivalent to $\setrep[d]$ \cite[Theorem 3.12]{BW99}.
In fact, restricted to the class $\cb$ of convex bodies,
four of our five representations are known computably equivalent.

\begin{fact}[{\cite[Cor.~4.13]{Ziegler02}}]        \label{fact:all-comp-equiv}
$\distrep[d] \equiv^{\cb} \setrep[d] \equiv^{\cb} \wmemrep[d]$
in any dimension $d \in \IN$.
\end{fact}
They are all equivalent because (intuitively speaking) points can be found
due to regularity (regular sets are full-dimensional), and can be checked
(locally) to be of the desired precision due to convexity
(check if all points in a small neighborhood are also contained in the set).

In this section we now systematically compare all representations from
\cref{def:set-representations} regarding their polynomial-time reducabilites
in
\begin{enumerate*}[label=(\emph{\alph*})]
\item dimension $d=1$ and for $d \geq 2$, and
\item over various subclasses of $\clset[d]$.
\end{enumerate*}
As a result, representations $\setrep[d] \enp{b}$, $\reldistrep[d] \enp{b}$,
and $\gridrep$ prove to be $\pleq$-equivalent over $\compset[d]$ for every
$d \in \IN$.
Taking $\setrep[d]$ as a representative for this equivalence class,
$\distrep[d] \pless \setrep[d]$ holds true for $d \geq 2$,
and $\setrep[d] \pless \wmemrep[d]$ in any dimension (both even on
$\KCR[d]$!), which leaves us in a very different situation compared to
\cref{fact:all-comp-equiv}.
However: This distinction between $\distrep[d]$ and $\wmemrep[d]$ disappears
when given the right set of additional parameters (\cref{s:convex-opt}),
yielding one equivalence class of representations for sets as the result.



\subsubsection{Polynomial-time reducibilities in dimension $d=1$.}


\begin{prop} \label{lem:polytime-relations-dim-one}
$\repnorm{\distrep[1]}{} \enp{b} \pequiv^{\compset}
 \repnorm{\reldistrep[1]}{} \enp{b} \pequiv^{\compset}
 \repnorm{\setrep[1]}{} \enp{b} \pequiv^{\compset}
 \repnorm{\gridrep[1]}{}$,
$\setrep[1] \pleq^{\regset} \wmemrep[1]$,
and $\repnorm{\wmemrep[1]}{} \enp{ar} \pleq^{\cb} \repnorm{\setrep[1]}{}$.
\end{prop}


\begin{proof}
Without loss of generality, let $\ndot \dfeq \ndot[\infty]$.
Notice that the reductions
$\distrep[1] \enp{b}
	\pleq^{\compset} \reldistrep[d] \enp{b}
	\pleq^{\compset} \setrep[1] \enp{b}
	\pleq^{\compset} \gridrep[1]$
already follow by definition of the respective representations.
%
\begin{itemize}
\item Reduction $\setrep[1] \enp{b} \pleq^{\compset} \distrep[1]$:
	Let $\enc{\phi,0^b}$ be a $\setrep[1] \enp{b}$-name of a closed
	$S \subseteq \cball(0,2^b)$.
	Further, set $b' \dfeq \max\{1,b\}$ and $c' \dfeq \lb{\max\{2,\norm{q}\}}$.
	For any $q \in \ID$ and $n \in \IZ$, test if $\phi(q,0^{n+1}) = 1$.
	If it is, then $0$ is a valid $2^{-n}$-approximation of $\dist{S}(q)$.
	If, on the other hand, $\phi(q,0^{n+1}) = 0$, then first find the smallest
	$i \in \IN+$, $1 \leq i \leq n+b'+c'+1$, with $\phi(q,0^{n+1-i}) = 1$.
	Having found $i$, continue with two binary searches, one in $[q-2^{i-n},q]$
	and the other in $[q,q+2^{i-n}]$, for points $p_-,p_+ \in \ID_{n+1}$
	eventually satisfying $\phi(p_\pm,0^{n+1}) = 1$ and minimizing
	$\norm{q - p_\pm}$.
	Then $\min \{ \abs{q-p_-}, \abs{q-p_+} \}$ consitutes a valid
	$2^{-n}$-approximation of $\dist{S}(q)$.
\item Reduction $\reldistrep[1] \enp{b} \pleq^{\compset} \distrep[1]$
	follows by $\reldistrep[1] \pleq^{\compset} \setrep[1]$ and
	$\setrep[1] \enp{b} \pleq^{\compset} \distrep[1]$ from above.
\item Reduction $\gridrep[1] \pleq^{\compset} \setrep[1]$:
	Any $\gridrep[1]$-name $\phi$ induces a $\setrep[1]$-name $\phi'$ of the
	same set by $\phi'(q,0^n) \dfeq \max_p \st[0]{ \phi(p,0^{n+1})}{
			p \in \ID_{n+1} \cap [q-2^{-n},q+2^{-n}]}$.
	Since $|\ID_{n+1} \cap [q-2^{-n},q+2^{-n}]| \leq 5$, constantly
	many queries to $\phi$ suffice to devise $\phi'$.
\item Reduction $\setrep[1] \pleq^{\regset} \wmemrep[1]$:
	Every $\setrep[1]$-name $\phi$ constitutes a $\wmemrep[1]$-name $\phi'$
	of the same set through $\phi'(q,0^n) \dfeq \phi(q,0^{n+1})$.
\item Reduction $\wmemrep[1] \enp{ar} \pleq^{\regset} \setrep[1]$:
	Given a $\wmemrep[1] \enp{ar}$-name $\enc{\phi,a,0^r}$ of
	$S \in \CR[1]$, do a binary search between $a$ and $q$ for a point
	$p \in \ID_m$, $m \dfeq \max\{n,\abs{r}\}+1$, which minimizes $\abs{q-p}$
	over all such points satisfying $\phi(p,0^m)$.
	Then $\phi'$ with $\phi'(q,0^n) \dfeq 1$ if
	$\abs{q-p} \leq 3 \cdot 2^{-(n+1)}$, and defined as $0$ otherwise,
	consitutes a $\setrep[1]$-name of $S$.
	Note that convexity is cruicial in order to perform a binary search given
	only a $\wmemrep$-name.
\qedhere
\end{itemize}
\end{proof}

\subsubsection{Arbitrary yet fixed dimension.}

Some of the formerly explained relations change onward from dimension $d = 2$.
As a first example we note a result due to Braverman.
\begin{fact}[{\cite[Thm.~3.2.1]{Braverman04}}]       \label{s:Braverman-distrep}
Let $d \geq 2$.
$\PTime = \NPTime$ holds true iff every polynomial-time
$\wrtn{\setrep[d]}{\normdot[2]}$-computable $S \in \compset[d]$
is also polynomial-time $\wrtn{\distrep[d]}{\normdot[2]}$-computable.
\end{fact}

In short:
Finding the distance from a point to a set only from local information
(that is, a $\setrep$-name) about the latter is as hard as solving
$\NPTime$-problems in polynomial time.
Thus, $\distrep[d]$ is \emph{richer} (\ie, it in a sense provides \emph{more}
information about closed non-empty sets) than any of the other representations
(\ie, the others are \emph{poorer}).

We note two implications, following immediately from the proof of
\cref{s:Braverman-distrep}.
\begin{itemize}
\item The statement also holds true over $\KR[d]$.
	In fact, it	uniformizes by an adversary argument as sketched in
	\cref{sec:adversary-method}; \ie, $\wrtn{\setrep[d]}{\normdot[2]}
		\enp{b} \not\pleq^{\KR} \wrtn{\distrep[d]}{\normdot[2]}$
	for $d \geq 2$.
\item \Cref{s:Braverman-distrep} is stated with respect to $\normdot[2]$,
	but it easily generalizes to arbitrary well-behaved norms $\normdot$
	by properly adapting the adversary set's shape; \ie, from
	$\normdot[2]$-balls to $\normdot$-balls.
\end{itemize}
These two statements also apply to $\gridrep$ due to the following observation.


\subsubsection{Representation $\setrep$ with outer radii.}

$\gridrep[d]$ can be reformulated as $\setrep[d] \enp{b}$
with \emph{necessary} outer radius parameter $b$ as
every $\setrep[d] \enp{b}$-name $\enc{\phi,0^b}$ constitutes a
$\gridrep[d]$-name $\phi'$ through $\phi'(\eword) \dfeq 0^b$ and
$\phi'(q,0^n) \dfeq \phi(q,0^{n+1})$ for $q \in \ID^d, n \in \IZ$.
The reverse direction requires a little bit more care:
A point $q$ which does not belong to $B_n$ might still be arbitrarily close
to the represented set, hence $\setrep[d]$-name would have to give $1$
when queried with $\enc{q,0^n}$.
However, any $\gridrep[d]$-name \emph{does} provide enough information if
only queried on a finite set of points close to $q$.

\begin{prop} \label{thm:setrep-gridrep-equiv}
$\wrtn{\setrep[d]}{\normdot} \enp{b} \pequiv^{\compset}
\wrtn{\gridrep[d]}{\normdot}$ holds in any dimension $d \in \IN$.
\end{prop}

\begin{proof}
By the above argumentation it just remains to prove the reduction
$\gridrep[d] \pleq^{\compset} \setrep[d] \enp{b}$.

Let $q \in \ID^d$ and $n \in \IZ$, and be $\phi$ a $\gridrep[d]$-name of
$S \in \compset[d]$.
Firstly, an outer radius parameter according to $\ens{b}$ can be obtained
through $\phi(\eword)$.
It thus remains to construct a $\setrep[d]$-name $\phi'$ from queries to
$\phi$.
We claim that $\phi'(q,0^n) \dfeq \max_{p \in P} \phi(p,0^{n+2})$ with
$P \dfeq \cball(q,3 \cdot 2^{-(n+1)}) \cap \ID^d_{n+2}$ is such a name.
The correctess follows by checking the two cases from definition of $\setrep$.
If $\ball(q,2^{-n}) \cap S \neq \emptyset$, then by \ref{eqn:kappa-x-d} there
must exist a $p \in P$ with $\phi(p,0^{n+2}) = 1$, which leads to
$\phi'(q,0^n) = 1$.
Now consider the second case: $\cball(q,2^{-n+1}) \cap S = \emptyset$.
We prove it by contradition.
To this end, assume $\phi'(q,0^n) = 1$.
Then there is a $p \in P$ which satisfies \ref{eqn:kappa-d-x}, \ie,
there exists an $x \in S$ such that $x \in \cball(p,2^{-(n+2)})$ which in
turn produces a contradiction because of
$\cball(p,2^{-(n+2)}) \subset \cball(q,2^{-n+1})$.
\qedhere
\end{proof}


\subsubsection{Local information and relative distance.}

On compact sets and enriched with outer radius parameter $b \in \IZ$,
representation $\reldistrep[d]$ is polynomial-time equivalent to $\setrep[d]$.%
\footnote{%
	Recall that by \cref{rem:only-si} we assume all representations to be
	scale-invariant.
	Without it, only $\sisetrep$
	would have been fully polynomial-time equivalent to $\reldistrep$,
	while $\setrep$ would have been only parameterized polynomial-time equivalent.}

\begin{prop}                                \label{prop:setrep-reldrep}
$\wrtn{\setrep[d]}{\ndot} \enp{b} \pequiv^{\compset}
	\wrtn{\reldistrep[d]}{\ndot} \enp{b}$
holds in any dimension $d \in \IN$.
\end{prop}
\begin{proof}
We prove the polynomial-time equivalence of $\reldistrep[d] \enp{b}$ and
$\setrep[d] \enp{b}$ for $\normdot \dfeq \normdot[\infty]$.
The full statement then is a direct application of
\cref{lem:setrep-norm-equiv}.

Direction $\reldistrep[d] \enp{b} \pleq^{\compset} \setrep[d] \enp{b}$:
Let $n \in \IZ$ and $q \in \ID^d$.
If $\enc{\phi,0^b}$ is a $\reldistrep[d]|^{\compset} \enp{b}$-name of some
$S \in \compset[d]$, then
\[
	\phi'\big( q, 0^n \big) \dfeq \begin{cases}
		1, & \text{if } \phi(q,0^{n+4}) \leq 5/4 \cdot 2^{-n} + 2^{-(n+4)} \\
		0, & \text{if } \phi(q,0^{n+4}) \geq 3/4 \cdot 2^{-n} + 2^{-(n+4)}
	\end{cases}
\]
yields $\enc{\phi',0^b}$ to be a $\setrep[d] \enp{b}$-name of $S$.

	\begin{figure}[htb]
		\centering
		\includegraphics{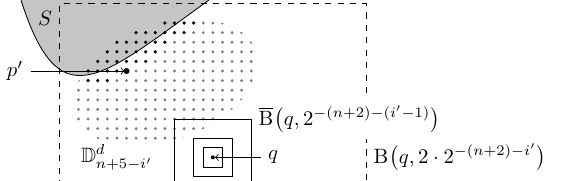}
		\caption{Reducing $\setrep[d]|^{\compset} \enp{b}$ to
			$\reldistrep[d]|^{\compset} \enp{b}$.
			Highlighted in black are points $p$ with $\phi(p,0^{n+5-i'}) = 1$.}
		\label{fig:reldistrep-sisetrep}
	\end{figure}

Direction $\setrep[d] \enp{b} \pleq \reldistrep[d] \enp{b}$:
Let $b' \dfeq \max\{1,b\}$ and $c' \dfeq \lb{\max\{2,\norm{q}\}}$.
We start by determining an initial approximation to $\dist{S}(q)$.
To this end, start with $k \dfeq 0$ and search for the
smallest value $k \leq n+b'+c'+1$ with $\phi(q,0^{n+1-k}) = 1$.
Denote this particular integer by $k'$.
Note that such a $k'$ does exist because of
$S \subseteq \clb(0){0}{2^b+\norm{q}} \subseteq \clb(0){0}{2^{b'+c'}}$.
This $k'$ then yields the bound $\dist{S}(q) \in [2^{-(n+2)+k'},2^{-n+k'}]$.

Now that we have a bound on $\dist{S}(q)$ we can decompose
$\cball(q,2^{-n+k'},2^{-(n+2)+k'})$ into a constant number of regions to
search in for a good approximation to $\dist{S}(q)$.
More precisely, let
$p' \in \ID^d_{n+5-k'} \cap \cball(q,2^{-n+k'},2^{-(n+2)+k'})$ be a dyadic
point with $\phi(p',0^{n+5-k'}) = 1$ which minimizes $\norm{q-p'}$ over
all points from the above hollow set (this argument is also depicted in
\cref{fig:reldistrep-sisetrep}).
This leads to $\abs{\dist{S}(q) - \norm{q-p'}} \leq 2^{-(n+4)+k'}$.
Moreover, $\phi'(q,0^n) \dfeq \norm{q-p'}$ satisfies \cref{eqn:def-reldistrep}.
The first half, the lower bound on $\phi'(q,0^n)$ in
\labelcref{eqn:def-reldistrep}, follows by validating that the above bound
on $\norm{q-p'}$ implies $\dist{S}(q) - 2^{-(n+4)-k'} \leq \norm{q-p'}$.
Comparing this bound with $3/4 \cdot \dist{S}(q) - 2^{-n}$ from
\labelcref{eqn:def-reldistrep} shows that
$\dist{S}(q) \geq 2^{-(n+2)+k'} - 2^{-n+2}$ has to hold in order to prove
the lower bound from \labelcref{eqn:def-reldistrep} to be true \emdash
which it does (cf.~the initial approximation we got on $\dist{S}(q)$).
The upper bound follows analogously.
Hence, $\enc{\phi',0^b}$ is a $\reldistrep[d] \enp{b}$-name of $S$.
\qedhere
\end{proof}


\subsubsection{Comparing local information.}

The situation regarding representation $\wmemrep$ is more diverse:
Although $\wmemrep$ is \emph{computably} equivalent to $\setrep$
over $\CR$-sets (\cref{fact:all-comp-equiv}),
\cref{lem:polytime-relations-dim-one} already showed that additional local
information (an inner point $a$ and an inner radius $2^{-r}$) is necessary
to reduce a $\wmemrep[1]$-name to a $\setrep[1]$-name.
The reduction itself was no more than a binary search, but the applicability
was tied to dimension $1$ and therefore does not extend to dimension
$d = 2$ onward.
Nonetheless, $\wmemrep[d]$ \emph{can} be shown to be polynomial-time
reducible to $\setrep[d]$ \emdash and even to $\distrep$! \emdash in dimension
$d \geq 2$ given enough additional information, although by a very
different argument.
We start by sketching the positive result about $\wmemrep$'s relation to
$\distrep$ (extending upon \cite[Cor.~4.3.12]{GLS88}), and then show that none
of the enrichments could have been directly computed (in polynomial-time)
from a $\wmemrep$-name alone.

\begin{thm} 
	\label{s:convex-opt}
$\wmemrep[d] \enp{ar} \enp{b} \pleq^{\bcb} \distrep[d]$ in any dimension $d$.
\end{thm}

This result follows by applying arguments from \emph{Convex Optimization}:
an adaption of the \emph{Ellipsoid Method} plus a \emph{polarity argument}.
The Ellipsoid Method allows to first reduce $\wmemrep[d] \enp{ar} \enp{b}$
to an intermediate representation $\wopt[d]$,
called \emph{weak optimization} representation
\cite[WOPT: Def.~2.1.10]{GLS88}.
A $\phi \in \Baire$ is a $\wopt[d]$-name of $S \in \KR[d]$, if for every
directional (or: cost-) vector $c \in \ID^d$ and precision $n \in \IZ$,
it satisfies
\begin{enumerate}
\item $\phi(c,0^n) = \eword$ if $\cball(S,-2^{-n})$ is empty; and
\item $\phi(c,0^n) = p$ for some $p \in \cball(S,2^{-n}) \cap \ID^d$ such that
	$\trsp{c} x \leq \trsp{c} p + 2^{-n}$ holds true for all
	$x \in \cball(S,-2^{-n})$.
\end{enumerate}
In the second case we also say that $p$ is an \emph{almost optimal point}
\wrt the cost vector $c$ (cf.~\cref{fig:wopt}).
This case can moreover be reformulated by means of \emph{halfspaces}
and \emph{hyperplanes}:
Let $c$ be some real-valued vector and $\alpha \in \IR$.
Then $H^{\leq \alpha}_{c} \dfeq \st{x \in \IR^d}{\trsp{c} x \leq \alpha}$ and,
analogously, $H^{\geq \alpha}_{c}$ are halfspaces, and their intersection
constitutes the hyperplane
$H^{=\alpha}_c \dfeq H^{\leq \alpha}_c \cap H^{\geq \alpha}$.
The aforementioned second case now reads as
$\bigcup_{x \in \cball(x,2^{-n})} H^{\leq \trsp{c} x}_c \subseteq
	H^{\leq \trsp{c} p + 2^{-n}}_c$ for $p$ as above.

\begin{figure}[htb]
	\centering
  \includegraphics{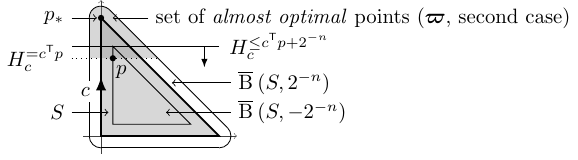}
	\caption{Weak optimization:
		Cost vector $c$, and a set $S \in \KR[2]$ along with an optimal
		solution $p_\ast$ and the set of almost optimal solutions.}
	\label{fig:wopt}
\end{figure}

\begin{rem}                                     \label{rem:wopt-properties}
Like for $\wmemrep$, representation $\wopt$ only makes sense for
regular sets since the first condition would otherwise always hold true,
\eg, for singletons.
The additional restriction to bounded sets moreover is necessary for the
existence of a point $p$ almost optimizing over $S$ in direction of $c$.
To get the ``usual'' notion of optimization in direction of $c$, first
approximate the normalized vector to $c$ (\ie, compute
$c \cdot \norm[2]{c}^{-1}$ up to the desired precision) and then apply the
$\wopt[d]$-name.

Further note that we tied representation $\wopt$ to the Euclidean norm:
The second condition in the definition of $\wopt$ is stated by means of the
the scalar product $\enc{\cdot,\cdot}$ induced by $\ndot[2]$ (\ie,
$\trsp{x} y = \enc{x,y} = 1/4 (\norm[2]{x+y}^2 - \norm[2]{x-y}^2)$).
This does not lead to the most generic definition of $\wopt$, however, it is
a sensible choice because $\ndot[2]$ is the \emph{only} norm on $\IR^d$
amongst the $p$-norms
$\norm[p]{(x_1,\dots,x_d)} \dfeq \big( \sum_{i=1}^d \abs{x_i}^p \big)^{1/p}$
(for $1 \leq p \leq \infty$)
that induces a scalar product.
We therefore only write $\wopt[d]$, but obviously mean
$\wrtn{\wopt[d]}{\ndot[2]}$
\end{rem}

Using $\wopt$ and the following \cref{s:gls-wmem-wopt}, we translate an
$\wmemrep$-name of a set $S$ to a $\distrep$-name of its \emph{polar set}
$\polar{S}$ (a related but in most instances not the same set), and then use
this as an intermediate step to prove the above \nameCref{s:convex-opt}.

\begin{fact}[{\cite[Cor.~4.3.12]{GLS88}}] \label{s:gls-wmem-wopt}
	$\wmemrep[d] \enp{ar} \enp{b} \pleq^{\bcb} \wopt[d]$.
\end{fact}

\begin{lem} \label{lem:wmem-wopt}
Define $\polar{} \colon \closedset[d] \to \closedset[d]$ through
$S \mapsto \polar{S} \dfeq
	\st[0]{y \in \IR^d}{\xall{x}{S} \trsp{y} x \leq 1}$,
and call $\polar{S}$ the \emph{polar} of $S$
(a well-known concept in convex geometry and optimization;
cf.~\cite[\secref{4.1}]{BL00}, \cite[\secref{8.2}]{BCKO08}).
Further, call $S$ \emph{centered} if $0 \in \inner{S}$.
\begin{enumerate}
\item For all $S \in \clset[d]$ and $r,b \in \IZ$,
	$\cball(0,2^{-r}) \subseteq S$ implies
	$\polar{S} \subseteq \cball(0,2^{r})$
	and $S \subseteq \cball(0,2^{b})$ implies
	$\cball(0,2^{-b}) \subseteq \polar{S}$.
\item Let $S \in \KCR[d]$ be centered.
	Then $\polar{S}$ is also contained in $\KCR[d]$, centered,
	and $S = \dpolar{S}$.
\item Let $\mathcal{Z} \dfeq
	\st[0]{S \in \KCR[d]}{S \text{ centered}\,}$.
	Further define an enrichment (``inner radius of centered sets'')
	$\ens{r_0} \dffn S \mmapsto
		\st[0]{ \uintrep(r) }{ \cball(0,2^{-r}) \subseteq S }$.
	Then $\polar{}|_{\mathcal{Z}}$ is
	$\big(
		\wmemrep[d] \enp{r_0} \enp{b}, \distrep[d] \enp{r_0} \enp{b}
	\big)$-computable in polynomial time.
\end{enumerate}
\end{lem}

\begin{proof}[of \cref{s:convex-opt}]
Let $\enc[1]{\phi,a,0^r,0^b}$ be a $\wmemrep[d] \enp{ar} \enp{b}$-name of
$S \in \mathcal{Z}$.
Then $\phi'(q,0^n) \dfeq \phi(q-a,0^n)$ gives a
$\wmemrep[d] \enp{r_0} \enp{b}$-name $\enc[1]{\phi',0^{r},0^b}$ of
$S' \dfeq \st{x-a}{x \in S}$.
Since $S'$ is centered and therefore contains $0$ as an inner point,
\cref{lem:wmem-wopt}(3) can be applied to get a
$\distrep[d] \enp{r_0} \enp{b}$-name $\enc[1]{\phi'',0^{r''},0^{b''}}$ of
$\polar{S'}$ out of $\phi'$ with $r'' \dfeq b$ and $b'' \dfeq r$.
Use the reduction $\distrep[d] \pleq \wmemrep[d]$ and apply
\cref{lem:wmem-wopt}(3) once again to get a $\distrep[d]$-name $\phi'''$
of $\dpolar{S'} = S'$ (by \cref{lem:wmem-wopt}(2)) out of $\phi''$.
The final translation $S' \mapsto \st{x+a}{x \in S'}$ through
$\psi(q,0^n) \dfeq \phi'''(q+a,0^n)$ yields a $\distrep[d]$-name $\psi$ of $S$.
\qedhere
\end{proof}

\begin{figure}
	\centering
	\begin{minipage}[t]{0.39\textwidth}
      \centering
      \includegraphics{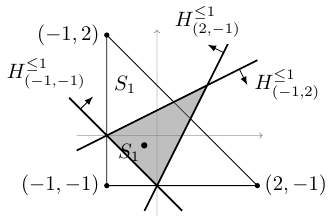}
      {\small (a) Illustrating the definition of $\polar{S}$.}
    \end{minipage}
    \qquad
    \begin{minipage}[t]{0.5\textwidth}
      \includegraphics{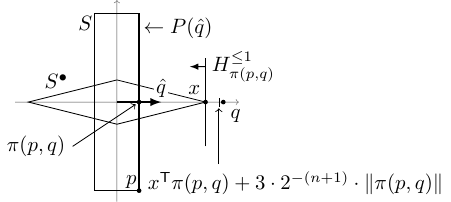}\\
      {\small (b)
		Optimization in direction of $\hat{q}$:
		Pick an arbitrary point $p \in P(\hat{q})$
		\emdash the set of optimal points \wrt $\hat{q}$ \emdash
		and project it onto $\hat{q}$.
		The resulting point $\pi(p,q)$ then allows to recover
		the distance of $q$ to $\polar{S}$ from
		$H^{=1}_{\pi(p,q)}$.
	}
    \end{minipage}
	\caption{Construction of and argumentation using polar sets.}
    \label{fig:dual-set-def-and-polar-close}
\end{figure}

The key ingredient in the proof of \cref{lem:wmem-wopt} is to take the ratio
$2^{-r}/2^b$ into account.
If done correctly, this then ensures that we get a sufficiently good
approximation of a bounding hyperplane $H^{\leq 1}_{p'}$ from which the
distance \emdash and hence a $\distrep$-name \emdash can be easily calculated.
We wrap the necessary technical details in the following statement.

\begin{prop} \label{prop:wopt-details}
Let $S \in \KCR[d]$ be centered.
Further let $r \in \IZ$ and $b \in \IZ$ be inner and outer radius parameter,
respectively.
\begin{enumerate}
\item Let $\phi$ be a $\wopt[d]$-name of $S$.
	Then $p \dfeq \phi(q,0^m)$ satisfies
	\[
		\some{p_\ast}{S}{
			p_\ast \in \cball(p,2^{-n})
			\text{ and }
			\xall{x}{S} \trsp{q}x \leq \trsp{q}p_\ast
		}
	\]
	if $m \geq n + \abs{b} + \abs{r} + 1$.
\item Denote by $\pi \dffn (p,q) \mapsto
	\trsp{p} q \cdot 1/(\trsp{q} q) \cdot q$ the projection of
	$p \in \IR^d$ onto the line spanned by $q \in \IR^d$.
	If $\pi(p,q)$ with $p \in \boundary{S}$ is approximated by
	$p' \in \ID^d$ with precision $m \geq n + \abs{b} + \abs{r} + 1$, then
	\[
		\dH\big(
			\cball(0,2^b+2^r) \cap H^{=1}_{p'},\,
			\cball(0,2^b+2^r) \cap H^{=1}_{\pi(p,q)}
		\big) \leq 2^{-n}
		\eqnsp .
	\]
	Stated differently, vectors $\pi(p,q)$ and $p'$ describe approximately the
	same hyperplane (with respect to bounds $b$ and $r$).
\end{enumerate}
\end{prop}

\begin{proof}[of \cref{lem:wmem-wopt}]
Note that the polar of a closed set is closed, too, as per definition
it is the intersection of closed halfspaces, \ie,
$\polar{S} = \bigcap_{x \in S} H^{\leq 1}_x$ with halfspaces
$H^{\leq 1}_x \dfeq \st{y \in \IR^d}{\trsp{x} y \leq 1}$.
Statement (2) now is a special case of the \emph{Bipolar Theorem}
(cf.~\cite[Thm.~4.1.5]{BL00}), while statement (1) follows by examining the
proof of the aforementioned theorem (cf.~\cite[Exercise 4.1(5)]{BL00}).

Concerning (3):
Let $\enc{\phi,0^r,0^b}$ be a
$\wmemrep[d]|^{\mathcal{Z}} \enp{r_0} \enp{b}$-name of $S \in \mathcal{Z}$,
$q \in \ID^d$ and $n \in \IN$.
Apply \cref{s:gls-wmem-wopt} to the above name to get a
$\wopt[d]|^{\mathcal{Z}} \enp{r_0} \enp{b}$-name $\enc{\phi',0^r,0^b}$
of $S$.

Notice beforehand that \cref{prop:wopt-details} allows us to describe all of
the following steps in terms of \emph{exact} computations while they actually
have to be carried out approximately.
To get the approximative (and hence correct) version, use the aforementioned
results and the closure of polynomial-time function computation under
composition.

As already mentioned in \cref{rem:wopt-properties}, optimization in a certain
direction in the usual sense is obtained from a $\wopt$-name by
first normalizing the respective cost vector; \ie,
$q' \dfeq q/\norm{q}$ in this setting.
Now take the $\wopt$-name $\phi'$ and apply it to $q'$ to obtain an
optimal point $p \dfeq \phi'(q',0^m)$.
The point $p$ itself usually does not describe the distance between $q$ and
$\polar{S}$ appropriately as depicted in \cref{fig:dual-set-def-and-polar-close}b, but its
projection onto $q'$ encodes precisely this information.
To this end, let $p' \dfeq \pi(p,q')$ (use \cref{prop:wopt-details}(2) to
get a good approximations) and observe that the distance of
$q$ from $\polar{S}$ can be obtained from the distance of $q$ to the
hyperplane $\st[0]{y \in \IR^d}{\trsp{y} p' = 1} = H^{=1}_{p'}$ under the
premise that $\trsp{q} p' \geq 1$.
More concretely, a valid $\distrep[d]|^{\mathcal{Z}}$-name $\varphi$ of
$\polar{S}$ can be defined as follows:
\begin{align*}
	\label{eqn:wmem-wopt-proj}
	\varphi(q,0^n) \dfeq 0
		\text{ if }
		\trsp{q} p' \leq 1 + 3 \cdot 2^{-(n+1)} \cdot \norm{p'}
	\eqnsp ; \quad
	\varphi(q,0^n)
		\dfeq (\trsp{q} p' - 1) \cdot \norm{p'}^{-1}
		\text{ otherwise}
	\eqnsp .
\end{align*}
\qedhere
\end{proof}

\begin{proof}[of \cref{prop:wopt-details}]
Considering (1):
First note that $\dH\big( S, \cball(S,-2^{-m}) \big) \leq 2^{-(n+1)}$
if $m \geq n + \abs{b} + \abs{r} + 1$ which follows by a geometric argument:
Observe that $S$ does contain a filled right-angled triangle $T$ with
adjacent side of length $\leq 2^b$ and opposite side of length $\geq 2^{-r}$.
The ratio $2^{-r}/2^b$ bounds how ``steep'' this triangle can be.
Stated differently, for all $x \in \boundary{T}$ there exists a
$y \in \clb{T}{-2^{-m}}$ with $\norm{x-y} \leq 2^{-(n+1)}$ for $m$ as above;
which implies the above statement about the Hausdorff-distance of
$\clb{S}{-2^{-m}}$ to $S$.

The definition of $\wopt$ now implies that each almost optimal
$p \in \cball(S,2^{-m})$ fulfills $\trsp{q}x \leq \trsp{q}p + 2^{-m}$
for all $x \in \cball(S,-2^{-m})$.
Combine this bound with the first argument over the Hausdorff distance to
obtain the claimed result, namely that there exists an optimal point
$p_\ast \in S \cap \cball(p,\delta)$ with
$\delta \dfeq (2^{-(n+1)} + 2 \cdot 2^{-m})/2 < 2^{-(n+1)}$ \wrt optimization
direction $q$.

Considering (2):
First note that $p \in \boundary{S}$ implies $\norm{p} \geq 2^{-r}$, and also
$\norm{p} \leq 2^{b}$ due to $S \subseteq \cball(0,2^b)$.
Without loss of generality, let $\pi(p,q) \dfeqrev (\lambda,0,\dots,0)$ and
$p' \dfeq (\lambda \pm 2^{-m},0,\dots,0)$ (the latter being a boundary case of
$p' \in \cball(\pi(p,q),2^{-m})$) with $2^{-r} \leq \lambda \leq 2^b$
(as noted before).
In this particular case $\pi(p,q)$ and $p'$ are codirectional which simplifies
the following argument.
Note that $\trsp{\pi(p,q)} x = 1$ if $x_1 = 1/\lambda$,
and $\trsp{p'}x' = 1$ if $x'_1 = 1/(\lambda+2^{-m})$.
Then the codirectionality of $\pi(p,q)$ and $p'$ imply that
$H^{=1}_{\pi(p,q)}$ and $H^{=1}_{p'}$ are parallel, and they are of Hausdorff
distance $\abs{1/\lambda - 1/(\lambda+2^{-m})}$.
By $2^{-r} \leq \lambda \leq 2^b$.
Therefore, $\abs{1/\lambda - 1/(\lambda+2^{-m})} \leq 2^{-n}$ by
$m \geq n + \abs{r} + \abs{b} + 1$, and thus implies
\[
	\dH\big(
		\cball(0,2^{b}+2^{r}) \cap H^{=1}_{p'},\,
		\cball(0,2^{b}+2^{r}) \cap H^{=1}_{\pi(p,q)}
	\big) \leq 2^{-n}
	\eqnsp .
	\vspace*{-2em}
\]
\qedhere
\end{proof}

Both the enrichments ($a$, $r$ and $b$) as well as the restriction to bounded
convex bodies $\bcb$ were necessary to make \cref{s:convex-opt} work, as
we summarize in the following statement.

\begin{prop}[enrichments of {$\wmemrep$}] \label{thm:enrichments-wmem}
For all $d \in \IN$ we get the following negative results:
\begin{enumerate}
\item The multi-valued operation
	$\mathsf{Bound} \colon \KR[d] \mto \IZ$,
	$\KR[d] \ni S \mmapsto \st{b \in \IZ}{S \subseteq \cball(0,2^b)}$
	is $(\wmemrep[d] \enp{ar}, \unaryrep)$-discontinuous.
	The analogous fact holds for $\setrep[d]$
	(cf.~\cite[Exercise 5.2.4]{Weih00}).
\item Convexity is crucial for \cref{s:convex-opt} to hold:
	$\wmemrep[d] \enp{ar} \enp{b} \not\preceq^{\KR} \distrep[d]$.
\item \label{item:two-third-advice-parameter}
	Addendum to the previous point:
	Convexity helps only in the presence of all of the above
	enrichments; \ie, there is \emph{no} machine operating on $\bcb[d]$
	that provided with $\wmemrep[d] \enp\chi_1 \enp\chi_2$ computes
	$\chi_3$ in \emph{polynomial time} for any permutation
	$\{ \chi_1,\chi_2,\chi_3 \}$ of $\{ \ens{a},\ens{r},\ens{b} \}$.
\item Over $\bcb[d]$, advice parameters $\ens{a}$, $\ens{r}$ and $\ens{b}$
	are uniformly computable from $\setrep[d]$.
	The same statement fails, however, for computability in
	\emph{polynomial time}.
\end{enumerate}
\end{prop}

\begin{proof}\leavevmode
\begin{enumerate}
\item Let $M^?$ be a hypothetical OTM to compute $\mathsf{Bound}$.
	Further, let $\enc{\phi,a,0^r}$ be a $\wmemrep[d]|^{\KR} \enp{ar}$-name
	for $S \in \KR[d]$.
	Machine $M^?$ terminates (in finite time) and produces a potential bound
	$b$.
	During its computation it can only have made finitely many queries to
	$\phi$ and thus has checked only points in, say,
	$\cball_{\normdot}(0, 2^{b'})$ for some $b' \in \IZ$.
	Therefore, $M^?$ would have produced the same potential bound $b \leq b'$
	if $S$ were replaced with the set
	$S' \dfeq S \cup \cball_{\normdot}(p, 1)$ for some point $p$ satisfying
	$\cball_{\normdot}(0, 2^{b'}) \cap \cball_{\normdot}(p,1) = \emptyset$.
\item We prove the stronger statement
	$\wmemrep[d] \enp{ar} \enp{b} \not\preceq^{\KR} \setrep[d]$.

	Let $S \dfeq \cball_{\normdot}(0, 2^b, 2^b-2^{-3})$, and be
	$\phi \dfeq \enc{\phi',a,0^r,0^b}$ with $\phi'(q',n') \dfeq \chi_{S}(q')$
	a concrete $\wmemrep[d]|^{\KR} \enp{ar} \enp{b}$-name of $S$.
	Further let $M^?$ be a hypothetical OTM translating any $\phi$ into a
	$\setrep[d]|^{\KR}$-name.
	The discrete inputs (tailor-made for the adversary argument) are
	$q \dfeq (0,\dots,0)$ and $n \dfeq 3$.
	On this input, $M^?$ does asks queries of precision at most $m \geq \abs{r}$.
	Therefore, it states ``0'' as the correct answer a
	$\setrep[d]$-name would have given on $\enc{q,0^n}$
	because of $\cball(0, 2^{-3+1}) \cap S = \emptyset$.
	$M^?$ surely produces the right answer for $S$, but it also does so on the
	slightly modified (adversary) set $S' \dfeq S \cup \cball(0, 2^{-(m+2)})$
	for all $\wmemrep[d]|^{\KR}$-names $\phi''$ for $S'$ with
	$\xall{p}{\ID^d} \xall{k}{\IN,\, k \leq m} \phi'(p,0^k)=\phi''(p,0^k)$;
	thus misleading $M^?$ to produce the wrong answer ($0$ instead of $1$).
\item Parameter $b$ can not be computed in polynomial time from $r$, $a$
	(and $n$) because the outer radius of a set $S$ is simply not bounded in
	this (local) information about $S$.
	Finding an inner point $a$ from $b$ and $r$ requires to query a
	$\wmemrep[d]|^{\KCR}$-name $\phi$ in roughly $2^{d \cdot \max\{0,b+r\}}$
	many points.
	To see why, consider the collection of adversary sets
	$\st[0]{S_{p} \dfeq \cball(p,2^{-r})}{p \in \cball(0,2^b) \cap \ID^d_{r+1}}$
	and observe that $S_p$ can only be distinguished from any other $S_{p'}$
	if $\phi$ is evaluated in $p$ \emph{and} $p'$ with precision $r+1$.
	An analogous argument shows why an inner radius parameter $r$ can not be
	bounded in terms of $a$ and $b$ only.
\item Computability of $b$, $a$ and $r$:
	Let $\phi$ be a $\wrtn{\setrep[d]}{\normdot}$-name of $S \in \KCR[d]$.
	An outer radius parameter $b$ exists since $S$ is bounded, and it is
	computable from $\phi$ by exploiting convexity:
	Starting at $0$, systematically ask queries $\phi(p,0^1)$ with
	$p \in \ID^d_1$ in order to find a value $b \geq 1$ such that
	(a) $\xsome{p}{\ID^d_1 \cap \clb(0){0}{2^{b-1}}} \phi(p,0^1) = 1$ and
	(b) $\xall{p}{\ID^d_1 \cap \clb(0){0}{2^b}} \phi(p,0^1) = 0$.
	It then follows $S \cap \clb(0){0}{2^{b-1}} \neq \emptyset$,
	and $S \subset \clb(0){0}{2^b}$ is implied by using the convexity of $S$.

	From $b$ one can find a point $a$ and also an inner radius parameter $r$
	by gradually increasing the precision:
	Starting with $r' \dfeq -b+3$, increase $r'$ until a point
	$p' \in \ID^d_{r'}$ in $\clb(0){0}{2^b}$ is found such that all
	$p \in \ID^d_{r'} \cap \clb(0){p'}{2^{-r'+3}}$ satisfy
	$\phi(p,0^{r'}) = 1$.
	Then $\clb(0){p'}{2^{-r'}} \subseteq S$ follows by convexity of $S$.
	Now choose $a \dfeq p'$ and $r \dfeq r'$.

	Non-polynomial-time computability:
	Any (deterministic) computation of $b$ and $r$ must necessarily be
	unbounded in $n$ (and, obviously, $|\phi|$), simply because the
	\emph{values} of both $b$ and $r$ are usually unbounded in $n$.
	The same is true for an inner point $a$ since it depends on (the
	unknown) inner radius parameter $r$; take $S = \cball(a,2^{-r})$,
	$a \in \ID^d_r \setminus \ID^d_{r-1}$ as an example.
\qedhere
\end{enumerate}
\end{proof}

\Cref{thm:enrichments-wmem}(2) covers, in fact,
\emph{several} constellations of enrichments of $\wmemrep$ because it asserts that,
informally, \emph{``if we can not deduce $\chi_3$ from $\wmemrep[d]$ and
two-thirds of other information ($\chi_1$ and $\chi_2$), then particularly
neither none nor one-third of it would help, too''}.

\section{Geometric operations on sets} \label{sec:operators}

By definition, both the computability and complexity of an operator is
inextricably linked to the choice of representations of elements it is based
on; examples can be found in \cite{Brattka99}, \cite[Thm.~5.1.13]{Weih00},
\cite{Ziegler02}, \cite{ziegler2004linalg} (for computability), and
\cite{ZM08} (for complexity). While the computability is pretty well-studied,
the complexity has been left
behind as, again, a result of the missing generic framework to formulate
explicit complexity bounds in.
In this section, we do our small part to shine a light on the complexity of
$\dsoch$ (finding \emph{some} point in a set), set operators $\dsocup$,
$\dsocap$ and $\mathsf{Projection}$, and basic function operators $\dsoinv$
(local inverse of a function) and $\dsoimg$.

\subsection{\texorpdfstring{{$\dsoch$}}{Choice}: Finding a point in a set}

We analyze the complexity to compute \emph{some} (multi-valued) member of a
set $S$, given only a name of $S$; \ie, the complexity of the in general
uncomputable (cf.~\cite{BG11,BBP12}) operator
$\dsoch \colon \closedset \mto \bigcup_{d\in\IN} \IR^d$,
$\closedset \ni S \mmapsto S$.
It is an interesting operator because, intuitively, at least this operator
should be (parameterized) polynomial-time computable for reasonable
representations of sets; like the operator
$\mathsf{Evaluation} \colon (f,x) \mapsto f(x)$ is in the realm of continuous
functions \cite{KawamuraCook}.

The following statement indeed proves parameterized complexity results for
$\dsoch$.
In particular, $\setrep$ enriched with $\ens{b}$ suffices, while even
more information is necessary for $\wmemrep$.

\begin{thm}[complexity of {$\dsoch$}] \ \label{thm:choice-point}
\begin{enumerate}
\item On compact sets, $\dsoch|_{\compset}$ is
	fully polynomial-time $(\setrep[d] \enp{b}, \realrep[d])$-computable.
\item 
	$\dsoch|_{\KR}$ is $(\wmemrep[d] \enp{r} \enp{b},\realrep[d])$-computable
	in time polynomial exponential in $\abs{b} + \abs{r}$.
	This bound also is sharp (\ie, no fully polynomial-time bound holds).
\end{enumerate}
\end{thm}

\begin{proof}
\begin{enumerate}
\item Let $S \in \compset[d]$, $\enc{\phi,0^b}$ be a $\setrep[d] \enp{b}$-name
	of $S$ with $b \geq 0$, and $n \in \IZ$.
	A point $q \in \ID^d$ with $\phi(q,0^n) = 1$ can then be found by the
	following iterative procedure.
	First, let $p_{0} \dfeq 0$.
	Now assume that $p_{i-1}$ for $1 \leq i \leq n+b+2$ is already given.
	Then deterministically pick one point $p_i$ out of
	$\cball(p_{i-1},2^{b-(i-1)}) \cap \ID^d_{i-b}$ with
	$\phi(p_i,0^{i-b}) = 1$.
	Then $p_{n+b+2}$ is guaranteed to be $2^{-n}$-close to $S$.
\item Let $\enc{\phi,0^r,0^b}$ be an $\wmemrep[d]|^{\KR} \enp{r} \enp{b}$-name
	of $S \in \KR[d]$.

	Upper bound:
	Perform an exhaustive search on $\ID^d_{r+1}(2^b)$. This way
	$\phi(p,0^{r+1}) = 1$ is guaranteed by $2^{-r}$ being an inner
	radius of $S$ for some point $p \in \ID^d_{r+1+n}(2^b)$. Moreover, such
	a point will be found and is close-enough (in the sense of
	representation $\realrep[d]$) to $S$.

	Sharpness:
	Consider the class of sets
	$\mathcal{B} \dfeq \st{ \cball(a,2^{-r}) }{
		a \in \ID^d_{r+n'}(2^b), n' \in \IZ } \subset \KR[d]$.
	The sharpness then is a consequence of \cref{thm:enrichments-wmem}(2),
	$\wmemrep[d] \enp{r} \enp{b} \not\pleq^{\mathcal{B}} \setrep[d]$:
	Exponentially many points $p \in \ID^d$ have to be considered in order to
	tell any of the above sets apart.
\qedhere
\end{enumerate}
\end{proof}

\subsection{Binary union}

\cite[Lem.~2.7]{ZM08} proved $\dsocup$ to be polynomial-time computable over
$\gridrep$ \wrt an output-sensitive measure of complexity:
Given two $\gridrep[d] \pequiv^{\compset} \setrep[d] \enp{b}$-names
$\phi_1$ and $\phi_2$, taking the maximum over the outer radii parameter as
well as the maximum over the answers at any point,
$\phi(q,0^n) = \max_i \phi_i(q,0^n)$, constitutes a name of the union.
Linear-time algorithms for $\setrep[d]$ and $\distrep[d]$ follow analogously.

However, the same method applied to $\wmemrep[d]$ over regular sets does
\emph{not} yield a valid $\wmemrep[d]$-name of the union.
As it turns out, $\dsocup$ is even \emph{uncomputable} over $\wmemrep[d]$.
Convexity, again, proves to be the key to render $\dsocup$ computable,
even in \emph{polynomial time}.

\begin{thm}\ %
\begin{enumerate}
\item $\dsocup|_{\KR \times \KR}$ is
	$\big((\wmemrep[d] \enp{r} \enp{b}) \times (\wmemrep[d] \enp{r} \enp{b}),
   	\wmemrep[d]\big)$-\emph{discontinuous}.
\item On $\cb[r]$, however, $\dsocup|_{\cb \times \cb}$
	becomes polynomial-time
	$(\wmemrep[d] \times \wmemrep[d], \wmemrep[d])$-computable.\footnote{%
		Keep in mind that the result may not be convex.
	}
\end{enumerate}
\end{thm}

\begin{figure}
   \centering
	\begin{minipage}[t]{0.46\textwidth}
		\includegraphics{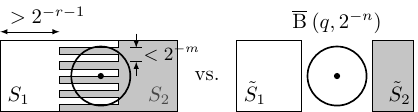}\\
		\small (a)
        $\dsocup|_{\KR \times \KR}$ is $\wmemrep[d]$-discontinuous:
        The adversary set on the left is indistinguishable from the right one
		provided the ``stripes'' are just small enough to \emph{not} contain a
		ball of radius $2^{-m}$.
      \end{minipage}
	\qquad
    \begin{minipage}[t]{0.42\textwidth}
      \includegraphics{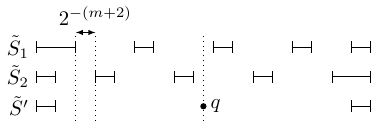}\\
      \small
      (b) Alternately cut small pieces of ``size''
		$2^{-(m+2)}$ off of two copies of the unit interval (in general:
		unit hypercube) to obtain an empty intersection around point $q$.
      \end{minipage}
	\caption{Adversary arguments, proving the discontinuity of
      (a) $\dsocup$ over $\wmemrep$ and (b) $\dsocap$ over $\setrep$.}
    \label{fig:union-wmem-disc-and-intersection-setrep-adv}
\end{figure}

\begin{proof}
\begin{enumerate}
\item The basic adversary construction of sets $S_i$ and $\tilde{S}_i$ is
	depicted in \cref{fig:union-wmem-disc-and-intersection-setrep-adv}a.
	First choose $S_i$ such that $\cball(q,2^{-n+1})$ has an empty intersection
	with $S_1 \cup S_2$ (\eg, as a simple rectangle/cuboid as depicted).
	Further construct $\tilde{S}_i$ as follows:
	\begin{enumerate*}[label=(\emph{\alph*})]
	\item $r_i$ is an inner radius parameter of $\tilde{S}_i$;
	\item the ``teeth'', being of length $> 2^{-n+2}$, are placed
		around $q$ as depicted;
	\item each rectangle/cuboid is of width $\leq 2^{-(m+1)}$,
	\end{enumerate*}
	where $m \in \IN$ marks the maximal precision a hypothetical OTM
	$M^?$ for $\dsocup$ asks on input $\enc{q,0^n}$.
	Now the only (and in this case correct) choice $M^{\enc{\phi_1,\phi_2}}$
	started with $\enc{q,0^n}$ has is to assert $0$ since
	$\cball(q,2^{-n+1}) \cap (S_1 \cup S_2) = \emptyset$.
	Now exchange the names for $S_i$ by names $\tilde{\phi}_i$ for
	$\tilde{S}_i$ which coincide with the previous ones on all queries up to
	precision $m$.
	Then $M^{\enc{\tilde{\phi}_1,\tilde{\phi}_2}}$ sill asserts $0$ in this
	case, although now $\ball(q,2^{-n}) \subset \tilde{S}_1 \cup \tilde{S}_2$
	proves $1$ to be the only correct answer.
\item Let $\phi_i$ be $\wmemrep[d]$-names of $S_i$, $i = 1,2$.
	Claim: Then $\phi'$, defined through
	\[
		\phi'(q,0^n) \dfeq \max_{i = 1,2} \st[1]{ \phi_i(p,0^{n+2}) }{
			p \in B \dfeq \cball(q, 3 \cdot 2^{-(n+2)}) \cap \ID^d_{n+2}
		}
		\eqnsp ,
	\]
	constitutes a $\wmemrep[d]$-name of $S' \dfeq S_1 \cup S_2$.

	Let $\cball(q,2^{-n}) \cap S' = \emptyset$.
	Then $\bigcup_{p \in B} \cball(p,2^{-(n+2)}) \subset S'$ must also have
	empty intersection with $S'$, hence $\phi_i(p,0^{n+2}) = 0$ for all
	$p \in B$ and $i \in \{1,2\}$.
	Now let $\cball(q,2^{-n}) \subseteq S'$.
	We prove the correctness of $\phi'(q,0^n) \dfeq 1$ by contradiction.
	To this end, assume $\phi'(q,0^n) = 0$, \ie,
	$\cball(p,2^{-(n+2)}) \not\subseteq S_1, S_2$ would have to hold for all
	$p \in B$.
	Because of $\cball(p,2^{-(n+2)}) \subset \cball(q,2^{-n}) \subseteq S'$
	and the convexity of $S_i$, there must be a $p' \in P$ such that
	$\cball(p',2^{-(n+2)})$ is contained entirely either in
	$S_1 \cap S_2$, $S_1 \setminus S_2$ or $S_2 \setminus S_1$.
	If $\cball(p',2^{-(n+2)})$ were contained in the first (convex) set, then
	we would get a contradiction because of
	$\cball(p',2^{-(n+2)}) \not\subseteq S_1,S_2$.
	If it were contained in (one of the connected regions of)
	$S_1 \setminus S_2$, then we would also get a contradiction to the
	assumption that $\cball(p',2^{-(n+2)}) \not\subseteq S_1$.
	The analogous argument also holds for the third set, thus proving
	$\wmemrep[d](\phi') = S'$.
\qedhere
\end{enumerate}
\end{proof}

\subsection{Binary intersection}

$\dsocap$ proves to be discontinuous for $\distrep$ over the class
$\closedset$ of closed sets \cite[Ex.~5.1(14)]{Weih00} by the usual adversary
argument:
Whenever a (hypothetical) algorithm decides upon a certain point $x$ to be a
member of the intersection, we can slightly modify the original sets by
excluding points from a small neighborhood of $x$, rendering $x$ to be far
off the actual intersection and therefore leading any hypothetical OTM to
produce a wrong answer.

Even when requiring the intersection of two \emph{regular} sets to be regular
again, this discontinuity remains \cite[\secref{3}]{Ziegler02}.
We show how convexity helps to establish \emph{computability}, and
how the complexity is bounded in terms of an \emph{inner radius of the
intersection}.

\begin{thm}
Let
\begin{align*}
	\mc{D} & \dfeq \st[1]{(S_1,S_2) \in \regset[d] \times \regset[d]}{
		S_1 \cap S_2 \in \regset[d]}
	\eqnsp ; \\
	\mc{E} & \dfeq \st[1]{(S_1,S_2) \in \KCR[d] \times \KCR[d]}{
		S_1 \cap S_2 \in \KCR[d]}
	\eqnsp ,
\end{align*}
and further define the enrichment $\ens{r'}$ to encode an inner radius
parameter of the intersection of two sets, \ie,
$\ens{r'} \dffn (S_1,S_2) \in \mc{D} \mmapsto
\st[1]{\uintrep(r')}{\xsome{a}{\IR^d} \cball(a,2^{-r'}) \subseteq S_1 \cap S_2}$.
\begin{enumerate}
\item $\dsocap|_\mc{D}$ is
	$(\reptpl[d] \times \reptpl[d], \reptpl[d])$-\emph{dis}continuous for
	all representations $\reptpl$ from \cref{def:set-representations}.
\item $\dsocap|_{\mc{E}}$ is
	$(\reptpl[d] \times \reptpl[d], \reptpl[d])$-computable for all
	representations $\reptpl$ from \cref{def:set-representations}.
\item $\dsocap|_{\mc{E}}$ is
	parameterized polynomial-time
	$\big((\reptpl[d] \enp{b}) \times (\reptpl[d] \enp{b}) \enp{r'},
		\reptpl[d]\big)$-computable
	for $\reptpl[d] \dfeq \setrep[d]$, and
	even
	fully polynomial-time computable for
	$\reptpl[d] \dfeq \wmemrep[d]$.
\end{enumerate}
\end{thm}

Notice the duality in the $\wmemrep[d]$-result for $\dsocap|_\mc{E}$
compared with $\dsocup$:
While for the first a correct answer was easy to produce when the point was
not deep-enough in at least one of the two sets, it was easy to produce a
correct answer for the latter if the point resided deep in both sets.
Further note that this is the direct opposite of what holds for $\setrep$;
an indication that $\setrep$ is dual to $\wmemrep$, just like the union of
sets is the lattice-dual operation of intersection.

\begin{proof}
\begin{enumerate}
\item We only show the $\setrep[d]$-discontinuity of $\dsocap$ (the proof
	loosely follows \cite[Thm.~5.1.13]{Weih00});
	the remaining statements follow by the same construction.

	Let $S_1, S_2 \dfeq [-1,1]^d$ be the sets provided to $\dsocap$
	through $\setrep[d]$-names $\phi_1,\phi_2$.
	It follows that any $\setrep[d]$-name $\phi'$ of
	$S' \dfeq S_1 \cap S_2 = [-1,1]^d$ has to satisfy $\phi'(q,0^n) = 1$
	for all $q \in \ID^d \cap [-1,1]^d$ and $n \in \IN$.
	Assume that a name for $S'$ is computed by a hypothetical OTM $M^?$
	for $\dsocap$, and let $m \in \IN$ be the maximal precision of queries
	$M^{\enc{\phi_1,\phi_2}}$ asks when started on input $\enc{q,0^n}$
	with $q \dfeq (0,\dots,0)$.
	Now exchange $\phi_i$ by a $\setrep[d]$-name $\tilde{\phi}_i$ for
	$\tilde{S}_i$ (depicted in \cref{fig:union-wmem-disc-and-intersection-setrep-adv}b)
	which fulfills $\tilde{\phi}_i(p,0^k) = \phi_i(p,0^k)$ for all
	$p \in \ID^d$ and $k \leq m$.
	Then $M^{\enc{\tilde{\phi}_1,\tilde{\phi}_2}}\enc{q,0^n}
		= M^{\enc{\phi_1,\phi_2}}\enc{q,0^n} = 1$
	although $\tilde{S}' \dfeq \tilde{S}_1 \cap \tilde{S}_2
		= \cball_{\ndot[\infty]}(q,1,1-2^{-n+1})$ and therefore
	$\cball(q,2^{-n+1}) \cap \tilde{S}' = \emptyset$ for $n \geq 3$.
\item Apply $\reptpl[d] \preceq^{\KCR} \setrep[d] \enp{r} \enp{b}$
	(\cref{fact:all-comp-equiv} + \cref{thm:enrichments-wmem}(3))
	to the domain-side and $\setrep[d] \preceq^{\KCR} \reptpl[d]$
	(\cref{fact:all-comp-equiv}) to the codomain-side, then use statement 3.
\item The second part, \ie with $\reptpl[d] \dfeq \wmemrep[d]$, has been
	proved in \cite[p.129]{GLS88}.

	\begin{figure}[htb]
		\centering
		\includegraphics{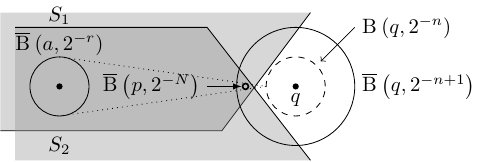}
		\caption{%
			Intersecting two \emph{convex} bodies when additional information
			about their intersection (an inner ball and an outer radius) is
			given.
		}
		\label{fig:intersection-psi-convset}
	\end{figure}

	Let $\enc[1]{\phi_1,\phi_2,0^{b_1},0^{b_2},0^{r'}}$ be a
	$( \setrep[d] \enp{b} \times \setrep[d] \enp{b} ) \enp{r'}$-name of
	$(S_1,S_2) \in \mc{E}$.
	Due to convexity, $S' \dfeq S_1 \cap S_2$ only meets $\clb{q}{2^{-n}}$ if
	$S'$ contains a polyhedron with precisely one vertex lying in $\ball(q,2^{-n})$.
	Therefore, we derive a lower bound $N \in \bigO(n + b_1 + b_2 + r')$ on
	the inner radius parameter of $S'$ close to $q$;
	\ie, a radius that guarantees the
	existence of a ball, say $\clb{p}{2^{-N}}$, which is contained in $S'$ and
	also is sufficiently close to $q$; \ie,
	$\clb(0){p}{2^{-N}} \subseteq S' \cap \clb(0){q}{2^{-n+1}}$.
	This argument is also depicted in \cref{fig:intersection-psi-convset};
	and it yields the following bound on $N$:
	\begin{align*}
		3/2 \cdot 2^{-n} \cdot \Big(
			2 \cdot 2^{-r'} \cdot (2^{\max\{b_1,b_2\}+1} - 2^{-r'} - 2^{-n})^{-1}
		\Big)
		\geq 2^{-(n + r' + \max\{b_1,b_2\})}
		\geq 2^{-N}
		\eqnsp .
	\end{align*}
	Finally, construct $\phi'$ by a local search around $q$:
	\begin{align*}
		\phi'(q, 0^n) \dfeq \max \st[2]{
				\min_{i = 1,2} \phi_i \big( p, 0^N \big)
			}{
				p \in \cball(q, 3/2 \cdot 2^{-(n+1)}) \cap \ID^d_N
			}
		\eqnsp .
	\end{align*}
	Due to the locality of this search, the number of points to be considered
	is exponential in $r$ and $\max\{b_1,b_2\}$, but polynomial in $n$.
\qedhere
\end{enumerate}
\end{proof}

\subsection{Projection operator}

For $d \in \IN$ and $d \geq e$ let $\dsoproj$ be the operator
\begin{align*}
  \dsoproj_{d,e} &\dffn \compset[d] \to \compset[e],\\
	S \mapsto \dsoproj_{d,e}(S) &\dfeq \st{x \in \IR^d}{\xsome{y}{\IR^{d-e}} (x,y) \in S}
	\eqnsp ,
\end{align*}
pointwise projecting a subset $S$ of $d$-dimensional Euclidean space down to
dimension $e$.\footnote{%
 	Note that the projection of a compact/convex/regular set is again compact/convex/regular.
}
Convexity again turns out to be the key to prove polynomial-time bounds.

\begin{fact}[{\cite[Thm.~3.2+Lem.~3.3]{ZM08}}]
	\label{thm:gridrep-projection}
\begin{enumerate}
\item Let $d \geq 2$.
	The statement ``if a set $S \in \compset[d]$ is polynomial-time
	$\gridrep[d]$-computable, then operator $\dsoproj_{d,1}(S)$ is
	polynomial-time $\gridrep[1]$-computable'' is equivalent to
	$\PTime = \NPTime$.
\item $\dsoproj_{2,1}|_{\KC}$ is polynomial-time
	$(\gridrep[2], \gridrep[1])$-computable.
\end{enumerate}
\end{fact}

The proof of the second argument can be extended to $\dsoproj_{d,d-1}$ \emdash
which by composition $\dsoproj_{e+1,e} \circ \cdots \circ \dsoproj_{d,d-1}$
implies the polynomial-time $(\gridrep[d],\gridrep[e])$-computability of
$\dsoproj_{d,e}$.
Moreover, the second result carries over to the seemingly poorest
representation $\wmemrep$, but only if restricted to bounded convex
\emph{bodies} $\KCR$.

\begin{prop}[$\wmemrep$-computability \& -complexity of $\dsoproj$]
Let $d > e \in \IN$.
\begin{enumerate}
\item $\dsoproj|_{\KR}$ is
	$(\wmemrep[d] \enp{r} \enp{b}, \wmemrep[e])$-\emph{discontinuous}
	for $d \geq 2$.
\item Nonetheless:
	$\dsoproj|_{\KCR}$ is $(\wmemrep[d] \enp{b}, \wmemrep[e])$-computable, and
	even $(\wmemrep[d] \enp{ar} \enp{b}, \wmemrep[e])$-computable in parameterized
	polynomial time.
\end{enumerate}
\end{prop}

\begin{proof}
\begin{enumerate}
\item Apply the adversary argument we have seen several times before.
	The discontinuity then follows by cutting the unit hypercube $[0,1]^d$
	up via a ``chess-board''-like pattern.
\item Use $\wmemrep[d] \enp{ar} \enp{b} \pleq^{\KCR} \setrep[d]$
	(\cref{s:convex-opt}) and
	$\setrep[d] \enp{b} \pequiv^{\compset} \gridrep[d]$
	\cref{thm:setrep-gridrep-equiv}) on the domain side,
	$\gridrep[e] \pleq^{\compset} \setrep[e] \pleq^{\KR} \wmemrep[e]$ on
	the co-domain side, and then apply the comment on the complexity of
	$\dsoproj_{d,e}|_{\KC}$ following \cref{thm:gridrep-projection}.
\qedhere
\end{enumerate}
\end{proof}

\section{Function inversion and image computation}     \label{sec:fun-set-ops}

In this section we discuss the complexity of function inversion and image
computation; \ie, of $(f,S) \mapsto (f|_{S})^{-1}$ for the former, and
$(f,S) \mapsto f[S]$ for the latter.

Recall that, for the representations in \cref{def:set-representations},
a name would return either a bit ($\wmemrep$, $\setrep$) or a dyadic rational
($\distrep$, $\reldistrep$) and/or an integer ($\gridrep$), all bounded in binary
length by that of the query and/or parameter.
This becomes different when encoding (approximations to) arbitrary continuous
real functions $f$.
To this end, we refine the previous notion of complexity
(\cref{def:cantor-polytime,def:baire-polytime}) and measure the running time in
both the discrete argument and the length of the name encoding $f$.
This generalization is covered by \cref{def:baire-polytime}(3) and permits
to bound the complexity of the aforementioned operators.

For convenience and supported by the results from \cref{sec:comparisons}, we
formulate the following definitions and results with respect to
$\normdot \dfeq \normdot[\infty]$.

\subsection{Prerequisites}

Following \cref{def:baire-polytime} we already discussed the need to add
advice parameters in order to state the complexity of operators solely in
the coding length of their discrete arguments.
As an example we saw $\distrep|^{\compset}$ with advice parameter $b$.
This approach works for sigma-compact metric spaces, but not for the space of
continuous real functions:
According to Arzela-Ascoli, its compact subsets are parameterized by a modulus
of equicontinuity, that is, an integer sequence as opposed to a single integer.
The following definition of \emph{second-order polynomials} and
\emph{second-order polynomial time} (devised and investigated in a sequence
of papers \cite{Mehlhorn76,KapronCook96,Lambov06,KawamuraCook}) provides a solution by defining a
notion of length in both the discrete argument \emph{and} the oracle.
A recent attempt to generalize from second-order to higher-order complexity
can be found in \cite{FH13}.

\begin{defi}[second-order polynomials and complexity;
	cf.~{\cite[\secref{3.2}]{KawamuraCook}}] \ %
	\begin{enumerate}
	\item A total function $\phi \colon \Sast \to \Sast$ is
		\emph{length-monotone} if $\len{\phi(s)} \leq \len{\phi(t)}$ holds true
		whenever $\len{s} \leq \len{t}$ for $s,t \in \Sast$. We denote the set
		of length-monotone functions by $\Reg$.
	\item On $\phi \in \Reg$ define a notion of \emph{length} through
      \begin{displaymath}
		\len{\phi}(m) \dfeq \len{\phi(0^m)} =
			\max_{s \in \Sast.\,\len{s} \leq m} \len{\phi(s)}.%
		\footnote{Notice the overloading of the length-function $\len{\cdot}$:
			Depending on the context, it denotes the length of either words
			or length-monotonic functions.
			But this overloading ``behaves well'' in the sense that
			every word $s \in \Sast$ can be associated with the constant function
			$\phi_s \dffn t \mapsto s$ so that the length of $s$ coincides with
			the length of $\phi_s$, \ie, $\len{s} = \len{\phi_s}$.}
        \end{displaymath}
	\item A \emph{second-order polynomial} $P \colon (\IN \to \IN) \to (\IN \to \IN)$
		in arguments $L \colon \IN \to \IN$ and $n \in \IN$ is defined inductively:
		Every constant $m \in \IN$ is a second-order polynomial, as well as
		variable $n$; assuming $Q$ and $Q'$ are second-order polynomials, then
		$Q+Q'$, $Q\cdot Q'$ and $L(Q)$ are, too.
	\end{enumerate}
\end{defi}

We make a few remarks why the above definitions are useful, and how they
subsume \cref{def:baire-polytime}.

\begin{rem}\ %
\begin{enumerate}
\item As by construction, the class of second-order polynomials is
	closed under addition, multiplication and composition
	(just like its first-order counterpart, $\IN[X]$).
\item \cref{def:second-order-rep} already showed how to encode multiple
	length-monotone functions $\phi_1,\phi_2$ into one,
	$\phi \dfeq \enc{\phi_1,\phi_2}$.
	This way, $\len{\phi}(m+1) = \len{\phi_1}(m) + \len{\phi_2}(m) + 1$
	for all $m \in \IN$.
\end{enumerate}
Let $\reptpl$ and $\reptpl'$ be second-order representations of sets $X$
and $X'$, respectively, and $\ens{E} \colon X \mto \Sast$ be a multi-valued
function.

\begin{enumerate}[resume]
\item As per \cref{def:second-order-rep}(3), any
	$\reptpl \times \reptpl'$-name is of form $\enc{\phi,\phi'}$ for
	$\phi,\phi' \in \Reg$, and thus itself in $\Reg$ by the previous
	point.
\item Similarly, any $\reptpl \enp{E}$-name $\phi = \enc{\phi_1,\phi_2}$
	is of length
	$\len{\phi}(m+1) = \len{\phi_1}(m) + \len{\ens{E}(\reptpl(\phi_2))} + 1$.
\end{enumerate}
As per the last two points, fully polynomial-time computability implies
computability in second-order polynomial-time.
\end{rem}

We are now ready to define representations for functions and lengths of names
thereof.

\begin{defi}[moduli, and representations for total functions]
Let $X \in \compset[d]$, and $f \colon X \to \IR^{e}$ be a continuous function.
\begin{enumerate}
\item A function $\modcont \colon \IN \to \IN$ is called \emph{modulus
	of (uniform) continuity of $f$} if $\norm{x - y} \leq 2^{-\modcont(n)}$
	implies $\norm{f(x) - f(y)} \leq 2^{-n}$ for all $x,y \in \dom(f)$ and
	precisions $n \in \IN$.\footnote{%
		The concepts and arguments in this section generalize to integer
		parameters.
	}
\item A $\tuple{\phi,\varphi} \in \Reg$ is a $\funrep[d,e][X]$-name of $f$ if
	\begin{enumerate}[label=(\emph{\alph*})]
	\item $\phi$ satisfies
		\begin{align} \label{eqn:cond-total}
			\xall{q}{\ID^d \cap X} \xall{n}{\IN}
			\vnorm{\phi(q,0^n) - f(q)} \leq 2^{-n}
		\end{align}
		and
	\item $\varphi$ encodes a modulus of continuity $\modcont$ of $f$, \ie,
		$\varphi \dffn s \in \Sast \to 0^{\modcont(\len{s})}$.
	\end{enumerate}
	In order to simplify the notation we associate $\varphi$ with $\modcont$
	and just write $\enc{\phi,\modcont}$ instead of $\enc{\phi,\varphi}$.
\item A function $\modsu \colon \IN \to \IN$ is called \emph{modulus of
	(uniform) unicity of $f$}
	(cf.~\cite[\secref{4.1}]{Ko91};
	introduced by Kohlenbach in \cite{Kohlenbach90,Kohlenbach93}, although
	in a more general way than we need it here)
	if $\norm{f(x) - f(y)} \leq 2^{-\modsu(n)}$
	implies $\norm{x - y} \leq 2^{-n}$ for all $x,y \in \dom(f)$ and
	precisions $n \in \IN$.
\item A $\enc{\phi,\modcont,\modsu} \in \Reg$ is a
	$\invrep[d,e][X]$-name of $f$ if $\modcont$ and $\modsu$ are, respectively,
	moduli of continuity and unicity of $f$, and
	$\funrep[d,e][X] \big( \enc{\phi,\modcont} \big) = f$.
\end{enumerate}
\end{defi}

Representations $\funrep$ and $\invrep$ only cover subclasses of total
functions with \emph{a priori known domains};
thus asking about the $(\funrep[d,d][X],\funrep[d,d][Y])$-computability
and -complexity of function inversion would only make sense if they were
restricted to total injective and \emph{surjective} functions of signature
$X \to Y$ only.
Phrased differently, formulating function inversion over a class $\mathcal{F}$
of functions and \wrt $\funrep$ only makes sense in case that for all functions
$f \in \mathcal{F}$ the (a priori known) codomain matches $\img(f)$;
thus, the inverses of functions in $\mathcal{F}$ had to be total and
(more importantly) had to share the same domain.
This is too restrictive a requirement
In general, the inverse $g$ of an injective function $f \dffn X \to Y$ is a
\emph{partial} function from $Y$ to $X$;
but \cref{eqn:cond-total} does not work in case of partial functions:
Any $\phi \in \Reg$ satisfying \cref{eqn:cond-total} and associated to a
partial function $g \parcol X \to Y$ is only defined for
\emph{dyadic points in $\dom(g)$}, but $\dom(g)$ does not necessarily contain
\emph{any} dyadic point.

By relaxing on the first universal quantification in \cref{eqn:cond-total} we
obtain new representations $\parfunrep[d,e]$ and $\parinvrep[d,e][X]$
(\ie, \emph{multi-representation}; cf.~\cite{GWXu08}) which extend
$\funrep[d,e]$ and $\invrep[d,e]$, respectively, and are tailor-made for
\emph{partial functions}.
They render any name to be defined on \emph{all} dyadic inputs (not only
those from $X \cap \ID^d$), but only give good approximations
(in the usual sense) if the input is close to the domain of the respective
function (specializing \cite[Ex.~1.19(h)]{KMRZarXiv}).

\begin{defi}[representing partial functions] \label{def:parinvrep}
	Let $f \parcol \IR^d \to \IR^{e}$ be a (possibly partial) function with
	compact domain.
	A $\enc{\phi,\modcont,\modsu}$ is a $\parinvrep[d,e]$-name of $f$ if
	$\modcont$ and $\modsu$ are moduli of $f$, respectively, and $\phi$ satisfies
	\begin{equation}
		\label{eqn:cond-partial}
		\begin{aligned}
		\xall{q}{\ID^d}
		\xall{n}{\IN}
		\Big(
	&	\dom(f) \cap \cball\big(q, 2^{-\modcont(n+1)}\big) \neq \emptyset\\[-1ex]
	&	\implies
		\xsome{x}{\dom(f) \cap \cball\big(q, 2^{-\modcont(n+1)}\big)}
		\norm{\phi \big(q,0^n\big) - f(x)} \leq 2^{-(n+1)}
		\Big).
		\end{aligned}
	\end{equation}
Similarly define $\parfunrep[d,e]$ as the generalization of $\funrep[d,e]$
to continuous partial functions.
\end{defi}

Note that by the above construction, every $\parinvrep[d,e]$-name $\phi$ of
some \emph{total} function $f$ in particular is a $\invrep[d,e]$-name of $f$,
too:
For each $q \in \dom(f)$ there is an
$x \in \dom(f) \cap \cball(q,2^{-\modcont(n+1)})$ such that
$\norm{f(q) - f(x)} \leq 2^{-(n+1)}$.
Applying \labelcref{eqn:cond-partial} then yields
\[
	\vnorm{\phi(q,0^{n}) - f(q)}
	\leq \vnorm{\phi(q,0^{n}) - f(x)} + \vnorm[0]{f(x) - f(q)}
	\leq 2^{-(n+1)} + 2^{-(n+1)}
	= 2^{-n}
	\eqnsp .
\]

\subsection{Function inversion: some upper and lower bounds}
	\label{sec:dsoinv}

The $\dsoinv$ operator takes a function $f$ and a subset
$\clset \ni S \subseteq \dom(f)$, and (under the assumption
on $f$ having a local inverse on $S$) maps $(f,S)$ to the inverse of $f|_S$.
In this section we focus on the \emph{parameterized complexity} of this
operator.

While $\dsoinv$ is polynomial-time
computable for injective functions \emph{from $[0,1]$ to $\IR$}
\cite[Thm.~4.6]{Ko91}, its complexity is linked to the existence of one-way
functions from dimension two onwards \cite[Thm.~4.23+4.26]{Ko91}.
If $f$ is \emph{bi-Hölder continuous} (\ie, both $f$ and its
inverse are Hölder continuous), then $\dsoinv$ still is only
\emph{computable in exponential time}, but becomes
\emph{parameterized polynomial-time} computable for
\emph{bi-Lipschitz functions} (\cref{s:inv-lip}).
It turns out that this bound is actually the best we can achieve:
There is \emph{no} parameterized polynomial-time algorithm for $\dsoinv$ over
bi-H\"older functions that are not bi-Lipschitz assuming that \emph{one-way
permutations} exist (\cref{s:owp-bi-hoelder}; an assumption stronger than
the existence of one-way string functions underlying contemporary cryptography).

We start to formally prove the above claims by reviewing a few non-uniform
bounds on function inversion.
The first fact is a uniform reformulation of the above mentioned inversion
result, \cite[Thm.~4.6]{Ko91}, for one-dimensional functions.

\begin{fact}
$\dsoinv$ is polynomial-time
$\big( \invrep[1,1][{[0,1]}], \parinvrep[1,1][\IR] \big)$-computable.
\end{fact}

Notice the \emph{necessity} of adding an inverse modulus $\modsu$ to make this
result work.
The algorithm behind the proof is based on \emph{trisection on $[0,1]$}
\cite[Ex.~6.3.6]{Weih00}:
For a given point $q$ in the range of $f$, start with $p = 1/2$ as a candidate
for a $2^{-n}$-approximation to $f^{-1}(q)$ and use that injectivity implies
strict monotonicity for injective functions $f \colon [0,1] \to \IR$ to
determine whether to continue this binary search in $[0,p]$ or $[p,1]$.
This algorithm stops and returns $p$ when it is of precision roughly
$\modcont(\modsu(n))$. By unrolling the definitions of both $\modcont$ and
$\modsu$ one verifies that this indeed gives a $2^{-n}$-approximation to
$f^{-1}(q)$.
This approach, however, fails from dimension two on due to lack of total order.

The following two results recall known lower and upper bounds on the
complexity of non-uniform function inversion.\newpage

\begin{fact}[non-uniform bounds for function inversion;
	{\cite[Thm.~4.23+4.26]{Ko91}}] \ %
	\label{s:ko-inv}
	\begin{enumerate}
		\item If $\PTime = \NPTime$, then $f^{-1}$ is polynomial-time
			$(\realrep[2],\realrep[2])$-computable on $\range(f)$
			whenever $f \colon [0,1]^2 \to \IR^2$ is injective,
			$(\realrep[2]|^{[0,1]^2},\realrep[2])$-computable in polynomial time
			and $\modsu$ is polynomially bounded.\footnote{Note that $f$ being
				polynomial-time computable already implies $\modcont$ to be
				polynomially bounded.}
		\item If $\PTime \neq \UPTime$, then there exists an injective, polynomial-time
			$(\realrep[2],\realrep[2])$-computable function $f \colon [0,1]^2 \to [0,1]^2$
			with polynomial modulus of unicity $\modsu$ for which $f^{-1}$ is
			\emph{not} $(\realrep[2],\realrep[2])$-computable in polynomial time on
			$\dom(f^{-1}) = \range(f)$.
	\end{enumerate}
\end{fact}

The second statement has been proved using the following result that connects
the $\PTime$ vs.~$\UPTime$ question with the existence of
\emph{one-way functions}
(which we discuss thereafter).

\begin{fact}[\cite{Ko85,GS88}]
	Total one-way functions exist if and only if
	$\PTime \neq \UPTime$.
\end{fact}

Notice the emphasis on totality (and implicitly on injectivity) since there
are other types of one-way functions whose existence, in contrast, are not
always connected to just $\PTime$ vs.~$\UPTime$ \cite[Thm.~3.2]{HT03}.
An injective polynomial-time computable function
$\phi \parcol \Sast \to \Sast$ is said to be a \emph{(worst-case)
one-way function} if
\begin{enumerate*}[label=(\emph{\alph*})]
	\item some polynomial $p$ exists such that
		$\len{\phi^{-1}(s)} \leq p(\len{s})$ whenever $s \in \range(\phi)$
		(\emph{polynomial honesty}); and
	\item if no polynomial-time computable function $\psi$ satisfies
		$\psi(\phi(s')) = s'$ for all $s' \in \dom(\phi)$ (not
		\emph{polynomial-time invertible}).
\end{enumerate*}

Now we are equipped to talk about the proof of \cref{s:ko-inv}(2):
Assume $\PTime \neq \UPTime$, and let $\phi \colon \Sast \to \Sast$ be a
total one-way function.
Based on $\phi$, construct a piecewise-linear function $f$ with the
properties described in \cref{s:ko-inv}(2) which is hard to invert if $\phi$ is.
This is achieved by encoding the image of $\phi$ into the domain of $f$ in a
way which only allows to recover the inverse $s = \phi^{-1}(t)$ from
$t$ and $f$ if $\phi$ is polynomial-time computable.
The moduli (of continuity/unicity) of $f$ are, moreover, polynomials.
More precisely, $\modcont(n) = \modsu(n) = cn + p(n) + \const$, where
$p(\len{s}) = \len{\phi^{-1}(s)}$ for any $s \in \Sigma^\ast$.
Since $p(n)$ is super-logarithmic\footnote{If not, one could just
try out all of the $2^{\lb(\len{s})}$ many possible preimages for
$s$ under $\phi$, thus computing $\phi$ in polynomial time,
contradicting the existence of one-way functions (since $\phi$ is
an arbitrary one), thus implying $\PTime = \UPTime$.},
the moduli are bounded linear (from below) in $n$. This suggests that
$\dsoinv$ could be polynomial-time computable for the class of Lipschitz-
or even H\"older-continuous functions; which we prove to be almost
correct in \cref{s:inv-lip}; and it can not be generalized to
\emph{arbitrary} polynomially bounded moduli (\cref{s:ko-inv}(2)).

The inversion algorithm we devise in \cref{s:inv-lip} will involve partial
injective functions, encoded using $\invrep[d,e]$ together with $\setrep[d]$
for their domain:

\begin{defi}[representation {$\setinvrep$}]
Let $f \parcol \IR^d \to \IR^{e}$ be a (possibly partial) function with
compact domain, and let $S \in \compset[d]$.
A $\enc{\phi,\modcont,\modsu,\psi} \in \Reg$ is a
$\setinvrep[d,e][\IR^d]$-name of $(f,S)$ if
\begin{enumerate*}[label=(\emph{\alph*})]
	\item $S \subseteq \dom(f)$;
	\item $\psi$ is a $\gridrep[d]$-name of $S$;
	\item $\enc{\phi,\modcont,\modsu}$ is a $\parinvrep[d,e]$-name of $f$.
\end{enumerate*}
\end{defi}

Recall that a function $f \colon X \to Y$ on normed spaces
$(X,\ndot[X])$ and $(Y,\ndot[Y])$ is \emph{$(\alpha,C)$-Hölder (continuous)}
with \emph{Hölder exponent} $0 < \alpha \leq 1$ and
\emph{Hölder constant} $C > 0$ if it satisfies
\[
	\norm[Y]{f(x)-f(y)} \leq C\cdot\norm[X]{x-y}^\alpha
\]
for any two $x,y \in X$.

In particular, an $(\alpha,C)$-Hölder function has modulus of continuity
$\modcont(n) \dfeq (n + \lb(C)) \cdot \alpha^{-1}$, and any $L$-Lipschitz
(continuous) function is in fact $(1,L)$-Hölder.
Take $[0,1] \ni x \mapsto \sqrt{x}$ as an example:
It is $(\alpha, 1)$-Hölder for $\alpha \leq 1/2$, and its inverse
$[0,1] \ni y \mapsto y^2$ is even $(1,1)$-Hölder (hence $1$-Lipschitz).

If the inverse of a Hölder function $f$ exists and if it moreover is a Hölder
function, than we call $f$ \emph{bi-Hölder}.
If $f$ is bi-Hölder, then there exist bounds $0 < \alpha,\alpha' \leq 1$
and $C,C' > 0$ such that
\[
	1/C' \cdot \norm[X]{x-y}^{1/\alpha'}
	\leq \norm[Y]{f(x) - f(y)}
	\leq C \cdot \norm[X]{x-y}^\alpha
	\eqnsp .
\]
If $\alpha = \alpha' = 1$, then we call $f$ \emph{bi-Lipschitz}.
\footnote{%
	Exponents $\alpha \in \{ 0 \} \cup (1,\infty)$ excluded by purpose:
	$\alpha = 0$ if the respective function is bounded, and
	$\alpha > 1$ if it is constant.
}

For convenience, denote by $\hfun$ the class of partial functions
$f \parcol \IR^d \to \IR^e$ that are also bi-Hölder, and by $\lfun$
the class of those $f$ being bi-Lipschitz.
Now we are equipped to state our result about inversion.

\begin{thm}[complexity of {$\dsoinv$}]
	\label{s:inv-lip}
Operator $\dsoinv$ is $\big(\setinvrep[d,e],\parinvrep[e,d] \big)$-computable
and its time complexity
is bounded \emph{exponentially} in $\modcont \circ \modsu \circ \modcont(n)$.
This exponential dependence still holds true when restricted to $\hfun$,
but leads to a \emph{parameterized polynomial-time} bound when further
restricted to $\lfun$.
\end{thm}

The exponential dependence on Hölder parameters in the above theorem is
actually optimal unless $\PTime = \UPTime \cap \co\UPTime$.
To see why, we consider the notion of \emph{one-way permutations},
that is, bijective one-way functions.
It is known by \cite[Thm.~3.1]{HT03} that total one-way permutations exist if
and only if $\PTime \neq \UPTime \cap \co\UPTime$.
Recall that the moduli of the function specifically constructed to prove
\cref{s:ko-inv}(2) were polynomially bounded in $n$.
Assuming the existence of total one-way permutations, they even become
\emph{linear} in $n$.
Noting that $f$ is a Hölder function if and only if it has a linearly bounded
modulus then implies the claimed optimality of \cref{s:inv-lip}.

\begin{lem} \label{s:owp-lengthpreserving}
Let $\varphi$ be a total one-way permutation.
Then a \emph{partial} one-way permutation \\ $\psi \parcol \Sast \to \Sast$ with
the following properties can be constructed from $\varphi$:
\begin{enumerate}
\item $\psi$ is length-preserving, \ie,
	$\psi[\Sigma^m] \subseteq \Sigma^m$ for all $m \in \IN \eqnsp ;$
\item $\psi \in \FPTime \eqnsp ;$
\item $\psi^{-1} \in \FPTime \implies \varphi^{-1} \in \FPTime \eqnsp .$
\end{enumerate}
\end{lem}

\begin{cor} \label{s:owp-bi-hoelder}
Assume $\PTime \neq \UPTime \cap \co\UPTime$.
Then there exists an injective, polynomial-time
$(\realrep[2],\realrep[2])$-computable function with moduli of continuity and
unicity $\modcont,\modsu$ both of the form $n \mapsto an+b$ with $a,b \in \IN$
for which $f^{-1}$ is \emph{not} $(\realrep[2],\realrep[2])$-computable in
polynomial time on $\dom(f^{-1}) = \range(f)$.
\end{cor}

\Cref{s:owp-bi-hoelder} leads to the conclusion that the exponential-time
bound in \cref{s:inv-lip} for $\dsoinv$ restricted to Hölder functions
is optimal (assuming $\PTime \neq \UPTime \cap \co\UPTime$) as a continuous
function is Hölder continuous if and only if it admits a linear modulus of
continuity.

We now sketch how to prove \cref{s:inv-lip} (also illustrated in
\cref{fig:inv-sketch}) and postpone the respective proofs of the last two
statements until the end of this subsection.

\begin{figure}[htb]
	\centering
  \includegraphics{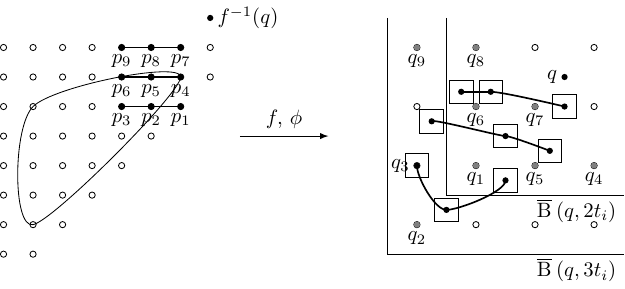}
	\caption{Points $p_j$ along with their correct and approximate images
		$f(p_j)$ and $q_j \dfeq \phi \big( p_j, 0^{m_i+1} \big)$,
		respectively. All approximate images $q_j$, except for $q_2, q_3$ and
		$q_9$, are close enough to $q$ (all that lies within the
		blue-highlighted ball), thus being candidates of being an approximate
		inverse image of $q$ in round $i+1$.}
	\label{fig:inv-sketch}
\end{figure}

Let $q$ be the point to compute
$(f|_S)^{-1}(q)$ for. Testing for all points $p$ on a fine grid, say
$\ID^d_{k(n)}$, whether their image approximate image is close to $q$ would be
a pure brute-force approach, and as such having an exponential running time.
Instead we search iteratively:
Start with a coarse grid $\ID^d_{k(0)}$ and keep
all these points from this (coarse) grid whose images are not too far from $q$.
The key idea in this step, which will lead to a low(er) complexity, is that
the number of points that have to be kept in this step can be bounded in terms
of both moduli ($\modcont$ and $\modsu$):
\begin{itemize}
	\item Any two distinct points $p,p' \in \ID^d_{\modcont(n)}$ are
		$\norm{p - p'} > 2^{-(\modcont(n) + 1)}$ apart,
	\item thus (by definition of the inverse modulus) their images satisfy
		\[
			\cball\big( f(p), 2^{-\modsu(\modcont(n)+1)-1} \big) \cap
			\cball\big( f(p'), 2^{-\modsu(\modcont(n)+1)-1} \big) = \emptyset
		\]
	\item which implies that only finitely many points from a fixed grid can
		be close to $q$ \emdash and we can bound their number in terms of
		$\modcont$, $\modsu$ and $n$.
\end{itemize}
For the next iteration, the grid will be refined to $\ID^d_{k(1)}$. But
instead of checking all these points, we will consider only those being close
to a point $p$ from the former grid $\ID^d_{k(0)}$ whose image has turned out
to be not too far from $q$. The complexity of this algorithm for finding a
good approximation to $f^{-1}(q)$ will then be of form
$\bigO(n \cdot \#\text{of points that have to be kept in each iteration})$.


\begin{proof}[of \cref{s:inv-lip}]
Let $\enc{\phi, \modcont, \modsu, \phi', 0^b}$ be a $\parinvrep[d,e]$-name
of $(f,S)$.
Further, let $n \in \IN$ and $q \in \ID^e \cap S$;
we postpone the discussion about the general case where
$\dist{f[S]}(q) \leq 2^{-\modsu(n+1)}$ to a later stage in this proof.
Without loss of generality, we assume $\modcont(n+1)-\modcont(n) \geq 1$
and $\modsu(n+1)-\modsu(n) \geq 1$.\footnote{%
	Hölder functions with Hölder exponent $\alpha \in (0,1]$
	have this property since
	$\mu(n+1)-\mu(n) = 1/\alpha \in [1,\infty)$
	for $\mu(n) = 1/\alpha \cdot (n + \lb{H})$.
}
Moreover, we prove the theorem only for $b \dfeq 0$ (just for convenience)
although the arguments extend to arbitrary outer radii parameter $b$.

To shorten the frequently used terms, we define precisions
$k_i \dfeq \modcont(\modsu(i)+1)+1$ and $m_i \dfeq \modsu(i)$, radii
$r_i \dfeq 2^{-k_i+1}$ and $t_i \dfeq 2^{-m_i}$, as well as approximations
$q_{p,i} \dfeq \phi(p,0^{k_i})$.
The proof is centered around the following sets:
\begin{align*}
	S_{0} & \dfeq \st{p \in \ID^d_{k_0}}{\phi'(p,0^{m_0}) = 1}
	\eqnsp , \\
	C_{i} & \dfeq \st[1]{p \in S_i}{
		p \in S_i \text{ and } \norm{q - q_{p,i}} \leq 2 t_i
	}
	\eqnsp , \\
	S_{i+1} & \dfeq \bigcup_{p \in C_i} S_{p,i+1}
	\eqnsp , \quad
	S_{p,i+1} \dfeq \st[1]{p' \in \cball(p,r_i) \cap \ID^d_{k_{i+1}}}{
		\phi'(p',0^{m_{i+1}}) = 1}
	\eqnsp .
\end{align*}
All we now have to do is to iteratively compute the \emph{candidate sets
$C_i$} and finally deterministically pick a point $p \in C_{n+2}$.
We claim that such a $p$ exists and that it is a $2^{-n}$-approximation to
$f^{-1}(q)$.

\emph{An important note before we continue.}
Since $f$ is a partial function the term ``$f(p)$'' might be undefined for
some $p \in S_i$.
We nonetheless want to talk about objects like ``$\cball(f(p),\cdot)$''.
The definition of $\invrep[d,e]$ solves this problem:
For $i \in \IN$ and $p \in \ID^d_{k_i}$ let $x_{p,i}$ be \emph{any}
point from $S \cap \cball(p,r_i)$ as in \labelcref{eqn:cond-partial}.
Then $f\big[\cball(p,r_i) \cap S\big] \subseteq \cball(f(x_{p,i}),t_i)$,
and we will therefore always reason about $\cball(f(x_{p,i}),\delta)$
instead of the maybe undefined $\cball(f(p),\delta/2)$.

Correctness:
We have to show that $C_i \neq \emptyset$ for all $0 \leq i \leq n+2$, and
that $f^{-1}(q) \in \cball(p,2^{-n})$ for any $p \in C_{n+2}$.
Instead of the statment ``$C_i \neq \emptyset$'' we prove the stronger
proposition ``$\xsome{p_i}{C_i} q \in \cball(f(x_{p_i,i}),t_i)$''.

For $i=0$ we first note that $\bigcup_{p \in S_0} \cball(p,r_0)$ is a
superset of $S$.
This plus the definition of $\modcont$ imply
$f[S] \subseteq \bigcup_{p \in S_0} \cball(f(x_{p,0}),t_0)$.
Therefore, there must exist a point $p_0 \in S_0$ whose image is close to $q$
in the sense that $q \in \cball(f(x_{p_0,0}),t_0)$.
Hence, $\norm{q - q_{p_0,0}} \leq 2t_0$ which gives $C_0 \neq \emptyset$.

Now let $i \geq 1$.
By construction of $C_{i-1}$ and $S_i$ it holds that
\[
	q \in
	\bigcup_{p \in C_{i-1}} f\big[ \clb{p}{r_{i-1}} \cap S \big]
	\subseteq
	\bigcup_{p \in C_{i-1}} \bigcup_{p' \in P_{p,i}}
		f\big[ \clb{p'}{r_i} \cap S \big]
	\eqnsp .
\]
Therefore the exists a $p' \in P_i$ with $\norm{q - f(x_{p',i})} \leq t_i$,
implying $\norm{q - q_{p',i}} \leq 2t_i$.
Thus, $C_i \neq \emptyset$.

In the end (\ie, for $i = n+2$), the definition of $\modsu$ implies that for
any $p \in C_{n+2}$ holds $q \in \cball(q_{p,n+2},2t_{n+2})$, which first
leads to $q \in \cball(f(x_{p,n+2}),3t_{n+2})$.
Using that $\modsu(n+2) - \modsu(n) \geq 2$ implies
$3t_{n+2} < 4t_{n+2} \leq t_n$
finally allows to conclude $f^{-1}(q) \in \cball(p,2^{-n})$.

\emph{A note on the general case of $\dist{f[S]}(q) \leq t_{n+1}$}:
By assumption, there exists an $x \in S$ such that
$\norm{f(x) - q} \leq t_{n+1}$.
Therefore, $\norm{f(x) - f(x_{p,n+2})} \leq t_{n+2}$ holds true for all
$p \in \ID^d_{k_{n+2}} \cap \clb{x}{r_{n+2}}$, implying
$\norm{f(x) - q_{p,n+2}} \leq 2t_{n+2}$.
Combining both bounds then gives
$\norm{q_{p,n+2} - q} \leq 4t_{n+2} \leq t_{n}$.
\\[.5em]
Complexity:
We have to bound the number of points in $S_{0}, C_i$ and $S_{i+1}$ for
$i \in \IN$.
The set $S_{0}$ contains at most $2^{d(b + k_0)}$ many points\footnote{%
	This is true modulo details:
	The exponential dependence on $k_0 = \modcont(\modsu(0)+1)+1$ only leads
	to an exponential dependence on the respective H\"older exponents.
	This exponential bound, however, can be reduced to linear:
	Extend the definition of $\modcont$ and $\modsu$ to integers and,
	instead of $k_0$, start with $k_{-j}$ for $j \in \IN_+$ being maximal
	with property $k_{-j} \geq 0$.
	Such a $j$ can be found in time logarithmic in the absolute value of
	both H\"older constants.
},
and $|S_{i+1}|$ is bounded by
\[
	|S_{i+1}|
	\leq \sum_{p \in C_i} \big| \cball(p,r_i) \cap \ID^d_{k_{i+1}} \big|
	\leq |C_i| \cdot (2r_i/r_{i+1})^d
	\leq |C_i| \cdot 2^{d(k_{i+1} - k_i + 2)}
	\eqnsp .
\]
The bound on $|C_i|$ requires a bit more care (as hinted prior to this proof):
Any two distinct points $p,p' \in C_i$ have the property that
$\cball(p,r_i/4)$ and $\cball(p',r_i/4)$ are disjoint.
It then follows by definition of $\modsu$ that
$\cball(x_{p,i},2^{-\modsu(k_i+2)})$ and $\cball(x_{p',i},2^{-\modsu(k_i+2)})$
are also disjoint.
This fact now allows to bound $|C_i|$ by counting how many disjoint
balls of radius $2^{-\modsu(k_i+2)}$ fit into $\clb{q}{2t_i+t_i}$:
\[
	|C_i|
	\leq (2 \cdot 3t_i/2^{-\modsu(k_i+2)})^d
	< (4 \cdot 2^{\modsu(k_i+2) - m_i})^d
	\eqnsp .
\]
The above describe procedure for computing $\dsoinv$ therefore checks at most
\begin{align*}
	\bigO\Big(
		\sum_{i=0}^{n+2}\nolimits \big|C_i\big| + \big|S_i\big|
	\Big)
\end{align*}
many points.
Their number is bounded by (and thus further simplifies to)
\begin{align}
	\label{eqn:inv-cpl-bound-general}
	\bigO\big(
		n \cdot 2^{\modsu(k_{n+2}+2)-m_{n+2}+k_{n+2}-k_{n+1}}
	\big)
	\eqnsp .
\end{align}
If $f|_S$ is bi-Hölder continuous, then its moduli are of form
$\modcont(n) = \ol{\alpha}^{-1}(n+\ol{c})$ and
$\modsu(n) = \ul{\alpha}^{-1} (n + \ul{c})$ with $\ol{c} \dfeq \lb\ol{C}$,
$\ul{c} \dfeq \lb\ul{C}$.
Moreover,
\[
	\modsu(k_{n+2}+2) - m_{n+2}
	= n \cdot ((\ol{\alpha}\ul{\alpha}^2)^{-1} - \ul{\alpha}^{-1})
		+ 2 \cdot (\ol{\alpha}\ul{\alpha}^2)^{-1}
		+ k_0/\ul{\alpha}
\]
and $k_{i+1} - k_{i} = (\ol{\alpha}\ul{\alpha})^{-1}$.
Assuming $\ol{\alpha}\ul{\alpha} = 1$ (which holds exactly for
bi-Lipschitz functions) allows to rewrite \cref{eqn:inv-cpl-bound-general}
to
$\bigO\big( n \cdot 2^{k_0} \big)$
by applying the identities we just obtained.

Note that the encoding length of each $p$ and $q_{p,i}$ is bounded linearly
in $b + k_{n+2} + \len{\enc{q}}$.
Finally, this bound combined with the former bound on the number of points
to check gives the claimed parameterized polynomial-time bound for $\dsoinv$
over $\lfun$.
\qedhere
\end{proof}

\begin{proof}[of {\cref{s:owp-bi-hoelder}}]
Follows directly from the proof of \cref{s:ko-inv}(2) by replacing the
one-way function with a partial one-way permutation as in
\cref{s:owp-lengthpreserving}.
Since $\psi$ is length-preserving it satisfies
$p(\len{s}) = \len{\psi^{-1}(s)}$ with $p \dfeq \id$.
By the remarks following \cref{s:ko-inv}(2), the moduli of the function
constructed to prove this direction are of form
$\mu(n) = cn + p(n) + \mathrm{const}$
\emdash a bound linear in $n$.
\qedhere
\end{proof}

\begin{proof}[of {\cref{s:owp-lengthpreserving}}]
Let $\varphi$ be a total one-way permutation and $p \in \IN[X]$ such that
$\len{s} \leq p(\len{\varphi(s)})$ for all $s \in \Sast$.
Set \[
	\Gamma_n \dfeq \sum\nolimits_{i=0}^n \big( p(i) + 2 \big)
	\eqnsp , \quad
	\gamma_n \dfeq \Gamma_n - \big( p(n) + 2 \big)
	\eqnsp , \quad
	\delta_{s,n} \dfeq p(n) - \len{s}
	\eqnsp ,
\]
and construct a partial function $\psi \parcol \Sast \to \Sast$ by
\[
	\psi \dffn
	0^{\gamma_n} \, 1 \, 0^{\delta_{s,n}} \, 1 \, s
	\longmapsto
	0^{\Gamma_n - (n+1)} \, 1 \, \varphi(s)
	\quad
	\text{for } \varphi(s) \in \Sigma^n
	\eqnsp .
\]
The idea behind the construction of $\psi$ is to first pad the all
arguments to $\varphi$ with length-$n$ images to be of length $\Gamma_n$,
and then to pad the image of each $t \in \Sigma^{\gamma_n}$ also to length
$\Gamma_n$.
This way, $\psi$ will be length-preserving.

Concerning (2):
Given a $t \in \Sast$, use $\len{t}$ to determine whether $t$ is contained
in $\Sigma^{\Gamma_n}$ for some $n$.
To this end, check if $t$ is of form $0^{\gamma_n}\,1\,0^{\delta_{s,n}}\,1\,s$
for some $s \in \Sigma^{\leq p(n)}$ and also if $\varphi(s) \in \Sigma^n$.
Note that the respective $n$ is bounded from above by $\len{t}$.
If $t$ is not of this particular form, then $t \not\in \dom(\psi)$ follows
immediately.
If, on the contrary, $t$ is of this form, but $\varphi(s) \not\in \Sigma^n$,
then $t \not\in \dom(\psi)$ follows, too.
If, however, $\varphi(s) \in \Sigma^n$,
then the (easy to compute) string $0^{\Gamma_n-(n+1)}\,1\,\varphi(s)$
is the image of $t$ under $\psi$.

Concerning (3):
Let $\psi^{-1} \in \FPTime$.
Given $t \in \Sast$, construct $t' \dfeq 0^{\Gamma_n}\,1\,t$.
Note that by surjectivity of $\varphi$ we know that elements of $\dom(\psi)$
can only be of the above form.
It thus suffices to compute
$s' \dfeq \psi^{-1}(t') = 0^{\gamma_n}\,1\,0^{\delta_{s,n}}\,1\,s$ and extract
$s$ from it which by construction of $\psi$ satisfies $\varphi^{-1}(t) = s$.
\qedhere
\end{proof}

\subsection{\texorpdfstring{{$\dsoimg$}}{Image}}
	\label{sec:dsoimg}

The operator $\dsoimg \dffn \usubseteq \cfn(\IR^d,\IR^e) \times \compset[d]
	\ni (f,S) \mapsto f[S] \in \compset[e]$ has been proven to be
$(\parfunrep[d,e] \times \gridrep[d], \setrep[e])$-computable
\cite[Thm.~6.2.4(4)]{Weih00} which, however, fails if we relax the restriction
on $S$ from compact to closed \cite[Thm.~6.2.4(3)]{Weih00}.
The respective proof unfortunately does not yield any complexity bounds.
However: Restricting $\dsoimg$ to Hölder functions does give parameterized
bounds.
To this end, 
define a representation $\imgfunrep[d,e]$ as follows:
A $\phi'$ is a $\imgfunrep[d,e]$-name of
$(f,S) \in \hfun \times \compset[d]$ if $\phi' = \enc{\phi,\varphi}$ with
$\parfunrep[d,e](\phi) = f$, $\gridrep[d](\varphi) = S$, and
$S \subseteq \dom(f)$.
Further denote by $\ens{\upalpha C}$ the enrichment by Hölder parameters, \ie,
\[
	\ens{\upalpha C} \dffn \hfun \mto \Sast
	\eqnsp , \quad
	\ens{\upalpha C} \dffn f \mmapsto \st
		{ \enc{ \unatrep(1/\alpha), \bnatrep(C) } }
		{ f \text{ is } (\alpha,C)\text{-Hölder continuous} }
	\eqnsp .
\]
Then the complexity of $\dsoimg$ restricted to Hölder functions follows
immediately from \cref{s:inv-lip}.

\begin{cor}
	$\dsoimg|_{\hfun \times \compset}$ is \emph{parameterized}
	polynomial-time $(\imgfunrep[d,e] \enp{\upalpha C}, \gridrep[e])$-computable,
	and $\dsoimg|_{\lfun \times \compset}$ is \emph{fully} polynomial-time
	$(\imgfunrep[d,e] \enp{\upalpha C}, \gridrep[e])$-computable.
\end{cor}

For the proof it essentially suffices to modify the proof of \cref{s:inv-lip}
as follows:
Replace all $k_i$ with $\modcont(i+1)+1$, $m_i$ with $i$, and instead of
deterministically picking a point $p \in S_{n+2}$ we check whether $S_{n+2}$
is empty.
If $S_{n+2}$ is empty, then it is a witness for $d_S(q) \geq 2^{-n}$,
while a non-empty $S_{n+2}$ witnesses $d_S(q) \leq 2^{-n+1}$.

\section{Future research}

\Cref{sec:operators,sec:fun-set-ops} can be understood as the base for further
interesting questions about operators and parameters that render them to be
polynomial-time computable; like the complexity of preimage
$\mathsf{PreImage} \colon (f,S') \mapsto f^{-1}[S']$
\cite[Lem.~24]{ziegler2004linalg}, or generalizations of the solution operator
for Poisson equations to arbitrary compact domains \cite{KSZ13}.
We also left open questions raised about the complexity of $\dsoinv$ for more
restricted classes of functions (continuous, smooth, Gevrey
\cite{LLM:Gevrey,KMRZarXiv}) and about improvements of \cref{s:ko-inv}.
For example:
Can Ko's construction be modified to produce a smooth function instead of
only a continuous one?
And do Ziegler and McNicholl's computability results on the implicit and
inverse function theorem \cite{Ziegler06,McNicholl08} (parameterized)
polynomial-time if restricted to a subset of $\cfn^2$ or Gevrey functions?

\section*{Acknowledgements}

I am grateful to Martin Ziegler for continuous advice and many helpful
suggestions; and to Akitoshi Kawamura, Ulrich Kohlenbach, Robert Rettinger,
and Florian Steinberg for seminal discussions, hints and lots of advice.
I also like to express my gratitude to the anonymous referees who have provided
invaluable suggestions on how to improve both the structure of this paper and
the presentation of results plus their respective proofs.

\end{document}